\documentclass[12pt]{amsart}
\usepackage{graphicx} 
\usepackage{amssymb}
\usepackage{amsmath}
\usepackage{amsthm}
\usepackage{ulem}
\usepackage{tikz}
\usepackage{tikz-cd}
\usepackage{xcolor}

\usepackage{color}
\definecolor{darkgreen}{rgb}{0,0.5,0}
\usepackage[colorlinks=true, linkcolor=darkgreen, urlcolor=blue, citecolor=red]{hyperref}
\usepackage{cleveref}
\usepackage{autonum}

\usetikzlibrary{arrows}
\usetikzlibrary{calc}

\usetikzlibrary{shapes.misc}

\tikzset{cross/.style={cross out, draw=black, minimum size=2*(#1-\pgflinewidth), inner sep=0pt, outer sep=0pt},
cross/.default={1pt}}

\definecolor{plum(web)}{rgb}{0.8, 0.6, 0.8}
\definecolor{plum(web)}{rgb}{0.5, 0.0, 0.5}

\newtheorem{proposition}{Proposition}
\newtheorem{lemma}{Lemma}
\newtheorem{theorem}{Theorem}
\newtheorem{definition}{Definition}
\newtheorem{remark}{Remark}

\title[Hybrid systems]{Quantum integrable systems on a classical integrable background.}

\author{Andrii Liashyk}
\address{A.L.: BIMSA, Beijing, China}
\email{a.liashyk@gmail.com}
\author{Nicolai Reshetikhin}
\address{N.R.: YMSC, Tsinghua University, Beijing, China; BIMSA, Beijing, China;
Saint Petersburg University, Saint Petersburg, Russia}
\email{reshetik@math.berkeley.edu}
\author{Ivan Sechin}
\address{I.S.: BIMSA, Beijing, China}
\email{sechin@bimsa.cn}
\begin{document}

\maketitle

\begin{abstract}
In this paper, we develop the framework for quantum integrable systems on an integrable
classical background. We call them hybrid quantum integrable systems (hybrid integrable systems),
and we show that they occur naturally in the semiclassical limit of quantum integrable systems.
We start with an outline of the concept of hybrid dynamical systems. Then, we give several examples
of hybrid integrable systems. The first series of examples is a class of hybrid integrable systems
that appear in the semiclassical limit of quantum spin chains. Then, we look at the semiclassical
limit of the quantum spin Calogero--Moser--Sutherland (CMS) system.
The result is a hybrid integrable system driven by usual classical Calogero--Moser--Sutherland dynamics.
This system at the fixed point of the multi-time classical dynamics CMS system gives the commuting spin
Hamiltonians of Haldane--Shastry model.
\end{abstract}

\tableofcontents

\section{Introduction}

The systems where quantum dynamics is mixed with the classical one were considered in physics a long
time ago. Perhaps most well-known is the Born--Oppenheimer approximation \cite{BO}, where classical
mechanics describes the motion of atoms, and the dynamics of electrons in the classical background of
these atoms is quantum. We use the term  \textbf{hybrid quantum systems}, or \textbf{hybrid systems}
for such dynamical systems.

In this paper, we formulate the general mathematical setting for hybrid systems, show how they appear
naturally in deformation families of associative algebras,
introduce the notion of hybrid integrable systems and give some examples.

One of the first examples of a hybrid integrable system is the discrete Sine-Gordon equation \cite{BBR}.
The discrete time evolution in this system is of the hybrid type: a quantum dynamics is ''driven''
by a background discrete time classical evolution.
For special, minimally periodic, classical solutions, the quantum evolution operator in this model
is equivalent to the transfer matrix of the Chiral Potts model.

After a brief outline of the general framework of hybrid quantum systems, we focus on examples of
integrable systems. One of the examples is derived from the semiclassical analysis of the quantum
spin Calogero--Moser--Sutherland (CMS) model describing $n$ quantum particles with internal degrees of freedom
(spins). The multi-time dynamics of this system in the semiclassical limit become
quantum time-dependent multi-time dynamics of the corresponding hybrid systems \cite{HW, MP}
\footnote{
This model is different from spin Calogero--Moser--Sutherland models obtained by the quantum version of
Hamiltonian reduction, see for example \cite{Reshetikhin-reduction, RStock}. These models
are related, but we will not discuss this relation here.
}.

We also show that in the semiclassical limit, integrable quantum spin chains provide examples of hybrid
integrable systems. The classical system in this case is the corresponding classical spin chain.
The quantum multitime dynamics is given by $M$-operators. We expect that multitime
hybrid evolution in such systems can be effectively studied by the semiclassical limit of
Bethe vectors in the spirit of work \cite{RS} and using Baker--Akhiezer type functions \cite{DKN}.

Here is an outline of the paper.

In section \ref{basics} we define the algebra of observables for hybrid systems and its
representations. Here we focus on matrix hybrid systems. The algebra of observables in such a
system is the algebra of sections of the bundle of matrix algebras over a symplectic
manifold with pointwise multiplication. The base of this bundle of algebras is the phase
space of the underlying classical system. This bundle is equipped with a Hermitian connection.
Thus, the algebra of observables of a hybrid system is an algebra that is finite-dimensional
over its center. The center is a Poisson algebra that acts by derivation on the whole algebra.
Algebras that are finite-dimensional over their centers are known as Azumaya algebras. So we call
algebras of observables in hybrid matrix systems Poisson Azumaya algebras.

In this section, we also describe representations of hybrid algebras of observables,
derivations of such algebras, homomorphisms, and the relation to the deformation quantization.

In section \ref{sec:hybrid-states}, we describe hybrid states, hybrid pure states, and
Lagrangian states. Pure Lagrangian states and Lagrangian representations of hybrid algebras
of observables appear naturally in matrix Schr\"odinger equations.

In section \ref{evolution}, we focus on hybrid evolution. A hybrid dynamical system is
described by two Hamiltonians: the classical Hamiltonian $H^{(0)}$ which defines the underlying
classical dynamics, and the quantum Hamiltonian $H^{(1)}$ which defines the quantum evolution
in the fibers. We also describe the evolution of states and how the evolution in
observables is related to deformation quantization.

The notion of a hybrid integrable system is introduced in section \ref{sec:hybrid-integrable-systems}.
Here we introduce hybrid multitime integrable dynamics and show that it appears naturally in the
semiclassical limit of quantum integrable systems.

In section \ref{Matrix} we describe the semiclassical asymptotic of matrix Schr\"odinger
operators and show how hybrid dynamics appear naturally in this context.

Hybrid integrable systems related to integrable spin chains are described in section
\ref{spin-chains}.

In section \ref{CM} we describe the hybrid system that emerges in the semiclassical limit
of spin Calogero--Moser--Sutherland (CMS) systems. In this case, the classical background is the
usual (spinless) CMS system. The quantum part of this system can be called dynamical
Haldane--Shastry system. Indeed, we show that the multitime flow in the CMS model has a fixed
point. It is also known in the literature as the freezing point. Quantum Hamiltonians at this
point commute and coincide with commuting Hamiltonians for the Haldane--Shastry model of
long-range interactions \cite{Haldane, Shastry}. The fixed point is a zero-dimensional
Liouville tori. In \cite{LMRS} we describe all low-dimensional degenerations in the CMS model.
Corresponding hybrid dynamics is a dynamical version of the Haldane--Shastry model.

To conclude the introduction, let us make a notational clarification. When we write $C(\mathcal{M})$
where $\mathcal{M}$ is a smooth manifold, we mean $C^\infty$-functions. When $\mathcal{M}$
is an affine algebraic variety, $C(\mathcal{M})$ is the algebra of polynomial functions
on $\mathcal{M}$.

The results of this paper were presented at a number of conferences. The earliest one was a
talk at the conference ''Integrable Systems and Field Theory'', Jussieu, Paris, October 2023.
The authors are grateful to S. Dobrokhotov, D. Freed, L. Feher, S. Gukov, A. Kapustin, A. Mikhailov,
G. Papayanov, and P. Wiegmann for discussions and valuable comments.
We are also grateful to A. Mikhailov for pointing out the reference \cite{MV}.
Our special thanks to an anonymous reviewer for multiple useful remarks.

The research of A.L. was supported by the Beijing Natural Science Foundation (IS24006).
The research of N.R. was supported by the Collaboration Grant ”Categorical Symmetries” from
the Simons Foundation, by the Changjiang fund, and by the project 075-15-2024-631 funded
by the Ministry of Science and Higher Education of the Russian Federation.

\newpage
\section{Hybrid algebra of observables and its representations} \label{basics}

\subsection{Hybrid algebra of observables}\label{sec: hybrid algebra}
We start with an example of a hybrid algebra. Let $(\mathcal{M}, \omega)$ be a symplectic manifold.
Think of it as the phase space of a classical Hamiltonian system.
We want to define a hybrid system, i.e., a quantum system on the background of this classical system.

The underlying structure in a hybrid system is the \textbf{bundle of observables}. It is a vector
bundle over a symplectic manifold $(\mathcal{M}, \omega)$
\begin{equation}
    \begin{tikzcd}
        E \arrow["\pi", d] & \arrow[l] A_x \\
        \mathcal{M} &
    \end{tikzcd}
\end{equation}
Here fibers $A_x = \pi^{-1}(x)$ are $*$-algebras.
Here by $*$-algebra we mean an associative unital algebra over $\mathbb{C}$ with a
$\mathbb{C}$-antilinear involution $* \colon A \to A, \ a \mapsto a^*$ such that
$(a b)^* = b^* a^*, \ (a^*)^* = a, \ (\lambda a)^* = \bar{\lambda} a^*$,
where $a, b \in A$ and $\lambda \in \mathbb{C}$.
We also require that $E$ is equipped with a connection $\alpha$, compatible with $*$-structure.

The space of smooth sections $A = \Gamma(\mathcal{M}, E)$ has a natural pointwise multiplication
\begin{equation}
    (s_1 s_2)_x = (s_1)_x (s_2)_x.
\end{equation}
The identity $\mathbf{1}$ in this algebra is the section $\mathbf{1} \colon x \mapsto (1_x, x)$,
where $1_x$ is the identity in $A_x$. The center of $A$ is
$Z(A) = C(\mathcal{M}) \cdot \mathbf{1}$, the subalgebra of sections of the form
\begin{equation}
    s(x) = f(x) \cdot 1_x,  \quad f(x) \in C(\mathcal{M}).
\end{equation}
We will identify $Z(A)$ with $C(\mathcal{M})$, the space of smooth functions on $\mathcal{M}$.
It has a natural Poisson structure
\begin{equation}
    \label{Poisson-brackets-center}
    \{z_1, z_2\} = \omega^{-1}(dz_1 \wedge dz_2), \quad z_1, z_2 \in Z(A).
\end{equation}
It also acts by derivations on $A$
\begin{equation}
    \label{Poisson-module-structure}
    \{z, s\} = \omega^{-1}(dz \wedge d_\alpha s), \quad z \in Z(A), \ s \in A.
\end{equation}
Here $d_\alpha$ is the de Rham differential twisted by $\alpha$.
It can be rewritten as
\begin{equation}
    \label{eq: derivation-Z-A}
    \left\{ z, s \right\} = \iota_{v(z)}\,  d_{\alpha} s, \quad z \in Z(A), \ s \in A,
\end{equation}
where $v(z)$ is a Hamiltonian vector field for $z \in C(\mathcal{M})$.

Let $E|_U \simeq U \times A_{x_0}$ be a local trivialization of $E$ over an open
neighborhood $U \subset \mathcal{M}$ of the arbitrary point $x_0 \in \mathcal{M}$. 
Then the connection \(\alpha\) is locally represented via Hermitian one-form
\(a^{(\alpha)} \in \Omega^1(U, E|_U)\)
\begin{equation}
    d_\alpha s = ds + i \, [a^{(\alpha)}, s],
\end{equation}
where $d$ is the de Rham differential.
In local coordinates, \(a^{(\alpha)} = a^{(\alpha)}_j dx^j, \ (a^{(\alpha)}_j)^* = a^{(\alpha)}_j\), and
\begin{equation}
    \{z, s\} = (\omega^{-1})^{jk} \partial_j z \partial_k s +
        i (\omega^{-1})^{jk} \partial_j z [a^{(\alpha)}_k, s].
\end{equation}
It is easy to check that \((\{z, s\})^* = \{z, s^*\}\) for all \(z \in C(\mathcal{M})\) and
\(s \in \Gamma(\mathcal{M}, E)\).

The independence of (\ref{Poisson-module-structure}) on the trivialization of $E$ is easy to check.
Two trivializations $E|_U \simeq A_{x_0} \times U$ are related by a gauge transformation
\begin{equation}
    a^{(\alpha)} \mapsto g^{-1} a^{(\alpha)} g - i g^{-1} dg, \quad
    s \mapsto g^{-1} s g.
\end{equation}
Because $d_\alpha$ is gauge invariant
\begin{gather}
    d_\alpha s \mapsto
        d(g^{-1} s g) + i \, [g^{-1} a^{(\alpha)} g, g^{-1} s g] + [g^{-1} dg, g^{-1} s g] = \\ =
        -g^{-1} dg g^{-1} s g + g^{-1} s g g^{-1} dg + i g^{-1} [a^{(\alpha)}, s] g +
        [g^{-1} dg, g^{-1} s g] + g^{-1} ds g = g^{-1} (d_\alpha s) g,
\end{gather}
the bracket (\ref{Poisson-module-structure}) is gauge invariant, i.e., globally defined.

Thus, we have defined an action of $Z(A)$ on $A$ by derivations
\begin{equation} \label{eq: Z(A)-A-Leibniz-second}
    \{z, s_1 s_2\} = \{z, s_1\} s_2 + s_1 \{z, s_2\}, \quad
        z \in C(\mathcal{M}), \ s_1, s_2 \in \Gamma(\mathcal{M}, E).
\end{equation}
It is also easy to check that
\begin{equation} \label{eq: Z(A)-A-Leibniz-first}
    \{z_1 z_2, s\} = z_1 \{z_2, s\} + z_2 \{z_1, s\}, \quad
        z_1, z_2 \in C(\mathcal{M}), \ s \in \Gamma(\mathcal{M}, E).
\end{equation}

However, $A$ is not a Poisson module
\footnote{
Recall that $V$ is a Poisson module over a Poisson algebra $P$ if $V$ is a module over a
commutative algebra $P$ endowed with a bilinear map $P \times V \to V, \ (p, v) \mapsto \{p, v\}$
such that for $p, \tilde{p} \in P, \ v \in V$
\begin{equation}
    \{\{p, \tilde{p}\}, v\} = \{p, \{\tilde{p}, v\}\} - \{\tilde{p}, \{p, v\}\}, \quad
    \{p, \tilde{p}\} \cdot v = p \cdot \{\tilde{p}, v\} - \{\tilde{p}, p \cdot v\}.
\end{equation}
}
over $Z(A)$, because it is not a module over the Lie algebra $Z(A)$ 
with Lie bracket induced by the Poisson structure on $Z(A)$. Instead, for all
$z_1, z_2 \in C(\mathcal{M})$ and $s \in \Gamma(\mathcal{M}, E)$ we have
\begin{equation} \label{eq: not-a-module-condition}
    \{\{z_1, z_2\}, s\} =
        \{z_1, \{z_2, s\}\} - \{z_2, \{z_1, s\}\} + i \, [\{z_1, z_2\}_2, s],
\end{equation}
here $\{z_1, z_2\}_2 = \iota_{v(z_1) \wedge v(z_2)} F_\alpha \in \Gamma(\mathcal{M}, E)$,
where $F_\alpha \in \Omega^2(\mathcal{M}, E)$ is the curvature form of the
connection $\alpha$: \(d_\alpha^2 s = i [F_\alpha, s]\). 
Note that \((\{z_1, z_2\}_2)^* = \{z_1, z_2\}_2\) for every \(z_1, z_2 \in Z(A)\).

The Bianchi identity for $F_\alpha$ implies
\begin{equation} \label{eq: probably_cocycle_condition}
    \left\{ x, \left\{ y, z \right\}_2 \right\} +
        \left\{ y, \left\{ z, x \right\}_2 \right\} +
            \left\{ z, \left\{ x, y \right\}_2 \right\} +
    \left\{ x, \left\{ y, z \right\} \right\}_2 +
        \left\{ y, \left\{ z, x \right\} \right\}_2 +
            \left\{ z, \left\{ x, y \right\} \right\}_2 = 0.
\end{equation}

Conditions \eqref{eq: not-a-module-condition} and \eqref{eq: probably_cocycle_condition}
allows to construct a non-abelian extension of \(Z(A)\) considered as a Lie algebra
(see, e.g., a review \cite{AMR} and references therein).
As a vector space, this extension is a direct sum
\begin{equation} \label{eq: P(A)}
    \mathcal{P}(A) = Z(A) \oplus A/Z(A),
\end{equation}
and the bracket is given by
\begin{equation} \label{eq: lie-bracket-P(A)}
    \{(z_1, \overline{a_1}), (z_2, \overline{a_2})\} = 
        (\{z_1, z_2\}, 
            \overline{\{z_1, a_2\} - \{z_2, a_1\} + i \, [a_1, a_2] - \{z_1, z_2\}_2}), 
\end{equation}
where \(\overline{a}\) is a class of an element \(a \in A\) in the factor \(A/Z(A)\).
It is easy to check that the bracket \eqref{eq: lie-bracket-P(A)} does not depend on
the choice of a representative in the class. The Jacobi identity for \eqref{eq: lie-bracket-P(A)}
follows directly from Jacobi identity for Poisson bracket on \(Z(A)\) and properties
\eqref{eq: not-a-module-condition} and \eqref{eq: probably_cocycle_condition}, 
and does not depend on the particular form of the bracket.

At \cite{MV}, Mikhailov and Vanhaecke introduced the Poisson algebra structure on 
the Lie algebra \(\mathcal{P}(A)\), extending the commutative algebra structure on \(Z(A)\).
In our example, \(Z(A) = C(\mathcal{M})\) and \(A = \Gamma(\mathcal{M}, E)\)
one can define the commutative multiplication on \(\mathcal{P}(\Gamma(\mathcal{M}, E))\)
\begin{equation} \label{eq: multiplication-P(Gamma)}
    (z_1, \overline{a_1}) \cdot (z_2, \overline{a_2}) = (z_1 z_2, \overline{z_1 a_2 + z_2 a_1}).
\end{equation}

\begin{proposition}
A Lie algebra \((\mathcal{P}(\Gamma(\mathcal{M}, E)), \{\cdot, \cdot\})\) endowed with the multiplication
\eqref{eq: multiplication-P(Gamma)} is a Poisson algebra.
\end{proposition}

\begin{proof}
We need to check that the Poisson bracket \eqref{eq: lie-bracket-P(A)} and the multiplication
\eqref{eq: multiplication-P(Gamma)} satisfy Leibniz rule, i.e. for every
\((z_1, \overline{a_1}), (z_2, \overline{a_2}), (z_3, \overline{a_3}) \in 
\mathcal{P}(\Gamma(\mathcal{M}, E))\)
\begin{equation} \label{eq: Leibniz-rule-P(A)}
    \{(z_1, \overline{a_1}) \cdot (z_2, \overline{a_2}), (z_3, \overline{a_3}) \} =
    (z_1, \overline{a_1}) \cdot \{(z_2, \overline{a_2}), (z_3, \overline{a_3})\} +
    (z_2, \overline{a_2}) \cdot \{(z_1, \overline{a_1}), (z_3, \overline{a_3})\}.
\end{equation}

This follows directly from the Leibniz rule for \(Z(A)\), derivation properties 
\eqref{eq: Z(A)-A-Leibniz-second} and \eqref{eq: Z(A)-A-Leibniz-first}, and the fact
that the curvature form satisfies
\begin{align}
    \{z_1 z_2, z_3\}_2 &= 
        F_\alpha(v(z_1 z_2), v(z_3)) = F_\alpha(z_1 v(z_2) + z_2 v(z_1), v(z_3)) = \\ &=
        z_1 F_\alpha(v(z_2), v(z_3)) + z_2 F_\alpha(v(z_1), v(z_3)) = 
        z_1 \{z_2, z_3\}_2 + z_2 \{z_1, z_3\}_2.
\end{align}
\end{proof}

Although $A$ is not a Poisson module over $Z(A)$, conditions \eqref{eq: not-a-module-condition}
and \eqref{eq: probably_cocycle_condition} guarantee that it is a Poisson module over the extended
Poisson algebra \(\mathcal{P}(A)\). 

Now let us give the general definition of a hybrid algebra.

\begin{definition} \label{def: hybrid-algebra}
Define a \textbf{hybrid algebra} as an associative algebra $A$ such that
\begin{enumerate}
    \item The center $Z(A)$ is a Poisson algebra with Poisson brackets
        $\{\cdot, \cdot\} \colon Z(A) \times Z(A) \to Z(A)$.
    \item There exists a skew-symmetric bilinear operation
        $\{\cdot, \cdot\}_2 \colon Z(A) \times Z(A) \to A$
        such that
    $$
        \left\{ x, \left\{ y, z \right\}_2 \right\} +
            \left\{ y, \left\{ z, x \right\}_2 \right\} +
                \left\{ z, \left\{ x, y \right\}_2 \right\} +
        \left\{ x, \left\{ y, z \right\} \right\}_2 +
            \left\{ y, \left\{ z, x \right\} \right\}_2 +
                \left\{ z, \left\{ x, y \right\} \right\}_2 = 0
    $$
    for all $x, y, z \in Z(A)$.
    \item $Z(A)$ acts on $A$ by derivations $\{\cdot, \cdot\} \colon Z(A) \times A \to A$:
    $$
        \{z, ab\} = \{z, a\}b + a\{z, b\},
    $$
    for all $z \in Z(A), \ a, b \in A$, and this action is connected with
    the Poisson algebra structure on \(Z(A)\) as
    \begin{gather}
        \{z w, a\} = z \{w, a\} + w \{z, a\}, \\
        \{\{z, w\}, a\} = \{z, \{w, a\}\} - \{w, \{z, a\}\} + i \, [\{z, w\}_2, a],
    \end{gather}
    for all $z, w \in Z(A), \ \forall a \in A$.
    \item The extension \(\mathcal{P}(A) = Z(A) \oplus A/Z(A)\) is a Poisson algebra
    with the Poisson bracket
    \begin{equation}
        \{(z_1, \overline{a_1}), (z_2, \overline{a_2})\} =
            \big(\{z_1, z_2\}, 
                \overline{\{z_1, a_2\} - \{z_2, a_1\} + i \, [a_1, a_2] - \{z_1, z_2\}_2}
            \big)
    \end{equation}
    and a commutative multiplication, satisfying Leibnitz rule
    \begin{equation}
        \{(z_1, \overline{a_1}) \cdot (z_2, \overline{a_2}), (z_3, \overline{a_3}) \} =
        (z_1, \overline{a_1}) \cdot \{(z_2, \overline{a_2}), (z_3, \overline{a_3})\} +
        (z_2, \overline{a_2}) \cdot \{(z_1, \overline{a_1}), (z_3, \overline{a_3})\}.
    \end{equation}
    for all \((z_1, \overline{a_1}), (z_2, \overline{a_2}), (z_3, \overline{a_3}) \in \mathcal{P}(A)\).
\end{enumerate}
\end{definition}

If for all $z, w \in Z(A) \ \{z, w\}_2 \in Z(A)$, we call the corresponding
hybrid algebra $A$ \textbf{flat}.
In this case, $A$ is a Poisson module over $Z(A)$
$$
    \{\{z, w\}, a\} = \{z, \{w, a\}\} - \{w, \{z, a\}\}.
$$

If the algebra $A$ is finite-dimensional and simple over its center, it is called Azumaya algebra.
If $A$ also has a hybrid structure, we call it \textbf{Poisson Azumaya algebra}.
In this paper, we mainly consider flat Poisson Azumaya algebras. There are many non-flat examples
related to quantum groups at a root of unity \cite{DCK} and related integrable systems,
for example, \cite{BBR}. Quantized universal enveloping algebras of affine Kac-Moody algebras
at a root of unity \cite{BK, DHR} and at the critical value of central extension \cite{FrR}
provide examples of hybrid algebras of non-Azumaya type (infinite-dimensional algebra over
its center).

Clearly, $A = \Gamma(\mathcal{M}, E)$ is a Poisson Azumaya algebra
\footnote{
    In general, the base $\mathcal{M}$ of the vector bundle $E$
    does not have to be a symplectic manifold; it can have a degenerate Poisson structure.
    Also, $E$ can be a sheaf of algebras as it happens in quantum groups at roots of unity
    \cite{DCK}.
}, and if the connection $\alpha$ is projectively flat, $\{z_1, z_2\}_2 \in Z(A)$,
$A$ is a flat Poisson Azumaya algebra.

\begin{remark} \label{rem: two-connections}
Note that if $\alpha$ is a connection on $E$, compatible with the \(*\)-structure, 
and $\lambda$ is $E$-valued Hermitian one-form on $\mathcal{M}$, 
then $\tilde{\alpha} = \alpha + i \lambda$ is also a connection on $E$, compatible with the \(*\)-structure.
Thus, we can define another hybrid algebra structure on $A = \Gamma(\mathcal{M}, E)$
using the connection $\tilde{\alpha}$. Poisson algebra structure on $Z(A) = C(\mathcal{M})$ does not depend
on the choice of connection, however, the action of $Z(A)$ on $A$ 
\eqref{eq: derivation-Z-A} shifts on the inner derivation
\begin{equation}
    \{z, s\}^{(\tilde{\alpha})} = \{z, s\}^{(\alpha)} + i \, [\eta(z), s],
\end{equation}
where $\eta \colon Z(A) \to A$ is the linear map $\eta(z) = \iota_{v(z)} \lambda$.

The curvature form changes as follows
\begin{equation}
    \{z, w\}_2^{(\tilde{\alpha})} = \{z, w\}_2^{(\alpha)} -
        \{z, \eta(w)\}^{(\alpha)} + \, \{w, \eta(z)\}^{(\alpha)} -
            i \, [\eta(z), \eta(w)] +\eta(\{z, w\}).
\end{equation}
\end{remark}

\subsection{Representation of a hybrid algebra of observables}
\label{sec: representations-hybrid-module}

Let $V$ be a Hermitian vector bundle
\begin{equation}
\begin{tikzcd}
    V \arrow[d] & \arrow[l] V_x \\
    \mathcal{M} &
\end{tikzcd}
\end{equation}
with a fiberwise module structure over $E$, i.e. for each $x \in \mathcal{M}$ we have a homomorphism
of algebras
\begin{equation}
    \rho_x \colon A_x \to \mathrm{End}(V_x).
\end{equation}
The space of sections of $V$, $\mathcal{H} = \Gamma(\mathcal{M}, V)$ has a natural structure
of an $A = \Gamma(\mathcal{M}, E)$-module with $\rho \colon A \to \mathrm{End}(\mathcal{H})$
\begin{equation}
    (\rho(s) v)_x = \rho_x(s_x) v_x.
\end{equation}
This is a $*$-representation of $A$ if
\begin{equation}
    \rho_x(s_x^*) = \rho_x(s_x)^+,
\end{equation}
where $a^+ \colon V_x \to V_x$ is Hermitian conjugate to an operator $a \colon V_x \to V_x$
\begin{equation}
    (a v_x, w_x)_x = (v_x, a^+ w_x)_x, \qquad v_x, w_x \in V_x,
\end{equation}
and $(\cdot, \cdot)_x$ is the Hermitian structure on $V_x$.

We also assume that $V$ has a connection $\beta$, which is compatible with the
connection $\alpha$, i.e.
\begin{equation}
    d_\beta (s v) = (d_\alpha s) v + s d_\beta v.
\end{equation}
Then $C(\mathcal{M})$ acts on $\mathcal{H} = \Gamma(\mathcal{M}, V)$ as
\begin{equation}
    \{z, v\} = \omega^{-1}(dz \wedge d_\beta v),
\end{equation}
here $z \in C(\mathcal{M})$, $v \in \mathcal{H}$.
Compatibility of connections $\alpha$ and $\beta$ gives
\begin{equation}\label{eq:PAmodule_LeiblinzJacobi}
    \{z, s v\} = \{z, s\} v + s \{z, v\},  \qquad
    \{\{z_1, z_2\}, v\} = \{z_1, \{z_2, v\}\} - \{z_2, \{z_1, v\}\} + i \, \{z_1, z_2\}_2 v.
\end{equation}

\begin{definition}
Let $A$ be a hybrid Poisson algebra, and $\mathcal{H}$ be a module over $A$ as 
an associative algebra. The module $\mathcal{H}$ is called a \textbf{hybrid module} 
if the center $Z(A)$ acts on $\mathcal{H}$ 
$\{\cdot, \cdot\} \colon Z(A) \times \mathcal{H} \to \mathcal{H}$ 
compatibly with $A$-module structure
\begin{equation} \label{eq: hybrid-module-Leibniz}
      \{z, a v\} = \{z, a\}v + a \{z, v\}
\end{equation}
for all $z \in Z(A), \ a \in A, \ v \in \mathcal{H}$ with the additional condition
\begin{equation} \label{eq: hybrid-module-Jacobi}
    \{\{z, w\}, v\} = \{z, \{w, v\}\} - \{w, \{z, v\}\} + i \, \{z, w\}_2 \, v
\end{equation}
for all $z, w \in Z(A)$ and $v \in \mathcal{H}$.
\end{definition}

Note that gauge transformations $v \mapsto e^{i\theta} v$, where $\theta$ is a scalar, 
change $\beta$ as $\beta \mapsto \beta + i \, d\theta$. If $\beta$ is compatible with $\alpha$,
clearly $\beta + i \, d\theta$ is also compatible with $\alpha$. Thus, we have a natural
"gauge group" $G_{\mathcal{M}} = \mathrm{Maps}(\mathcal{M}, U(1))$ acting on a hybrid module.

\begin{remark}
If for every $z, w \in Z(A) \ \{z, w\}_2 = 0$, the hybrid module $\mathcal{H}$ is a Poisson module
over $Z(A)$
$$
    \{\{z, w\}, v\} = \{z, \{w, v\}\} - \{w, \{z, v\}\}.
$$
For $A = \Gamma(\mathcal{M}, E)$, this means that the connection $\alpha$ is flat.
It is possible only if the first Chern class $c_1(E) = 0$. Let $A = \Gamma(\mathcal{M}, E)$
be a flat Poisson Azumaya algebra with projectively flat connection $\tilde{\alpha}$ 
represented in a hybrid module $\mathcal{H} = \Gamma(\mathcal{M}, V)$. Assuming that the
rank of vector bundle $V$ is finite $\mathrm{rk}(V) = N$, and $c_1(E) = 0$, we
can always make the connection $\tilde{\alpha}$ flat by adding a one-form. The
curvature 2-form $F_{\tilde{\alpha}} \in \Omega^2(\mathcal{M}, E)$ in this case is 
represented as an $N \times N$ matrix of 2-forms. The vanishing of the first Chern class
$c_1(E) = 0$ implies that the matrix trace of the curvature 2-form 
$F_{\tilde{\alpha}} \in \Omega^2(\mathcal{M}, E)$ is exact
\begin{equation}
    \mathrm{Tr}(F_{\tilde{\alpha}}) = df,
\end{equation}
where $f \in \Omega^1(\mathcal{M})$. Since $\tilde{\alpha}$ is projectively flat, 
\begin{equation}
    F_{\tilde{\alpha}} = \tfrac{1}{N} df \cdot \mathbf{1},
\end{equation}
where $\mathbf{1}$ is the identity $N \times N$ matrix. Thus, the connection
$\alpha = \tilde{\alpha} - \frac{1}{N} f \cdot \mathbf{1}$ is flat.
\end{remark}

An example of a hybrid algebra of observables is the trivial bundle of matrix algebras,
i.e. $A_x \simeq \mathrm{End}(\mathbb{C}^N), \ E = \mathcal{M} \times \mathrm{End}(\mathbb{C}^N)$ 
with a trivial connection $\alpha = 0$. In this case the trivial vector bundle 
$V = \mathcal{M} \times \mathbb{C}^N$ with a trivial connection $\beta = 0$ is an example
of a hybrid module.

\subsection{Lagrangian modules}
\label{sec:lagrangian_modules}

Let $\mathcal{M} \stackrel{\pi}{\rightarrow} B$ be a Lagrangian fibration on $\mathcal{M}$, i.e.
a surjective mapping s.t. $\pi^{-1}(b) \subset \mathcal{M}$ is a Lagrangian submanifold for
generic $b \in B$.

Let $L \subset \mathcal{M}$ be a Lagrangian submanifold such that the intersection
\begin{equation}
    L \cap \pi^{-1}(b) = \{x_{L, b}\}
\end{equation}
is a point for all $b$ in an open dense subset $B' \subset B$. In this case the mapping
\begin{equation}
    s_L \colon b \mapsto x_{L, b} \in \mathcal{M},
\end{equation}
is a section of $\pi$ over $B'$, i.e. $\pi \circ s_L = id$.

For example, if $Q_n$ is a smooth manifold of dimension $n$ and $\mathcal{M} = T^* Q_n$, then
$\pi \colon T^* Q_n \to Q_n$ is a Lagrangian fibration with fibers being $T^*_q Q_n$. An
example of a global Lagrangian section of $\pi$ is $L_f = \{(p = df(q), q)\}$, where
$f \in C(Q_n)$ is such that $df(q_1) = df(q_2)$ iff $q_1 = q_2$
(for example a monotonic function on $\mathbb{R}^n$).
Note that $\pi(L_f) = Q_n$ is a diffeomorphism $L_f \simeq Q_n$.

\begin{definition}
Let $\mathcal{H} = \Gamma(\mathcal{M}, V)$ be a hybrid module over $A = \Gamma(\mathcal{M}, E)$.
Define the vector bundle $V^{B, L}$ over a dense open subset $B'$ of generic points of $B$
as a vector bundle with the fiber $V^{B, L}_b = V_{x_{L, b}}$ over generic $b \in B'$.

The space of sections of $\pi^L_B$, $\mathcal{H}_B^L = \Gamma(B, V^{B, L})$ is called a
\textbf{Lagrangian module} over $A$. The $A$-module structure on $\mathcal{H}^L_B$ is
\begin{equation}
    (s f)(b) = s(x_{L, b}) f(b).
\end{equation}
Here $s \in A = \Gamma(\mathcal{M}, E)$, $f \in \mathcal{H}_B^L$ and $f(b) \in V_{x_{L, b}}$.
\end{definition}

\subsection{Derivations and automorphisms}
Here we outline some basic facts about the derivations of hybrid algebras.

\begin{definition} A \textbf{derivation of a hybrid algebra} is a derivation
$D \colon A \to A$ of the associative algebra $A$, i.e. a linear map $A \to A$ such that
$D(ab) = D(a) b + a D(b)$ for $a, b \in A$, 
which is also a derivation of the Poisson structure on $Z(A)$, i.e.
\begin{equation}
    D(\{z, w\}) = \{D(z), w\} + \{z, D(w)\}
\end{equation}
for any $z, w \in Z(A)$.
\end{definition}

Note that from the definition of a derivation of an associative algebra
it follows that if $z \in Z(A)$, its derivation $D(z)$ also belongs to the center,
therefore, the Poisson brackets $\{D(z), w\}$ and $\{z, D(w)\}$ are correctly defined.

Derivations of a hybrid algebra \(A\) form a Lie algebra \(\mathrm{Der}(A)\)
with the bracket given by the commutator \([D, E](a) = D(E(a)) - E(D(a))\).

We will use the following terminology:
\begin{itemize}
\item $D$ is a \textbf{quantum derivation} if it is an inner derivation,
    i.e. $D(a) = i \, [H^{(1)}_D, a]$ for some $H^{(1)}_D \in A$
    \footnote{
    Note that $H^{(1)}_D$ is determined by the derivation $D$ only up to a central element.
    }.
\item $D$ is a \textbf{Hamiltonian derivation} if $D(a) = \{H^{(0)}_D, a\}$
    for some $H^{(0)}_D  \in Z(A)$.
\item $D$ is a \textbf{hybrid derivation}
    if $D(a) = \{H^{(0)}_D, a\} + i \, [H_D^{(1)}, a]$ for
    $H^{(0)}_D \in Z(A)$ and $H^{(1)}_D \in A \ $
    \footnote{
    In \cite{MV} this derivation of algebra $A$ also appeared,
    but it is interpreted as a Hamiltonian derivation.}.
\end{itemize}

Let \(\mathrm{HDer}(A)\) be the space of hybrid derivations of a hybrid algebra \(A\).
We have a natural isomorphism of vector spaces
\begin{equation}
    \mathrm{HDer}(A) \simeq Z(A) \oplus A/Z(A),
\end{equation}
given by \(D \mapsto (H^{(0)}_D, \overline{H^{(1)}_D})\), where \(\overline{H^{(1)}_D}\)
is the class of \(H^{(1)}_D \ \mathrm{mod} \ Z(A)\).

Let \(D\) and \(E\) be two hybrid derivations of a hybrid algebra \(A\)
\begin{equation}
    D(a) = \{H^{(0)}_D, a\} + i \, [H^{(1)}_D, a], \quad
    E(a) = \{H^{(0)}_E, a\} + i \, [H^{(1)}_E, a].
\end{equation}
Then the commutator of these derivations is also a hybrid derivation
\begin{equation}
    \label{eq: hybrid-derivations-commutator}
    [D, E](a) = \{\{H^{(0)}_D, H^{(0)}_E\}, a\} + 
        i \, \big[\{H^{(0)}_D, H^{(1)}_E\} - \{H^{(0)}_E, H^{(1)}_D\} + 
            i \, [H^{(1)}_D, H^{(1)}_E] - \{H^{0}_D, H^{(0)}_E\}_2, a \big].
\end{equation}
Formula \eqref{eq: hybrid-derivations-commutator} implies that \(\mathrm{HDer}(A)\)
is a Lie subalgebra in \(\mathrm{Der}(A)\). Moreover, we have an isomorphism
of Lie algebras
\begin{equation}
    \mathrm{HDer}(A) \simeq \mathcal{P}(A),
\end{equation}
where \(\mathcal{P}(A)\) is defined in the definition \ref{def: hybrid-algebra}.

A derivation $D \colon A \to A$ is \textbf{represented} in an $A$-module ($\mathcal{H}$,
$\rho \colon A \to  \mathrm{End}(\mathcal{H})$), if $\mathcal{H}$ is equipped with a linear map
$D_{\mathcal{H}} \colon \mathcal{H} \to \mathcal{H}$, such that
\begin{equation}
    \label{derivation-rep}
    D_{\mathcal{H}}(a v) = D(a) v + a D_{\mathcal{H}}(v), \quad a \in A, \ v \in \mathcal{H}.
\end{equation}
In particular, a hybrid derivation $D$ is represented in an $A$-module $\mathcal{H}$ by
\begin{equation}
    D_{\mathcal{H}}(v) = \{H^{(0)}_D, v\} + i \tilde{H}_D^{(1)} v, \quad v \in \mathcal{H}.
\end{equation}
where $\tilde{H}_D^{(1)} = H_D^{(1)} + \delta$ and $\delta$ is any operator such that 
$[\delta, \rho(a)] = 0$ for any $a \in A$, i.e. it is an element of the centralizer
of $\rho(A) \subset \mathrm{End}(\mathcal{H})$.

\begin{definition}
A linear mapping $\varphi \colon A \to B$ is called a
\textbf{homomorphism of hybrid algebras} if it is a homomorphism of associative algebras
\begin{equation}
    \varphi(ab) = \varphi(a) \varphi(b)
\end{equation}
and a morphism of Poisson structures on $Z(A)$ and $Z(B)$
\begin{equation}
    \varphi(\{z, w\}_A) = \{\varphi(z), \varphi(w)\}_B,
\end{equation}
for all $z, w \in Z(A)$.
\end{definition}

Note that the homomorphism of associative algebras maps center to center, so the Poisson
bracket $\{\varphi(z), \varphi(w)\}_B$ in this definition is the Poisson bracket on
$Z(B)$.

An invertible homomorphism of hybrid algebras $\varphi \colon A \to A$ is called
an \textbf{automorphism} of $A$.

\subsection{The relation to deformation quantization} \label{sec: deformation_quantization}
Let $A_0$ be an associative algebra and $Z(A_0)$ be its center. Let $A_\hbar$ be a flat deformation family
of $A_0$, i.e. a family of associative algebras $A_\hbar$ together with linear isomorphisms
$\phi_\hbar \colon A_\hbar \to A_0$ such that $\phi_0 = id$.
In many practically interesting cases, such linear isomorphisms are given by an identification
of linear bases or PBW bases in $A_\hbar$ and in $A_0$.

Let $* \colon A_0 \times A_0 \to A_0$ be the corresponding $*$-product on $A_0$
\begin{equation}\label{eq: star-product}
    a * b = \phi_\hbar \left( \phi_\hbar^{-1}(a) \phi_\hbar^{-1}(b)\right).
\end{equation}

Define the $*$-commutator as
\begin{equation} \label{eq: star-commutator}
    \left[a, b \right]_{*} = a * b - b * a =
        \phi_\hbar \big([\phi_\hbar^{-1}(a), \phi_\hbar^{-1}(b)]\big).
\end{equation}
In these formulas, the product on the right-hand side is taken in $A_\hbar$.
Let us assume that the product \eqref{eq: star-product} is given by the analytic coefficient
functions in an appropriate topological sense (for example, coefficients in a PBW basis
where coefficients depend analytically on \(\hbar\), or the algebra of 
\(\hbar\)-differential operators, etc.). Then for small \(\hbar\)
\begin{equation} \label{eq: star-product-A0}
    a * b = a b + \sum_{k \ge 1} \hbar^k m_k(a, b).
\end{equation}

If one wants to separate the algebraic and analytical aspects of deformation quantization, a natural
setting is formal deformation quantization, where the product \eqref{eq: star-product} is defined
over \(A_0 [[\hbar]]\) and is a formal power series in \(\hbar\).

Restricted on a commutative subalgebra $Z(A_0) \subset A_0$, the $*$-commutator in the leading order
in $\hbar$ induces a Poisson structure on $Z(A_0)$. Define a skew-symmetric bilinear operation
\begin{equation}
    \{z, w\} = i \big(m_1(z, w) - m_1(w, z)\big), \qquad z, w \in Z(A_0).
\end{equation}
The Jacobi identity for the commutator in $A_\hbar$ guarantees that $\{z, w\} \in Z(A_0)$,
and satisfies Leibniz and Jacobi identities:
\begin{equation}
    \{zw, x\} = \{z, x\} w + z \{w, x\}, \qquad
    \{z, \{w, x\}\} + \{w, \{x, z\}\} + \{x, \{z, w\}\} = 0
\end{equation}
for each $z, w, x \in Z(A_0)$.

Let us show that $A_0$ is a hybrid algebra. Define skew-symmetric bilinear operations
\begin{align} \label{eq: Poisson-bracket-Z-A}
    \{z, a\} &= i \, \big(m_1(z, a) - m_1(a, z)\big), \qquad z \in Z(A_0),\ a \in A_0.
\end{align}
and
\begin{equation} \label{eq: second_bracket-Z}
    \{z, w\}_2 = -i \, \big(m_2(z, w) - m_2(w, z) \big), \qquad z, w \in Z(A_0).
\end{equation}

Let us start from auxilary lemmas describing the properties of these operations.

\begin{lemma} \label{lem:C}
For all $z, w, x \in Z(A_0)$ and $a, b, c \in A_0$,
\begin{gather}
    \{z, ab\} = \{z, a\} b + a \{z, b\}, \\
    \{zw, a\} = z \{w, c\} + w \{z, c\}, \\
    \{\{z, w\}, a\} = \{z, \{w, a\}\} - \{w, \{z, a\}\} + i \, [\{z, w\}_2, a], \\
    \{z, \{w, x\}_2\} + \{w, \{x, z\}_2\} + \{x, \{z, w\}_2\} +
        \{z, \{w, x\}\}_2 + \{w, \{x, z\}\}_2 + \{x, \{z, w\}\}_2 = 0.
\end{gather}
\end{lemma}
\begin{proof}
The $*$-commutator satisfies the Leibniz rule
\begin{equation}
    [a * b, c]_* = a * [b, c]_* + [a, c]_* * b
\end{equation}
and Jacobi identity
\begin{equation}
    [a, [b, c]_*]_* + [b, [c, a]_*]_* + [c, [a, b]_*]_* = 0
\end{equation}
for any $a, b, c \in A_0$.

We obtain the proof of lemma \ref{lem:C} by expanding the Leibniz rule and the
Jacobi identity to the first non-trivial terms.
The first identity we obtain in the order \(\hbar\) in the
Leibniz rule for \(a, b \in A_0\), and \(c = z \in Z(A_0\),
the second one appears in the order \(\hbar\) in the
Leibniz rule for \(c \in A_0\), and \(a = z, b = w\) in \(Z(A_0)\).
The third and the last properties follows from the Jacobi identity:
in the order $\hbar^2$ for $a \in A_0$, and $b = z, \ c = w$ in $Z(A_0)$,
and in the order $\hbar^3$ for $a = z, \ b = w, \ c = x$, where $z, w, x \in Z(A_0)$.
\end{proof}

\begin{lemma} (Mikhailov--Vanhaecke Poisson algebra) 
\(\mathcal{P}(A_0) = Z(A_0) \oplus A_0/Z(A_0)\) with the Poisson bracket
and the multiplication defined as
\begin{gather}
    \label{eq: P(A)-deformation-bracket}
    \{(z_1, \overline{a_1}), (z_2, \overline{a_2})\} =
        \big(\{z_1, z_2\}, 
            \overline{\{z_1, a_2\} - \{z_2, a_1\} + i \, [a_1, a_2] - \{z_1, z_2\}_2}
        \big), \\
    \label{eq: P(A)-deformation-multiplication}
    (z_1, \overline{a_1}) \cdot (z_2, \overline{a_2}) = 
        \big(z_1 z_2, \overline{z_1 a_2 + z_2 a_1 + m_1(z_1, z_2)} \big),
\end{gather}
where \(z_1, z_2 \in Z(A_0)\), and \(\overline{a_1}, \overline{a_2}\) are classes of
elements \(a_1, a_2 \in A_0\) in \(A_0/Z(A_0)\), is the Poisson algebra.
\end{lemma}

The proof of this lemma is given at \cite{MV}.
Note that both operations \eqref{eq: P(A)-deformation-bracket} and 
\eqref{eq: P(A)-deformation-multiplication} are correctly defined
(the results do not depend on the choice of the representatives \(a_1, a_2 \in A_0\)
in classes \(\overline{a_1}, \overline{a_2}\)) 
and that the multiplication \eqref{eq: P(A)-deformation-multiplication}
is commutative (because the element $m_1(z_1, z_2) - m_1(z_2, z_1) = -i \{z_1, z_2\} \in Z(A_0)$
for two central elements \(z_1, z_2\)).

Mikhailov and Vanhaecke also showed that the algebra $A_0$ is a Poisson module over
$\mathcal{P}(A_0)$, with
\begin{equation}
    (z, \overline{a}) \cdot b = z b, \qquad
    \{(z, \overline{a}), a\} = \{z, b\} - [a, b], \qquad
    (z, \overline{a}) \in \mathcal{P}(A_0), \ b \in A_0.
\end{equation}
Note that $A_0/Z(A_0) \subset \mathcal{P}(A_0)$ acts trivially on $A_0$ as a commutative subalgebra.  

Thus, these two lemmas show that a flat deformation family of $A_0$ induces a hybrid
algebra structure on $A_0$.
If we add the assumption that $A_0$ is finite-dimensional and simple over $Z(A_0)$,
we arrive at the definition of a Poisson Azumaya algebra. 
In this context of deformation quantization, 
hybrid Poisson algebras appear in quantum groups at roots of unity \cite{DCK}
and in quantum affine algebras at the critical level \cite{FrR}.

\begin{remark}\label{rem: two quantum strusture}
Note that the Poisson bracket on $Z(A_0)$ does not depend on the changes in the
identification of vector spaces $A_0$ and $A_\hbar$. However, the hybrid algebra
structure on $A_0$ changes. Indeed, let $\psi_\hbar \colon A_\hbar \to A_0$ be
another linear isomorphism such that
\begin{equation}
    \phi_\hbar = \eta_\hbar \circ \psi_\hbar, \qquad
        \eta_\hbar \colon A_0 \to A_0, \quad
            \eta_\hbar(a) = a + \sum_{k \ge 1} \hbar^k \eta^{(k)}(a).
\end{equation}
Then, for $z, w \in Z(A_0)$
\begin{equation}
    \psi_\hbar([\psi_\hbar^{-1}(z), \psi_\hbar^{-1}(w)]) =
        -i \hbar {\{z, w\}}^{(\psi)} +
            i \hbar^2 \{z, w\}_2^{(\psi)} + O(\hbar^3),
\end{equation}
and for $z \in Z(A_0), \ a \in A_0$
\begin{equation}
    \psi_\hbar([\psi_\hbar^{-1}(z), \psi_\hbar^{-1}(a)]) =
        -i \hbar \{z, a\}^{(\psi)} + O(\hbar^2),
\end{equation}
where
\begin{gather}\label{eq: psi-phi brackets relations}
    \{z, w\}^{(\psi)} = \{z, w\}^{(\phi)}, \\ 
    \label{eq: psi-phi brackets relation - 1}
    \{z, a\}^{(\psi)} = \{z, a\}^{(\phi)} + i [\eta^{(1)}(z), a], \\ 
    \label{eq: psi-phi brackets relation - 2}
    \{z, w\}_2^{(\psi)} = \{z, w\}_2^{(\phi)} - \{z, \eta^{(1)}(w)\}^{(\phi)} +
        \{w, \eta^{(1)}(z)\}^{(\phi)} - i \, \big[\eta^{(1)}(z), \eta^{(1)}(w)\big] + 
            \eta^{(1)}\big(\{z, w\}^{(\phi)}\big).
\end{gather}
\end{remark}

Note that the derivations of $A_\hbar$ naturally induce the derivations of $A_0$.

\section{Hybrid states}
\label{sec:hybrid-states}

\subsection{Classical states}
Recall that a \textbf{classical state} on $\mathcal{M}$ is a probability distribution on $\mathcal{M}$.
An example of such a state is a distribution given by a nonnegative normalized density function $\rho_c$\ \footnote{
    Recall that a density function is a function only on open neighborhoods of $\mathcal{M}$. On the
    intersection $U \cap V$ we have
    $\rho_{U, c}(x) = \rho_{V, c}(y) \left| \tfrac{\partial y}{\partial x} \right|$.
    The Euclidean volume $\rho_{U, c}(x) d^{2n} x$ in this case is globally defined. An orientation of
    $\mathcal{M}$ gives an identification of densities with top forms on $\mathcal{M}$.
}
\begin{equation}
    \int_\mathcal{M} \rho_c\, d^{2n} x = 1.
\end{equation}
The value of a classical observable $f$ on the classical state with the density function $\rho_c$ is
\begin{equation}
    \mathbb{E}_{\rho_c}(f) = \int_\mathcal{M} f(x) \rho_c(x) d^{2n} x.
\end{equation}
Because $\mathcal{M}$ is symplectic, we have the symplectic volume form $\omega_x^n$. The density
function $\rho_c$ of the classical state can now be identified with a function
$\rho^c \in C(\mathcal{M})$, such that
\begin{equation}
    \rho^c \omega^n = \rho_c d^{2n} x.
\end{equation}
Then for the expectation value of an observable, we have
\begin{equation}
    \mathbb{E}_{\rho_c}(f) = \int_\mathcal{M} \rho^c(x) f(x) \omega_x^n, \quad
        n = \frac{\dim \mathcal{M}}{2}.
\end{equation}

\subsection{Hybrid states}
\label{sec: hybrid states}
Define the bundle of local quantum states as a fiber bundle
\begin{equation}
\begin{tikzcd}
    S \arrow[d] & \arrow[l] S_x \\
    \mathcal{M} &
\end{tikzcd}
\end{equation}
where $S_x$ is the space of positive linear functions on $A_x$,
$\lambda_x \colon A_x \to \mathbb{C}$ such that
\begin{equation}
    \int_\mathcal{M} \lambda_x(1_x) \omega_x^n = 1.
\end{equation}
We will call such positive linear functionals \textbf{normalized}.

In the case
$A_x \cong \mathrm{End}(\mathbb{C}^N)$, the space $S_x$ can be identified with the space of Hermitian matrices with nonnegative eigenvalues, i.e. with the space of density matrices.
For a density matrix $\rho = \{\rho_x\}_{x \in \mathcal{M}}$ its trace $\mathrm{Tr}(\rho_x) = \rho^c_x$ is a
positive-valued function on $\mathcal{M}$.

A hybrid state with the density matrix $\rho \in \Gamma(\mathcal{M}, S)$ is \textbf{normalized} if
\begin{equation}
    \int_\mathcal{M} \mathrm{Tr}(\rho_x) \omega_x^n = 1,
\end{equation}
i.e. if $\rho_x^c = \mathrm{Tr}(\rho_x)$ is a classical state.

For a given $\rho^c_x$ the space of density matrices with $\mathrm{Tr}(\rho_x) = \rho^c_x$ is a compact
convex subset $S_x(\rho^c_x) \subset \mathrm{End}(\mathbb{C}^N)$. A \textbf{pure hybrid state}
with given $\rho^c_x$ is an extremal point of $S_x(\rho^c_x)$,
normalized as above. Density matrices for such states are
one-dimensional orthogonal projectors. They can be written as
\begin{equation}
    \rho_x = v_x \otimes v_x^*,
\end{equation}
where $v_x \in V_x$.

The value of an observable $s \in A$ on the hybrid state with the density matrix $\rho$ is
\begin{equation}
    \mathbb{E}_\rho(s) =
        \int_\mathcal{M} \mathrm{Tr}_{V_x}(s_x \rho_x) \omega_x^n.
\end{equation}

Note that pure hybrid states with \( v_x \) and \( e^{i \alpha_{x}} v_x \) are equivalent,
i.e. the space of pure states can be identified with \( \Gamma(V) / G_\mathcal{M} \),
where \( G_\mathcal{M} = \mathrm{Maps}(\mathcal{M}, U(1)) \) is the ``gauge group'',
and \( V \) is a hybrid module over \( A = \Gamma(\mathcal{M}, E) \).
Thus, pure hybrid states are parametrized by gauge classes of vectors in a hybrid module over the hybrid algebra of observables.

\subsection{Lagrangian states}
\subsubsection{Classical Lagrangian states}
Fix a Lagrangian fibration on $\mathcal{M}$, i.e. fix a projection $\pi \colon \mathcal{M} \to B$,
where generic fiber $\pi^{-1}(b)$ is a Lagrangian submanifold.

Let $\rho_B$ be a density function on $B$, and the measure $\mu(U) = \int_U \rho_B(b) d^n b$ is globally
defined. Assume that $\rho_B$ is normalized, i.e. $\int_B \rho_B(b) d^n b = 1$.

The classical state with the density function
\begin{equation}
    \label{classical-density}
    \rho_x^{\mathrm{cl}} = \int_B \rho_B(b) \delta(x, x_{L, b}) d^n b.
\end{equation}
is called \textbf{classical Lagrangian state}.
Here $\delta(x, y)$ is a distribution supported on the diagonal of $\mathcal{M} \times \mathcal{M}$, i.e.
\begin{equation}
    \int\limits_{\mathcal{M} \times \mathcal{M}}
        \delta(x, y) g(x, y) \omega_x^n \omega_y^n =
    \int\limits_\mathcal{M} g(x, x) \omega_x^n
\end{equation}
for every test function $g(x, y)$. In other words
\begin{equation}
    \int_\mathcal{M} f(x) \delta(x, x') \omega_x^n = f(x').
\end{equation}

The expectation value of a classical observable on a classical Lagrangian state is
\begin{equation}
    \mathbb{E}_{\rho^{\mathrm{cl}}}(f) =
        \int_\mathcal{M} \rho_x^{\mathrm{cl}} f(x) \omega_x^n =
        \int_B f(x_{L, b}) \rho_B(b) d^n b.
\end{equation}

\subsubsection{Hybrid Lagrangian states}
Let $x_{L, b} = L \cap \pi^{-1}(b)$ be as above and
\begin{equation}
    \rho^L_B(b) \colon V_{x_{L, b}} \to V_{x_{L, b}}
\end{equation}
be a Hermitian nonnegative operator. Assume that $\mathrm{Tr}(\rho^L_B(b))$ is a density on $B$. Define
a \textbf{hybrid Lagrangian state} as the following linear functional on $A$
\begin{equation}
    \mathbb{E}_{\rho_B^L}(s) =
        \int_B \mathrm{Tr}_{V_{x_{L, b}}}
            (\rho_B^L(b) s_{x_{L, b}}) d^n b,
\end{equation}
assuming that $\rho_B^L$ is normalized, i.e. $\mathbb{E}_{\rho^L_B}(1) = 1$.

\subsubsection{Pure hybrid Lagrangian states}
For a Lagrangian fibration $\pi \colon \mathcal{M} \to B$ define the ''space of wavefunctions''
$\mathcal{H}_B^L$ as the space of $1/2$-density sections of the vector bundle $V^{B, L} \to B$
with the fibers $V_{x_{L, b}}$.
For generic $b \in B$,
$V_{x_{L, b}} \simeq \mathbb{C}^N$ with the natural Hermitian structure inherited from $V$.

For $\varphi \in \mathcal{H}_B^L$ define the density matrix of the corresponding hybrid Lagrangian pure
state as the one-dimensional orthogonal projector
\begin{equation}
    \rho_B^L(b) = \varphi(b) \otimes \varphi^*(b)
        \colon V_{x_{L, b}} \to V_{x_{L, b}}
\end{equation}
normalized as
\begin{equation}
    \int_B \mathrm{Tr}_{V_{x_{L, b}}}(\rho_B^L(b)) d^n b =
        \int_B \| \varphi(b) \|^2_{x_{L, b}} d^n b = 1,
\end{equation}
here $\varphi(b) \in V_{x_{L, b}}$ and $\| \varphi(b) \|^2_{x_{L, b}}$ is the norm in $V_{x_{L, b}}$.

The expectation value of an observable $s$ on this state is
\begin{equation}
    \mathbb{E}_{\rho_B^L} (s) =
        \int_B (\varphi(b), s_{x_{L, b}} \varphi(b))_{x_{L, b}} d^n b.
\end{equation}

Note that Lagrangian hybrid states and pure Lagrangian hybrid states are hybrid states
supported (in the sense of distributions) on the Lagrangian subspace \( L \subset M \)
(the Lagrangian section of \( \pi \colon M \to B \)) which is projectable to \( B \). 
So we can write it as an integral over \( L \).

\section{The hybrid evolution} \label{evolution}

\subsection{The time evolution of observables} \label{sec: observables-evolution}

A derivation $D$ of a hybrid algebra $A$  defines a $1$-parametric family of
automorphisms of $A$, $\varphi_t \colon A \to A$ such that $a(t) = \varphi_t(a)$ is
a solution of the differential equation
\begin{equation}\label{eq: a evo}
    \frac{\partial a(t)}{\partial t} = D(a(t)), \quad \text{with} \quad a(0)= a.
\end{equation}

If \( D \) is a \( * \)-derivation, i.e. \( D(a)^* = D(a^*), \) the family 
\( \varphi_t \colon A \to A \) is a family of unitary \( * \)-automorphisms 
\( \varphi_t(a)^* = \varphi_t^{-1}(a^*) = \varphi_{-t}(a) \).

\begin{definition}
A \textbf{hybrid system} is a hybrid \( * \)-algebra with a classical Hamiltonian
\( H^{(0)} \in Z(A) \) and a quantum Hamiltonian \( H^{(1)} \in A \),
such that \( H^{(0)} = (H^{(0)})^*, H^{(1)} = (H^{(1)})^* \).
In this case, the algebra \( A \) is the algebra of hybrid observables.
The space of observables is a real subspace of \( * \)-invariant elements in \( A \).
\end{definition}

The \textbf{time evolution} in a hybrid system is a 1-parametric family of automorphisms
\( \varphi_t(a) = a(t) \) \eqref{eq: a evo} given by the derivation
\( D(a) = \{ H^{(0)}, a \} + i \, [H^{(1)}, a] \). This is the hybrid analog of the Heisenberg
evolution in quantum mechanics and of the Hamiltonian evolution in classical mechanics.

Note that hybrid evolutions with \( H^{(1)} \) and with \( H^{(1)} + z \), where
\( z \in Z(A) \), are identical.
Thus, the subspace of fixed points of the \( * \)-involution in Lie algebra 
\( Z(A)^{\vee} = Z(A) \oplus A / Z(A) \) (see the section \ref{sec: hybrid algebra}) 
is the space of possible hybrid evolutions.

In the example of a bundle of algebras over a symplectic manifold 
the \textbf{hybrid evolution} of $s \in \Gamma(\mathcal{M}, E)$ is
\begin{equation}
    \frac{\partial s(t)}{\partial t} =
        \{H^{(0)}, s(t)\} + i \, [H^{(1)}, s(t)], \quad
    s(0) = s.
\end{equation}
Fiberwise on $A_x$ we have
\begin{equation}
    \frac{\partial s_x(t)}{\partial t} =
        \{H^{(0)}, s(t)\}_x + i \, [H^{(1)}_x, s_x(t)].
\end{equation}
Note that $H^{(0)}$ is an integral of motion for this evolution, but $H^{(1)}$ is not.

\subsection{The classical case}

Assume $H^{(1)} = 0$; in this case, the classical dynamics is lifted to quantum fibers
using the connection $\alpha$
\begin{equation}
    \frac{\partial s(t)}{\partial t} = \{H^{(0)}, s(t)\}.
\end{equation}

\begin{proposition} \label{classical-evolution-connection}
The formula
\begin{equation}
    s_x(t) = h_{x, x(t)} s_{x(t)} h_{x(t), x},
\end{equation}
where $h_{x, x(t)}$ and
$h_{x(t), x}$ are parallel transport operators along a classical trajectory defined by the connection $\alpha$:
\begin{equation}
    \frac{\partial h_{x, x(t)}}{\partial t} =
        h_{x, x(t)} \alpha_j(x(t)) \dot{x}^j(t), \quad
    \frac{\partial h_{x(t), x}}{\partial t} =
        -\dot{x}^j(t) \alpha_j(x(t)) h_{x(t), x}.
\end{equation}
gives the solution to the Cauchy problem
\begin{equation}
    \frac{\partial s(t)}{\partial t} = \{H^{(0)}, s(t)\}, \quad
    s(0) = s.
\end{equation}
\end{proposition}

\begin{proof}
Consider $s_x(t) = h_{x, x(t)} s_{x(t)} h_{x(t), x}$ and derive the equation for $s_x(t)$
\begin{multline} \notag
    \frac{\partial s_x(t)}{\partial t} =
        h_{x, x(t)} \alpha_j(x(t)) \dot{x}^j(t) s_{x(t)} h_{x(t), x} +
        h_{x, x(t)} \dot{x}^j(t) \frac{\partial s}{\partial x^j} \Big|_{x(t)} h_{x(t), x} -
        h_{x, x(t)} s_{x(t)} \dot{x}^j(t) \alpha_j(x(t)) h_{x(t), x} = \\ =
        (\omega^{-1})^{ij}(x(t)) \frac{\partial H^{(0)}}{\partial x^i} \Big|_{x(t)}
            h_{x, x(t)} [\alpha_j(x(t)), s_{x(t)}] h_{x(t), x} +
        (\omega^{-1})^{ij}(x(t)) \frac{\partial H^{(0)}}{\partial x^i} \Big|_{x(t)}
            h_{x, x(t)} \frac{\partial s}{\partial x^j} \big|_{x(t)} h_{x(t), x} = \\ =
        (\omega^{-1})^{ij}(x(t)) \frac{\partial H^{(0)}}{\partial x^i} \Big|_{x(t)}
            h_{x, x(t)} \nabla_j^{x(t)} s_{x(t)} h_{x(t), x} =
              (\omega^{-1})^{ij}(x(t))
            \left( J^{-1}\right)^l_i
            \frac{\partial H^{(0)}}{\partial x^l} \,
            h_{x, x(t)} \nabla_j^{x(t)} s_{x(t)} h_{x(t), x}.
\end{multline}
where $J = \tfrac{\partial x(t)}{\partial x} $ is Jacobian.

\begin{lemma} The following holds:
    \begin{equation}
        \nabla_j (h_{x, x(t)}  s_{x(t)} h_{x(t), x}) =
            \frac{\partial x^k(t)}{\partial x^j}
                h_{x, x(t)} \nabla_k^{x(t)} s_{x(t)} h_{x(t), x}.
    \end{equation}
\end{lemma}

\begin{proof}
By definition
\begin{equation}
    \nabla_j (h_{x, x(t)}  s_{x(t)} h_{x(t), x}) =
        \frac{\partial}{\partial x^j} (h_{x, x(t)}  s_{x(t)} h_{x(t), x}) +
        [\alpha_j(x), h_{x, x(t)}  s_{x(t)} h_{x(t), x}].
\end{equation}
By definition of the holonomy $h_{x, y}$
\begin{equation}
    \frac{\partial h_{x, y}}{\partial x^j} = -\alpha_j(x) h_{x, y}, \quad
    \frac{\partial h_{x, y}}{\partial y^j} = h_{x, y} \alpha_j(y).
\end{equation}
Thus,
\begin{multline} \notag
    \frac{\partial}{\partial x^j} (h_{x, x(t)}  s_{x(t)} h_{x(t), x}) =
        \frac{\partial h_{x, x(t)}}{\partial x^j} s_{x(t)} h_{x(t), x} +
        h_{x, x(t)}  s_{x(t)} \frac{\partial h_{x(t), x}}{\partial x^j} + \\ +
        \frac{\partial x^k(t)}{\partial x^j} \left(
            \frac{\partial h_{x, x(t)}}{\partial x^k(t)} s_{x(t)} h_{x(t), x} +
            h_{x, x(t)}  \frac{\partial s_{x(t)}}{\partial x^k(t)} h_{x(t), x} +
            h_{x, x(t)}  s_{x(t)} \frac{\partial h_{x(t), x}}{\partial x^k(t)}
        \right) = \\ =
        -\alpha_j(x) h_{x, x(t)}  s_{x(t)} h_{x(t), x} +
            h_{x, x(t)}  s_{x(t)} h_{x(t), x} \alpha_j(x) + \\ +
            \frac{\partial x^k(t)}{\partial x^j}
            h_{x, x(t)} \left(
                \alpha_k(x(t)) s_{x(t)} +
                \frac{\partial s_{x(t)}}{\partial x^k(t)} -
                s_{x(t)} \alpha_k(x(t))
            \right) h_{x(t), x} = \\ =
        -[\alpha_k(x), h_{x, x(t)} s_{x(t)} h_{x(t), x}] +
        \frac{\partial x^k(t)}{\partial x^j} h_{x, x(t)}
            \nabla_k^{x(t)} s_{x(t)} h_{x(t), x}.
\end{multline}
Then,
\begin{equation}
    \nabla_j(h_{x, x(t)}  s_{x(t)} h_{x(t), x}) =
        J^k_j h_{x, x(t)} \nabla_k^{x(t)} s_{x(t)} h_{x(t), x}.
\end{equation}
\end{proof}

Applying this lemma, we have
\begin{multline}\nonumber
    \frac{\partial s_x(t)}{\partial t} =
        (\omega^{-1})^{ij}(x(t))  \left( J^{-1}\right)^l_i
        \frac{\partial H^{(0)}}{\partial x^l}
        \left( J^{-1}\right)^k_j
        \nabla_k (h_{x, x(t)} s_{x(t)} h_{x(t), x}) = \\ =
          (\omega^{-1})^{kl}(x)  \frac{\partial H^{(0)}}{\partial x^l}
        \nabla_k (h_{x, x(t)} s_{x(t)} h_{x(t), x})  =
        (\omega^{-1})^{kl}(x) \frac{\partial H^{(0)}}{\partial x^l}
        \nabla_k s_x(t) =
        \{H^{(0)}, s(t)\}_x.
\end{multline}

Thus, we proved the proposition.
\end{proof}

\subsection{The hybrid case}
Now assume that $H^{(1)} \ne 0$.

Let $\mathcal{U}(t)$ be a solution to the Cauchy problem
\begin{equation}
    \frac{\partial \mathcal{U}(t)}{\partial t} =
        \{H^{(0)}, \mathcal{U}(t)\} + i H^{(1)} \mathcal{U}(t), \quad
    \mathcal{U}(0) = 1.
\end{equation}
\textbf{Remark:}
Note that if $\{H^{(0)}, H^{(1)}\} = 0$, we have $\mathcal{U}_x(t) = e^{i H^{(1)}_x t}$.
But since this is generally not the case, $\mathcal{U}(t)$ has a more complicated form.

\begin{theorem}\label{thr:s_x(t)}
The solution to
\begin{equation}
    \frac{\partial s_x(t)}{\partial t} = \{H^{(0)}, s(t)\}_x + i [H^{(1)}_x, s_x(t)], \quad
    s(0) = s
\end{equation}
in the $H_1 \ne 0$ case is given by
\begin{equation}
    \label{s(t)-motion}
    s_x(t) = \mathcal{U}_x(t) h_{x, x(t)} s_{x(t)}  h_{x(t), x}\, \mathcal{U}_x(t)^{-1}.
\end{equation}
\end{theorem}
\begin{proof}
Differentiating (\ref{s(t)-motion}) in time, we get
\begin{multline}\nonumber
    \frac{\partial s_x(t)}{\partial t} =
        \frac{\partial \mathcal{U}_x(t)}{\partial t}
            h_{x, x(t)} s_{x(t)} h_{x(t), x}\, \mathcal{U}_x(t)^{-1} + \\ +
        \mathcal{U}_x(t) \frac{\partial}{\partial t}
            \big(h_{x, x(t)} s_{x(t)} h_{x(t), x}\big) \mathcal{U}_x(t)^{-1} +
        \mathcal{U}_x(t) h_{x, x(t)} s_{x(t)} h_{x(t), x} \frac{\partial \mathcal{U}_x(t)^{-1}}{\partial t}.
\end{multline}
We have already proven in proposition \ref{classical-evolution-connection} that
\begin{equation}
    \frac{\partial}{\partial t} (h_{x, x(t)} s_{x(t)} h_{x(t), x}) =
        \{H^{(0)}, h_{x, x(t)} s_{x(t)} h_{x(t), x}\}.
\end{equation}
This implies
\begin{multline}\nonumber
    \frac{\partial}{\partial t} (h_{x, x(t)} s_{x(t)} h_{x(t), x}) =
        \{H^{(0)}, h_{x, x(t)} s_{x(t)} h_{x(t), x}\} =
            \{H^{(0)}, \mathcal{U}_x(t)^{-1} s_x(t) \mathcal{U}_x(t)\} = \\ =
    \mathcal{U}_x(t)^{-1} \{H^{(0)}, s(t)\}_x \,\mathcal{U}_x(t) -
        \mathcal{U}_x(t)^{-1} \{H^{(0)}, \mathcal{U}(t)\}_x \,
            \mathcal{U}_x(t)^{-1} s_x(t) \mathcal{U}_x(t) +
            \mathcal{U}_x(t)^{-1} s_x(t) \{H^{(0)}, \mathcal{U}(t)\}_x,
\end{multline}
and therefore,
\begin{equation}
    \frac{\partial s_x(t)}{\partial t} =
        \{H^{(0)}, s(t)\}_x +
        \left[\frac{\partial \mathcal{U}_x(t)}{\partial t} \mathcal{U}_x(t)^{-1} -
            \{H^{(0)}, \mathcal{U}(t)\}_x \, \mathcal{U}_x(t)^{-1}, s_x(t) \right].
\end{equation}
Using the equation on $\mathcal{U}(t)$, we obtain
\begin{equation}
    \frac{\partial s_x(t)}{\partial t} =
        \{H^{(0)}, s(t)\}_x + i[H^{(1)}_x, s_x(t)],
\end{equation}
which proves the theorem.
\end{proof}

\subsection{The evolution of hybrid states}
\label{sec:The evolution of hybrid states}
By definition, density matrices evolve as
\begin{equation}
    \label{expectation-value-time}
    \mathbb{E}_{\rho(t)}(s) = \mathbb{E}_\rho(s(t)),
\end{equation}
where $s$ is any observable and $\mathbb{E}_\rho(s)$ is the expectation value
of $s$ with the density matrix $\rho$.

\begin{proposition} States evolve according to solutions to the differential equation
\begin{equation}
    \frac{\partial \rho_x(t)}{\partial t} =
        -\{H^{(0)}, \rho(t)\}_x - i [H^{(1)}_x, \rho_x(t)].
\end{equation}
\end{proposition}
\begin{proof}
The local value of a state on evolving observable is
\begin{equation}
 \mathrm{Tr}_{V_x}(\rho_x s_x(t)) =
        \mathrm{Tr}_{V_x} \left( \rho_x \,\mathcal{U}_x(t) h_{x, x(t)} s_{x(t)}
            h_{x(t), x} \,\mathcal{U}_x(t)^{-1} \right) =
        \mathrm{Tr}_{V_x} \left(h_{x(t), x} \,\mathcal{U}_x(t)^{-1} \rho_x
            \,\mathcal{U}_x(t) h_{x, x(t)} s_{x(t)}\right).
\end{equation}
The global value of a state on an observable
\begin{gather}
    \mathbb{E}_\rho(s(t)) =
        \int_\mathcal{M} \mathrm{Tr}_{V_x}\left(\rho_x\, s_x(t) \right) \omega_x^n =
        \int_\mathcal{M} \mathrm{Tr}_{V_{x(t)}} \left(h_{x(t), x}
            \,\mathcal{U}_x(t)^{-1} \rho_x \,\mathcal{U}_x(t) h_{x, x(t)}\, s_{x(t)} \right) \omega_{x(t)}^n.
\end{gather}
Changing the variables $y = x(t), x = y(-t)$, we obtain
\begin{multline}
    \mathbb{E}_\rho(s(t)) =
        \int_\mathcal{M} \mathrm{Tr}_{V_y} \left(h_{y, y(-t)} \left( \mathcal{U}_{y(-t)}(t)\right)^{-1}
            \rho_{y(-t)} \,\mathcal{U}_{y(-t)}(t)\, h_{y(-t), y} \, s_y \right) \omega_y^n = \\ =
        \int_\mathcal{M} \mathrm{Tr}_{V_y}\left(\rho_y(t)\, s_y \right) \omega_y^n = \mathbb{E}_{\rho(t)}(s) .
\end{multline}
This implies
\begin{equation}
    \rho_y(t) = h_{y, y(-t)} \left( \mathcal{U}_{y(-t)}(t)\right)^{-1}  \rho_{y(-t)} \,\mathcal{U}_{y(-t)}(t)\, h_{y(-t), y},
\end{equation}
which gives the differential equation for $\rho_x(t)$.
\end{proof}

Note that the evolution of hybrid states is also defined by \(*\)-invariant elements of
\(\mathcal{P}(A) = Z(A) \oplus A / Z(A)\).

\subsection{The evolution of pure hybrid states}
\label{sec:The evolution of pure hybrid states}
Let \( V \) be a representation bundle over the hybrid algebra of observables 
\( \Gamma(\mathcal{M}, E) \) and let \( \rho_x = v_x \otimes v_x^* \) be the density matrix 
corresponding to \( v \in \Gamma(\mathcal{M}, V) \) (as in section \ref{sec: hybrid states}).

Vectors in the representation \( V \) evolve according to the hybrid analog of the Schr\"{o}dinger picture of the evolution of \( v \):
\begin{equation}\label{eq: schrodinger}
    \frac{\partial v(t)}{\partial t} =
        -\{ H^{(0)}, v(t) \} - i H^{(1)} v(t), \quad v(0) = v.
\end{equation}

The evolution of the density is invariant with respect to transformations 
\( H^{(1)} \mapsto H^{(1)} + z \), \( z \in Z(A) = C(\mathcal{M}) \). 
The Schr\"{o}dinger evolution is also invariant if we take into account gauge transformations.

\begin{proposition} 
Let \(v(t)\) be a solution to \eqref{eq: schrodinger} with \(H^{(0)}\) and \(H^{(1)}\),
\(z \in C(\mathcal{M})\), and \(\theta(t) = \int_0^t z(x(\tau)) \, d\tau\),
where \(x(\tau)\) is the flow line of the Hamiltonian vector field generated by \(H^{(0)}\).
Then \(e^{i \theta(t)} v(t)\) is a solution to \eqref{eq: schrodinger}  with 
\(H^{(0)}\) and \(H^{(1)} + z\).
\end{proposition}

\begin{proof}
For a classical trajectory $\{x(\tau)\}_0^t$ with $x(0) = x$
\begin{equation}
    \{ H^{(0)}, \ \theta(t) \} = 
        \int_0^t \{ H^{(0)}, \ z(x(\tau)) \} \, d\tau = z(x(t)) - z(x).
\end{equation}
We have
\begin{equation}
    \frac{\partial}{\partial t} \left( e^{i \theta(t)} v(t) \right) =
        i \frac{\partial \theta(t)}{\partial t} e^{i \theta(t)} v(t) +
        e^{i \theta(t)} \frac{\partial v(t)}{\partial t} =
    i z(x(t)) e^{i \theta(t)} v(t) + e^{i \theta(t)} \frac{\partial v(t)}{\partial t}.
\end{equation}
On the other hand
\[
    \{ H^{(0)}, e^{i \theta(t)} v(t) \} =
        i \{ H^{(0)}, \theta(t) \} e^{i \theta(t)} v(t) + e^{i \theta(t)} \{ H^{(0)}, v(t) \} =
    i (z(x(t)) - z(x)) e^{i \theta(t)} v(t) + e^{i \theta(t)} \{ H^{(0)}, v(t) \}.
\]
Thus
\[
    \{ H^{(0)}, e^{i \theta(t)} v(t) \} + i (H^{(1)} + z(x)) e^{i \theta(t)} v(t) =
        e^{i \theta(t)} \left( i z(x(t)) v(t) + \{ H^{(0)}, v(t) \} + i H^{(1)} v(t) \right).
\]
The proposition follows.
\end{proof}

\subsection{The evolution of pure Lagrangian states}
\label{sec:The evolution of Lagrangian states}
Define the evolution of Lagrangian density matrices as
\begin{equation}
    \mathbb{E}_{\rho_B^{L(t)}(t)}(s) = \mathbb{E}_{\rho_B^L}(s(t)),
\end{equation}
where $s(t)$ is the evolution of the observable $s$.
According to the theorem \ref{thr:s_x(t)}
\begin{equation}
    s_x(t) =
        \mathcal{U}_x(t) h_{x, x(t)} s_{x(t)} h_{x(t), x} \,\mathcal{U}_x(t)^{-1}.
\end{equation}
For a pure hybrid Lagrangian state
\begin{multline}\nonumber
    \mathbb{E}_{\rho_B^{L(t)}(t)}(s) =
        \int_B \left(\varphi(b),\, \mathcal{U}_{x_{L, b}}(t) h_{x_{L, b}, x(t)} s_{x(t)}
            h_{x(t), x_{L, b}} \left(\mathcal{U}_{x_{L, b}}(t)\right)^{-1}
                \varphi(b) \right)_{x_{L, b}} d^n b = \\ =
        \int_B \left(h_{x(t), x_{L, b}} \left(\mathcal{U}_{x_{L, b}}(t)\right)^{-1} \varphi(b),\,
            s_{x(t)} h_{x(t), x_{L, b}} \left(\mathcal{U}_{x_{L, b}}(t)\right)^{-1}
                \varphi(b) \right)_{x(t)} d^n b = \\ =
        \int_B \left(\varphi(b(t), t),\, s_{x(t)} \varphi(b(t), t) \right)_{x(t)} d^n b.
\end{multline}
The endpoint $x(t)$ coincides with $L(t) \cap \pi^{-1}(b(t))$ where $L(t)$ is the evolution
of Lagrangian section $L \subset \mathcal{M}$ along classical trajectories. The scalar
product $(\cdot, \cdot)_{x(t)}$ is the scalar product in $V_{x(t)}$.
Note, that $\varphi(b(t), t) \in V_{x(t)}$ is defined as
\begin{equation}
    \varphi(b(t), t) = h_{x(t), x_{L, b}} \left(\,\mathcal{U}_{x_{L, b}}(t)\right)^{-1} \varphi(b) =
    \left(\,\mathcal{U}_{x(t)}(t)\right)^{-1} \varphi(b) \in \mathcal{H}_B^{L(t)}.
\end{equation}
For $\varphi(b(t), t)$ we have
\begin{equation}
    \frac{d}{dt} \varphi(b(t), t) = - i H^{(1)}(x(t)) \varphi(b(t), t).
\end{equation}

\subsection{Correlation functions}
Hybrid systems are not conservative, so the natural physical quantities that characterize quantum
dynamics are time-dependent correlation functions.

Time dependent correlation functions for quantum observables $s^{(1)}, \ldots, s^{(n)}$ in
the state with the density matrix $\rho^{(n)}$ are
\begin{equation}
    \mathbb{E}_\rho \left( s^{(1)}(t_1) \otimes \ldots \otimes s^{(n)}(t_n) \right) =
    \int_\mathcal{M} \mathrm{Tr}_{V_x^{\otimes{n}}} \left(
        \rho_x^{(n)} \big( s_x^{(1)}(t_1) \otimes \ldots \otimes s_x^{(n)}(t_n) \big)
    \right) \omega_x^n.
\end{equation}
Here $\rho_x^{(n)}$ is the section of $S^{\otimes n}$ (the $n$-th power of the state bundle
of local quantum states $S$).

In general, eigenvalues of $H^{(1)}(x(t))$ are time-dependent,
so are the spectral functions of $H^{(1)}(x(t))$.
However, if the trajectory is periodic, we have monodromy operators
\begin{equation}
    M_{x_0} = \mathrm{Pexp} \left( i \int_0^T H^{(1)}\big(x(t)\big) dt \right)
        \in \mathrm{End}(V_{x_{0}}),
\end{equation}
where $\{x(t)\}$ is a $T$-periodic trajectory with $x_{0} = x(0) = x(T)$. 
The spectrum of these monodromy operators is similar to the Bloch spectrum 
for periodic potentials.

In the example of the discrete Sine-Gordon model \cite{BBR} such a monodromy operator for a minimal
periodic orbit is the transfer matrix for the Chiral Potts model.

Computation of correlation functions and spectra of monodromy operators in hybrid
integrable systems is an interesting problem, but we will not focus on it here.

\subsection{The hybrid evolution and the deformation quantization}
Let $A_\hbar$ be a flat deformation family of $A_0$ as in section 
\ref{sec: deformation_quantization}.  Let $A_\hbar$ be the algebra of quantum observables
of a quantum system.
Its evolution is determined by the choice of the Hamiltonian $\hat{H} \in A_\hbar$. 
In the Heisenberg picture, an observable $\hat{a} \in A_\hbar$ evolves as 
$\hat{a} \mapsto \hat{a}(t)$ where\
\begin{equation}\label{eq:Heis_dyn}
    -i \hbar \frac{\partial \hat{a}(t)}{\partial t} = 
        [\hat{H}, \hat{a}(t)], \quad \hat{a}(0) = \hat{a}.
\end{equation}
Note that evolutions defined by $\hat{H}$ and $\hat{H} + \hat{z}$, $\hat{z} \in Z(A_\hbar)$,
are identical.

Assume that $\hat{H} \in A_\hbar$ be an element of the deformed algebra such that
\begin{equation}
    \phi_\hbar(\hat{H}) = H^{(0)} + \hbar H^{(1)} + O(\hbar^2), \quad
        H^{(0)} \in Z(A_0), \ H^{(1)} \in A_0,
\end{equation}
where $\phi_\hbar \colon A_\hbar \to A_0$ is a linear isomorphism as in section
\ref{sec: deformation_quantization}.
We will call such element $\hat{H}$ \textbf{semiclassically hybrid}.

Semiclassically hybrid elements form a subalgebra $A_\hbar^{\mathrm{SH}}$ in $A_\hbar$.
Indeed, if $\hat{F}$ and $\hat{G}$ are two semiclassically
hybrid operators
\begin{gather}
    \phi_\hbar(\hat{F}) = F^{(0)} + \hbar F^{(1)} + O(\hbar^2), \\
    \phi_\hbar(\hat{G}) = G^{(0)} + \hbar G^{(1)} + O(\hbar^2), \\
    F^{(0)}, G^{(0)} \in Z(A_0), \quad
    F^{(1)}, G^{(1)} \in A_0,
\end{gather}
then their sum and product are also semiclassically hybrid:
\begin{align}
    \phi_\hbar(\hat{F} + \hat{G}) &= 
        (F^{(0)} + G^{(0)}) + \hbar (F^{(1)} + G^{(1)}) + O(\hbar^2), \\
    \phi_\hbar(\hat{F} \hat{G}) &=
        F^{(0)} G^{(0)} + 
            \hbar (F^{(0)} G^{(1)} + G^{(0)} F^{(1)} + m_1(F^{(0)}, G^{(0)})) +
                O(\hbar^2).
\end{align}
It is also clear that the multiplication by a constant preserves
$A_\hbar^{\mathrm{SH}} \subset A_\hbar$.

Note that the definition of semiclassically hybrid elements does not depend on the choice
of the linear isomorphism \(\phi_\hbar \colon A_\hbar \to A_0\). 

Let $a \in A_0$ be the semiclassical limit of $\hat{a} \in A_{\hbar}$ as $\hbar \to 0$
\begin{equation}
    \phi_\hbar(\hat{a}) = a + O(\hbar), \quad  a \in A_0.
\end{equation}
Then as $\hbar \to 0$, the Heisenberg evolution generated by semiclassically hybrid $\hat{H}$ on $A_\hbar$
becomes a split hybrid Heisenberg evolution on $A_0$
\begin{equation}\label{eq: split hyb dyn}
    \frac{\partial a}{\partial t} = \{H^{(0)}, a\} + i \, [H^{(1)}, a], \quad
        a \in A_0.
\end{equation}

\begin{remark} \label{rem: one-time-deformation-quantization-correctness}
Note that the hybrid Heisenberg evolution \eqref{eq: split hyb dyn} does not depend on
a quantization scheme.
Indeed, let $\phi_\hbar, \psi_\hbar \colon A_\hbar \to A_0$ be two linear isomorphisms such that
\begin{equation}
    \phi_\hbar = \eta_\hbar \circ \psi_\hbar, \qquad
        \eta_\hbar \colon A_0 \to A_0, \quad
            \eta_\hbar(a) = a + \sum_{k \ge 1} \hbar^k \eta^{(k)}(a).
\end{equation}
Then
\begin{equation}
    \phi_\hbar(\hat{a}) = a + O(\hbar), \qquad
    \psi_\hbar(\hat{a}) = a + O(\hbar),
    \quad \hat{a} \in A_\hbar, \quad a \in A_0.
\end{equation}
Along with that, the image of a semiclassically hybrid $\hat{H}$ under $\psi_{\hbar}$ is
\begin{equation}
    \psi_\hbar(\hat{H}) = 
        \eta_\hbar^{-1}\big(\phi_\hbar(\hat{H})\big) = 
            H^{(0)} + \hbar \big(H^{(1)} - \eta^{(1)}(H^{(0)})\big) + O(\hbar^2).
\end{equation}

Applying $\phi_{\hbar}, \psi_{\hbar}$ to \eqref{eq:Heis_dyn}, we have two versions of
hybrid Heisenberg dynamics
\begin{align} \label{eq: two hybrid equations}
    \frac{\partial a}{\partial t} &= 
        \{H^{(0)}, a\}^{(\phi)} + i \, [H^{(1)}, a], \\
    \frac{\partial a}{\partial t} &= 
        \{H^{(0)}, a\}^{(\psi)} + i \, \big[H^{(1)} - \eta^{(1)}(H^{(0)}), a\big],
\end{align}
where $\{\cdot, \cdot\}^{(\phi)}$ and $\{\cdot, \cdot\}^{(\psi)}$ are defined as in 
remark \ref{rem: two quantum strusture}.

Taking into account the relation \eqref{eq: psi-phi brackets relation - 1}, we have
\begin{equation}
    \{H^{(0)}, a\}^{(\psi)} + i \, \big[H^{(1)} - \eta^{(1)}(H^{(0)}), a\big] =
    \{H^{(0)}, a\}^{(\phi)} + i \, [H^{(1)}, a].
\end{equation}
That implies that hybrid evolutions \eqref{eq: two hybrid equations} are identical.
\end{remark}

\section{Compatible hybrid multi-time evolutions}
\label{sec:hybrid-integrable-systems}

\subsection{Compatible hybrid multi-time evolutions.} 
Let $A$ be a hybrid algebra of observables and $Z(A)$ be its center. Fix elements
$H^{(0)}_1, \ldots, H^{(0)}_n \in Z(A)$, $H^{(1)}_1, \ldots, H^{(1)}_n \in A$ and 
define differential operators $D_j$ acting on $C(\mathbb{R}^n, A)$ as
\begin{equation}
    D_j \, a(\mathbf{t}) = 
        \frac{\partial a(\mathbf{t})}{\partial t_j} - 
            \{H^{(0)}_j, a(\mathbf{t})\} - i \, [H^{(1)}_j, a(\mathbf{t})], \quad
                j = 1, \ldots, n, \quad 
                \mathbf{t} = (t_1, \ldots, t_n) \in \mathbb{R}^n.
\end{equation}
It is natural to call the multi-time evolution 
\begin{equation}
    D_j \, a(\mathbf{t}) = 0
\end{equation}
of Cauchy data \(a(0) = a\) \textbf{compatible} if 
\begin{equation}
    [D_j, D_k] = 0
\end{equation}
for all \(j, k = 1, \ldots, n\).

\begin{definition} \label{def: integrable_hybrid_Heisenberg}
Let $A$ be a hybrid algebra of observables and $Z(A)$ be its center. 
A \textbf{hybrid multi-time evolution} of an observable $s \in A$ with the classical background
dynamics generated by classical Hamiltonians $H^{(0)}_1, \ldots, H^{(0)}_n \in Z(A)$ and
with quantum Hamiltonians $H^{(1)}_1, \ldots, H^{(1)}_n \in A$ is the solution to the system
of differential equations
\begin{equation}
    \label{eq:evolution-2}
    \frac{\partial s(\mathbf{t})}{\partial t_k} =
        \{H^{(0)}_k, s(\mathbf{t})\} + i \, [H^{(1)}_k, s(\mathbf{t})],
\end{equation}
with the initial condition $s(0) = s$.
\end{definition}

The equations \eqref{eq:evolution-2} are compatible only if the classical Hamiltonians
$H^{(0)}_1, \ldots, H^{(0)}_n$ and the quantum Hamiltonians $H^{(1)}_1, \ldots, H^{(1)}_n$ 
satisfy the following compatibility conditions 
\begin{gather}
    c_{k, l}^{(1)} = \{H^{(0)}_k, H^{(0)}_l\}, \quad 
        \{c_{kl}^{(1)}, a\} = 0 \ \text{for any } a \in A, \\
    \label{eq: H0-H1-compatibility}
    c_{k, l}^{(2)} = \{H^{(0)}_k, H^{(1)}_l\} - \{H^{(0)}_l, H^{(1)}_k\} + 
            i \, [H^{(1)}_k, H^{(1)}_l] - \{H^{(0)}_k, H^{(0)}_l\}_2 \in Z(A).
\end{gather}
If $Z(A) = C(\mathcal{M})$ is the Poisson algebra of functions on a symplectic manifold
$\mathcal{M}$, the Poisson brackets are non-degenerate, and elements from the center of
Poisson algebra $c_{k, l}^{(1)}$ are necessarily constants.

Note that the hybrid evolution defined by $H^{(1)}_1, \ldots, H^{(1)}_n$
and $\tilde{H}^{(1)}_1 = H^{(1)}_1 + z_1, \ldots, \tilde{H}^{(1)}_n = H^{(1)}_n + z_n$, 
where $z_1, \ldots, z_n \in Z(A)$, are identical. The coefficients $c_{k, l}^{(2)}$ changes
as 
\begin{equation}
    \tilde{c}_{k, l}^{(2)} = c_{k, l}^{(2)} + \{H^{(0)}_k, z_l\} - \{H^{(0)}_l, z_k\},
\end{equation}
remaining central: $\tilde{c}_{k, l}^{(2)} \in Z(A)$, if $c_{k, l}^{(2)} \in Z(A)$.

Equation (\ref{eq:evolution-2}) describes the Heisenberg hybrid integrable multi-time evolution
of an observable $s \in A$. Let us define the the Schr\"odinger picture, which describes
the multi-time evolution of vectors.

\begin{definition}
Let $\mathcal{H}$ be a hybrid module over a hybrid algebra of observables $A$ as in section
\ref{sec: representations-hybrid-module}. A \textbf{hybrid multi-time Schr\"odinger dynamics} of a
vector $f \in \mathcal{H}$ on the classical background dynamics generated by Poisson commuting
classical Hamiltonians $H^{(0)}_1, \ldots, H^{(0)}_n \in Z(A)$ and with quantum Hamiltonians
$H^{(1)}_1, \ldots, H^{(1)}_n \in A$ is the solution to the system of differential equations
\begin{equation}
    \label{evolution-3}
    \frac{\partial f(\mathbf{t})}{\partial t_k} =
        -\{H^{(0)}_k, f(\mathbf{t})\} - i H^{(1)}_k f(\mathbf{t}), \quad f(0) = f.
\end{equation}
\end{definition}

Taking into account \eqref{eq: hybrid-module-Leibniz} and \eqref{eq: hybrid-module-Jacobi}, 
it is easy to show that the multi-time evolution \eqref{evolution-3} is compatible if and only if
\begin{gather}
    \label{eq: c_kl-schroedinger}
    c_{k, l}^{(1)} = \{H^{(0)}_k, H^{(0)}_l\}, \quad
        \{c_{k, l}^{(1)}, f\} = 0 \ \text{for all } f \in \mathcal{H}, \\
    \label{eq:zero-curvature_shr}
    c_{k, l}^{(2)} = \{H^{(0)}_k, H^{(1)}_l\} - \{H^{(0)}_l, H^{(1)}_k\} + 
            i \, [H^{(1)}_k, H^{(1)}_l] - \{H^{(0)}_k, H^{(0)}_l\}_2 = 0.
\end{gather}

Note that \eqref{eq: c_kl-schroedinger} and \eqref{eq: hybrid-module-Leibniz} implies that
$\{c_{k, l}^{(1)}, a\} = 0$ for all $a \in A$. Thus, any hybrid Schr\"odinger dynamics
defines a hybrid Heisenberg dynamics of observables from $A = \mathrm{End}(\mathcal{H})$.

The multi-time Schr\"odinger dynamics is invariant under the gauge transformation
\begin{equation}
    f \mapsto e^{i \theta} f, \quad
        H^{(1)}_k \mapsto H^{(1)}_k + i \, \{H^{(0)}_k, \theta\}, \quad 
            \theta \in Z(A).
\end{equation}
Coefficients $c_{k, l}^{(1)}$ \eqref{eq: c_kl-schroedinger} and $c_{k, l}^{(2)}$
\eqref{eq:zero-curvature_shr} do not change under this transformation
\begin{align}
    c_{k, l}^{(1)} &\mapsto c_{k, l}^{(1)}, \\
    c_{k, l}^{(2)} &\mapsto 
        c_{k, l}^{(2)} + 
            i \, \{H^{(0)}_k, \{H^{(0)}_l, \theta\}\} -
                i \, \{H^{(0)}_l, \{H^{(0)}_k, \theta\}\} =
        c_{k, l}^{(2)} + i \, \{\{H^{(0)}_k, H^{(0)}_l\}, \theta\} =
            c_{k, l}^{(2)}.
\end{align}

The Heisenberg dynamics \eqref{eq:evolution-2} can be evaluated in a representation
$\rho \colon A \to \mathrm{End}(\mathcal{H})$. But to define the Schr\"odinger dynamics on vectors
from $\mathcal{H}$ we have to find $\tilde{H}^{(1)}_k$ in \eqref{evolution-3} that satisfy
\eqref{eq:zero-curvature_shr}, such that $\tilde{H}^{(1)}_k - H^{(1)}_k = \Delta_k$ with $\Delta_k$
being in the centralizer $Z\left(\rho(A), \mathrm{End}(\mathcal{H})\right)$ of $\rho(A)$ in
$\mathrm{End}(\mathcal{H})$. 
\footnote{
    If \(\mathcal{H}\) is irreducible, this means that $\Delta_k = g_k \cdot 1$. 
}
Thus, a hybrid Heisenberg dynamics can be lifted to a 
Schr\"odinger dynamics in the representation space $\mathcal{H}$ if
\begin{equation}
    c_{k,l} = \{H^{(0)}_k, \Delta_l\} + \{\Delta_k, H^{(0)}_l\}
\end{equation}
for some $\Delta_k \in Z(\rho(A), \mathrm{End}(\mathcal{H})),\ k = 1, \ldots, n$.

\subsection{Hybrid multi-time dynamics and deformation quantization}

Let $A_\hbar$ be a flat deformation family of $A_0$. Consider a Poisson structure on
$Z(A_0)$ and a hybrid algebra structure on $A_0$ induced by this deformation as in section
\ref{sec: deformation_quantization}. 

Consider the multi-time Heisenberg evolution of $\hat{a} \in A_\hbar$ generated by 
quantum Hamiltonians $\hat{H}_1, \dots, \hat{H}_n \in A_{\hbar}$. 
It is given by the system of differential equations:
\begin{equation} 
    \label{eq: quantum-multi-time-dynamics}
    -i \hbar \frac{\partial \hat{a}(\mathbf{t})}{\partial t_j} = 
        [\hat{H}_j, \hat{a}(\mathbf{t})], \quad \hat{a}(0) = \hat{a}.
\end{equation}
This system is compatible if
\begin{equation} \label{eq: quantum-compatibility-condition}
    [\hat{H}_k, \hat{H}_l] = \hat{c}_{k,l} \in Z(A_\hbar).
\end{equation}

Assume that each Hamiltonian $\hat{H}_{k}$ is semiclassically hybrid
\begin{equation}
    \phi_\hbar \colon \hat{H}_k \mapsto H_k^{(0)} + \hbar H_k^{(1)} + O(\hbar^2), \quad
        H_k^{(0)} \in Z(A_0), \ H_k^{(1)} \in A_0.
\end{equation}
In this case, the semiclassical expansion of $\hat{c}_{k, l}$ starts from the term
of $\hbar^1$-order
\begin{multline} \notag
    \phi_\hbar(\hat{c}_{k, l}) = \phi_\hbar([\hat{H}_k, \hat{H}_l]) = 
        \big[H_k^{(0)} + \hbar H_k^{(1)} + O(\hbar^2), \,
            H_l^{(0)} + \hbar H_l^{(1)} + O(\hbar^2)\big]_* = \\ =
        - i\hbar \{H_k^{(0)}, H_l^{(0)}\} + \hbar^2 \big(
            -i \, \{H_k^{(0)}, H_l^{(1)}\} + i \, \{H_l^{(0)}, H_k^{(1)}\} + 
                [H_k^{(1)}, H_l^{(1)}] + i \, \{H_k^{(0)}, H_l^{(0)}\}_2
        \big) + O(\hbar^3).
\end{multline}

Denote by $c_{k, l}^{(i)}$ the coefficients in the expansion of $\phi_\hbar(\hat{c}_{k, l})$
as $\hbar \to 0$
\begin{equation}
    \label{eq: hbar-limit-c_kl}
    \phi_{\hbar}(\hat{c}_{k, l}) = 
        - i \hbar\, c^{(1)}_{k, l} - i \hbar^{2} c^{(2)}_{k, l} + O(\hbar^{3}).
\end{equation}
Then, since any element 
$\hat{a} \in A_\hbar, \ \phi_\hbar(\hat{a}) = a + \hbar a^{(1)} + O(\hbar^2)$
commutes with $\hat{c}_{k, l}$ in $A_\hbar$
\begin{multline} \notag
    0 = \phi_\hbar([\hat{c}_{k, l}, \hat{a}]) = 
        \big[-i \hbar c_{k, l}^{(1)} - i \hbar^2 c_{k, l}^{(2)} + O(\hbar^3), \, 
            a + \hbar a^{(1)} + O(\hbar^2)
        \big]_* = \\ =
        -i\hbar [c^{(1)}_{k, l}, a] - 
        i\hbar^2 \Big( [c^{(1)}_{k, l}, \, a^{(1)}] + [c^{(2)}_{k,l}, \, a] + 
            m_{1}\big(c^{(1)}_{k, l}, a\big) - m_{1}\big(a, c^{(1)}_{k, l}\big) \Big) + O(\hbar^{3}),
\end{multline}
which implies that $c^{(1)}_{k, l} \in Z(A_{0})$ and
\begin{equation}
    \label{eq: expansion-ckl-central}
    [c^{(2)}_{k, l}, \, a] - i \{c^{(1)}_{k, l}, \, a\} = 0, \quad \text{for all } a \in A_0.
\end{equation}

Therefore, if we assume that $\{c_{k, l}^{(1)}, a\} = 0$ for every $a \in A_0$, 
\eqref{eq: expansion-ckl-central} guarantees $c_{k, l}^{(2)} \in Z(A_0)$. 
Thus, we obtain compatible hybrid multi-time dynamics on $A_0$ with classical Hamiltonians
$H^{(0)}_1, \ldots, H^{(0)}_n$ and quantum Hamiltonians $H^{(1)}_1, \ldots, H^{(1)}_n$
(as in definition \ref{def: integrable_hybrid_Heisenberg}) as the semiclassical limit of the
quantum multi-time dynamics generated by $\hat{H}_1, \ldots, \hat{H}_n$.

\begin{remark}
The hybrid multi-time dynamics obtained from deformation quantization is defined
correctly and does not depend on the choice of the isomorphism \(\phi_\hbar \colon A_\hbar \to A_0\). 
Indeed, let $\psi_\hbar \colon A_\hbar \to A_0$ be another linear isomorphisms such that
\begin{equation}
    \phi_\hbar = \eta_\hbar \circ \psi_\hbar, \qquad
        \eta_\hbar \colon A_0 \to A_0, \quad
            \eta_\hbar(a) = a + \sum_{k \ge 1} \hbar^k \eta^{(k)}(a).
\end{equation}
From remark \ref{rem: one-time-deformation-quantization-correctness} we see that
every dynamics is defined correctly. Let us check that the compatibility conditions
also do not depend on \(\eta^{(1)}\). We have
\begin{gather}
    \label{eq: c_kl_1_psi_deformation}
    c_{k, l}^{(1), (\psi)} = \{H^{(0)}_k, H^{(0)}_l\}, \quad 
        \{c_{k, l}^{(1), (\psi)}, a\}^{(\psi)} = 0, \text{ for all } a \in A_0, \\
    \label{eq: c_kl_2_psi_deformation}
    c_{k, l}^{(2), (\psi)} = 
        \{H^{(0)}_k, H^{(1)}_l - \eta^{(1)}(H^{(0)}_l)\}^{(\psi)} - 
            \{H^{(0)}_l, H^{(1)}_k - \eta^{(1)}(H^{(0)}_k)\}^{(\psi)} + \\ +
            i \, [H^{(1)}_k - \eta^{(1)}(H^{(0)}_k), H^{(1)}_l - \eta^{(1)}(H^{(0)}_l)] -
            \{H^{(0)}_k, H^{(0)}_l\}_2^{(\psi)} \in Z(A_0).
\end{gather}

The coefficients $c_{k, l}^{(1)}$ do not depend on the choice of the isomorphism, so
\eqref{eq: c_kl_1_psi_deformation} follows from the corresponding identity 
for $\phi_\hbar$.

Applying the relations \eqref{eq: psi-phi brackets relations}, 
\eqref{eq: psi-phi brackets relation - 1}, and \eqref{eq: psi-phi brackets relation - 2}
to the right hand side of \eqref{eq: c_kl_2_psi_deformation} we get
\begin{equation}
    c_{k, l}^{(2), (\psi)} = 
        c_{k, l}^{(2), (\phi)} - \eta^{(1)}(\{H^{(0)}_k, H^{(0)}_l\}) =
        c_{k, l}^{(2), (\phi)} - \eta^{(1)}(c_{k, l}^{(1)}).
\end{equation}

Applying the linear isomorphism $\psi_\hbar$ to \eqref{eq: hbar-limit-c_kl}, we obtain
\begin{gather}
    \psi_\hbar(\hat{c}_{k, l}) = 
        -i \hbar c_{k, l}^{(1), (\psi)} - i \hbar^2 c_{k, l}^{(2), (\psi)} + O(\hbar^2), \\
    c_{k, l}^{(1), (\psi)} = c_{k, l}^{(1), (\phi)} = \{H_k^{(0)}, H_l^{(0)}\}, \\
    c_{k, l}^{(2), (\psi)} = c_{k, l}^{(2), (\phi)} - \eta^{(1)}(c_{k, l}^{(1), (\phi)}).
\end{gather}
Now, applying $\phi_\hbar$ to $[\hat{c}_{k, l}, \hat{a}] = 0$ for $\hat{a} \in A_\hbar$, 
and taking into account that $\{c_{k, l}^{(1)}, a\} = 0$ for $a \in A_0$, we get that
$c_{k, l}^{(2), (\psi)} \in Z(A_0)$, so the multi-time dynamics defined via $\psi_\hbar$ is also
compatible.

\end{remark}

\section{Hybrid integrable systems.}

\subsection{Hybrid integrable systems}

Let $\mathcal{M}$ be a symplectic manifold of dimension $2n$. 
A classical integrable system on $\mathcal{M}$ is a Lagrangian fibration
$\pi \colon \mathcal{M} \to B$. Locally, a classical integrable system on \(\mathcal{M}\) is 
defined via $n$ independent Poisson commuting functions $H^{(0)}_1, \ldots, H^{(0)}_n$, i.e.
\begin{equation}
    \{H^{(0)}_i, H^{(0)}_j\} = 0, \quad \forall i, j, \quad
        dH^{(0)}_1 \wedge \ldots \wedge dH^{(0)}_n \not \equiv 0.
\end{equation}
Geometrically, $H^{(0)}_1, \ldots, H^{(0)}_n$ are pullbacks of local coordinates on $B$
\begin{equation}
    \pi^*(b_k) = H^{(0)}_k, \quad k = 1, \ldots, n.
\end{equation}
The corresponding Hamiltonian vector fields $v(H^{(0)}_k)$ are parallel to the fibers
of \(\pi\), i.e., their flow lines do not leave the fiber on which they originated.

Let \(U\) be a coordinate neighborhood on \(B\) with local coordinates \(\{b_i\}\),
\(V\) be a coordinate neighborhood with coordinates \(\{\tilde{b}_i\}\), and 
\(f \colon U \cap V \to U \cap V\) be the transition function, 
\(\tilde{b}_i = f_i(b_1, \ldots, b_n)\). Then, we have Poisson commuting Hamiltonians 
\(H^{(0)}_i = \pi^*(b_i), \ 
\tilde{H}^{(0)}_i = \pi^*(\tilde{b}_i) = f_i(H^{(0)}_1, \ldots, H^{(0)}_n)\). 
Flow lines of \(H^{(0)}_i\) and
\(\tilde{H}^{(0)}_i\) form an affine coordinate system on each fiber of \(\pi^{-1}(U \cup V)\).

Algebraically, a \(2n\)-dimensional classical integrable system is a maximal Poisson commutative
subalgebra $\mathcal{B}$ in the Poisson algebra $\mathcal{P}$ with the trivial center
and of finite even rank \(2n\).

The classical Hamiltonians $H^{(0)}_1, \ldots, H^{(0)}_n$ generate a multi-time
Hamiltonian dynamics on $\mathcal{M}$. The multi-time flow lines
$x(\mathbf{t}) = x(t_1, \ldots, t_n)$ are solutions to the Hamilton's equations
\begin{equation}
    \frac{\partial x(\mathbf{t})}{\partial t_k} = 
        \omega^{-1}\big(dH^{(0)}_k(x(\mathbf{t}))\big).
\end{equation}

We want to define hybrid integrable dynamics on this classical integrable background.

\begin{definition}
Let $A$ be a hybrid algebra of observables with center $Z(A)$ being
the Poisson algebra of rank \(2n\) with the trivial Poisson center. 
Let $\mathcal{B}$ be a maximal Poisson commutative subalgebra in the Poisson algebra $Z(A)$.
A \textbf{hybrid integrable system} on $A$ is a homomorphism of Poisson algebras
$\hat{\lambda} \colon \mathcal{B} \to \mathcal{P}(A) \simeq Z(A) \oplus A/Z(A)$,
\(\hat{\lambda}(H^{(0)}) = (H^{(0)}, \lambda(H^{(0)}))\).
\end{definition}

Thus, for any choice of independent classical Hamiltonians 
$H^{(0)}_1, \ldots, H^{(0)}_n$, defining this classical integrable system, 
we have a set of quantum Hamiltonians (defined up to central elements)
\begin{equation}
    \overline{H^{(1)}_1} = \lambda(H^{(0)}_1), \ \ldots, \ 
        \overline{H^{(1)}_n} = \lambda(H^{(0)}_n) \in A/Z(A).
\end{equation}
We have a homomorphism of Poisson algebras, so the quantum Hamiltonians
satisfy the condition
\begin{equation}
    \label{eq: zero-curvature-hybrid-definition}
    \{H^{(0)}_k, H^{(1)}_l\} - \{H^{(0)}_l, H^{(1)}_k\} + 
        i \, [H^{(1)}_k, H^{(1)}_l] - \{H^{(0)}_k, H^{(0)}_l\}_2 \in Z(A).
\end{equation}
Note that this condition does not depend on the choice of the representatives \(H^{(1)}_j \in A\)
in classes \(\overline{H^{(1)}_j} \in A/Z(A)\).

Similarly to a classical integrable system, which defines the classical multi-time
evolution, a hybrid integrable system also defines the Heisenberg 
multi-time hybrid evolutions \eqref{eq:evolution-2}. 
The condition \eqref{eq: zero-curvature-hybrid-definition} together with 
Poisson commutativity of classical Hamiltonians guarantees that the multi-time hybrid
dynamics is compatible. 

If it is possible to choose the quantum Hamiltonians 
\(H^{(1)}_j \in A\) in classes \(\overline{H^{(1)}_j} \in A/Z(A)\)
in the way that the right-hand side in \eqref{eq: zero-curvature-hybrid-definition}
vanishes, the compatible multi-time Heisenberg dynamics can be represented as the compatible
multi-time Schr\"odindger dynamics in an \(A\)-module \(\mathcal{H}\).

\subsection{Geometric example}
Let \(\mathcal{M}\) be a symplectic manifold of dimension \(2n\), 
\(\pi \colon \mathcal{M} \to B\) be a classical integrable system on \(\mathcal{M}\),
and \(\varphi \colon E \to \mathcal{M}\) be a hybrid bundle of algebras with the
connection \(\alpha\). Assume that for each fiber \(\pi^{-1}(b)\) we have a projectively flat
connection \(\beta\) over \(E|_{\pi^{-1}(b)}\). Then \(\beta - \alpha = i \gamma\),
where \(\gamma\) is a Hermitian one-form on each fiber \(\pi^{-1}(b) \subset \mathcal{M}\) 
with coefficients in \(E|_{\pi^{-1}(b)}\). This one-form defines the mapping 
\begin{equation}
    \lambda \colon C(B) \to \Gamma(\mathcal{M}, E), \quad
    \lambda(f) = \iota_{v(\pi^*(f))} \gamma,
\end{equation}
where \(v(\pi^*(f))\) is a Hamiltonian vector field with the Hamiltonian \(\pi^*(f)\).

In local coordinates, \(\gamma(x) = \sum_i \gamma_i(x) dt^i\) and
\(v(\pi^*(f))(x) = \sum_{ij} (\omega^{-1})^{ij} \frac{\partial f}{\partial b^i}(\pi(x))
\frac{\partial}{\partial t^j}\). Here we assume that \(\{t^i\}\) are affine
coordinates on \(\pi^{-1}(b)\) generated by flow lines of \(H^{(0)}_i = \pi^*(b^i).\)

\begin{theorem}
The mapping \(\lambda \colon C(B) \to \Gamma(\mathcal{M}, E)\) defines a hybrid
integrable system.
\end{theorem}

\begin{proof}
The curvature of connection $\beta$, evaluated on the commuting
Hamiltonian vector fields $v(H^{(0)}_k)$ and $v(H^{(0)}_l)$ has the form
(see remark \ref{rem: two-connections})
\begin{equation} \label{eq: curvature-for-geometric-hybrid}
    F_\beta \big(v(H^{(0)}_k), v(H^{(0)}_l)\big) = 
        \{H^{(0)}_k, H^{(0)}_l\}_2 - 
            \{H^{(0)}_k, H^{(1)}_l\} + \{H^{(0)}_l, H^{(1)}_k\} -
            i \, [H^{(1)}_k, H^{(1)}_l],
\end{equation}
where we have introduced quantum Hamiltonians \(H^{(1)}_j = \lambda(H^{(0)}_j)\).
In the formula \eqref{eq: curvature-for-geometric-hybrid}, 
the first term comes from the curvature of $\alpha|_{\pi^{-1}(b)}$, and
the rest appear when we add a one-form $\gamma$. Therefore, the projectively flatness 
condition \(F_\beta \in Z(A)\) is equivalent to \eqref{eq: zero-curvature-hybrid-definition}.
Thus, \(H^{(0)} \mapsto (H^{(0)}, \overline{H^{(1)}})\) is a homomorphism of Lie algebras
\(C(B) \to \mathcal{P}(\Gamma(\mathcal{M}, E))\).

Since
\begin{equation}
    \lambda(H^{(0)}_j H^{(0)}_k) = \iota_{v(H^{(0)}_j H^{(0)}_k)} \gamma =
        \iota_{H^{(0)}_j v(H^{(0)}_k) + H^{(0)}_k v(H^{(0)}_j)} \gamma =
            H^{(0)}_j \lambda(H^{(0)}_k) + H^{(0)}_k \lambda(H^{(0)}_j),
\end{equation}
it is also a homomorphism of Poisson algebras, i.e., we have defined a hybrid integrable
system on \(A = \Gamma(\mathcal{M}, E)\).
\end{proof}

Note that in this case, it is enough to check that quantum Hamiltonians 
$H^{(1)}_1, \ldots, H^{(1)}_n$ satisfy the condition 
\eqref{eq: zero-curvature-hybrid-definition} for some particular choice of 
independent classical Poisson commuting Hamiltonians $H^{(0)}_1, \ldots, H^{(0)}_n$.
If we choose another set of independent classical Hamiltonians 
$\tilde{H}^{(0)}_1, \ldots, \tilde{H}^{(0)}_n$
\begin{equation}
    \tilde{H}^{(0)}_k = F_k(H^{(0)}_1, \ldots, H^{(0)}_n),
\end{equation}
then the corresponding quantum Hamiltonians $\tilde{H}^{(1)}_1, \ldots, \tilde{H}^{(1)}_n$
are
\begin{equation}
    \tilde{H}^{(1)}_k = \sum_{l = 1}^n \frac{\partial F_k}{\partial H^{(0)}_l} H^{(1)}_l.
\end{equation}

Integrability condition on $\tilde{H}^{(1)}_1, \ldots, \tilde{H}^{(1)}_n$ follows from
the fact that for $A = \Gamma(\mathcal{M}, E)$ Poisson brackets $\{\cdot, \cdot\}$
and brackets $\{\cdot, \cdot\}_2$ are
derivations on $Z(A)$, so
\begin{multline} \notag
    \{\tilde{H}^{(0)}_k, \tilde{H}^{(0)}_l\}_2 - 
        \{\tilde{H}^{(0)}_k, \tilde{H}^{(1)}_l\} + 
            \{\tilde{H}^{(0)}_l, \tilde{H}^{(1)}_l\} -
                i \, [\tilde{H}^{(1)}_k, \tilde{H}^{(1)}_l] = \\ =
    \sum_{m = 1}^n \sum_{p = 1}^n 
        \frac{\partial F_k}{\partial H^{(0)}_m} \frac{\partial F_l}{\partial H^{(0)}_p}
        \left(
            \{H^{(0)}_m, H^{(0)}_p\}_2 - 
                \{H^{(0)}_m, H^{(1)}_p\} + 
                    \{H^{(0)}_p, H^{(1)}_m\} - i \, [H^{(1)}_m, H^{(1)}_p]
        \right) \in Z(A).
\end{multline}

\subsection{The multi-time evolution in Lagrangian modules.}

Let \(A = \Gamma(\mathcal{M}, E)\) be a hybrid algebra represented in a hybrid
module \(\mathcal{H} = \Gamma(\mathcal{M}, V)\). Assume that the connection \(\alpha\)
defining the hybrid algebra structure on \(A\) be flat on Hamiltonian vector fields
corresponding to classical Hamiltonians \(H^{(0)}_1, \ldots, H^{(0)}_n\). 
If the quantum Hamiltonians \(H^{(1)}_1, \ldots, H^{(1)}_n\) satisfy the
zero curvature condition
\begin{equation}
    \{H^{(0)}_k, H^{(1)}_l\} - \{H^{(0)}_l, H^{(1)}_k\} + i [H^{(1)}_k, H^{(1)}_l] = 0,
\end{equation}
then the corresponding hybrid integrable system has integrable Schr\"odinger
representation in a hybrid module \(\mathcal{H}\).

In particular, in this case of flat \(\alpha\), one can choose the quantum Hamiltonians 
$H^{(1)}_k = 0$ for all $k = 1, \ldots, n$. Both hybrid Heisenberg and Schr\"odinger
dynamics in this case are just the lifts of the Hamiltonian dynamics generated by
$\{H^{(0)}_k\}$ to sections of $E$ and $V$ respectively.

Now let us show that in this situation the Schr\"odinger picture of hybrid integrable evolution restricts
to Lagrangian module $\mathcal{H}^L_B$.

Indeed, let $F(\mathbf{t}, x) = h_{x, x(\mathbf{t})} f(x(\mathbf{t}), \mathbf{t})$,
where $x(\mathbf{0}) = x$, $x(\mathbf{t})$ is the multi-time evolution generated by $\{H^{(0)}_k\}$,
\(h_{x, x(\mathbf{t})}\) is the holonomy of \(\alpha\) for any path connecting \(x\) and \(x(\mathbf{t})\)
(it does not depend on the path because the connection \(\alpha\) is flat),
and $f(\mathbf{t})$ is the multi-time evolution (\ref{evolution-3}) in $\mathcal{H}$. 
For $F(\mathbf{t}, x)$ we have (compare with section \ref{sec:The evolution of hybrid states})
\begin{equation}
    \label{evolution-4}
    \frac{\partial F(\mathbf{t}, x)}{\partial t_k} =
        -i h_{x, x(\mathbf{t})} H^{(1)}_k(x(\mathbf{t})) h_{x(\mathbf{t}), x} F(\mathbf{t}, x),
\end{equation}
here $F(\mathbf{t}, x) \in V_{x(\mathbf{t})}$, where $V_x$ is the fiber of $V$ over $x$.

Differential equations (\ref{evolution-4}) also define the dynamics of Lagrangian states similar
to the "one time" dynamic described in section \ref{sec:The evolution of Lagrangian states}.
Let $\phi_\mathbf{t} \colon x \mapsto x(\mathbf{t})$ where $x = x(0)$, be the multi-time
evolution on $\mathcal{M}$.
Let $L(\mathbf{t}) = \phi_{\mathbf{t}}(L)$ be the multi-time evolution of the Lagrangian subspace $L$.
The multi-time integrable evolution of vectors in $\mathcal{H}^L_B$ is a family
$\varphi(\mathbf{t}, b) \in \mathcal{H}^{L(\mathbf{t})}_B$ given by solutions to:
\begin{equation}
    \frac{\partial \varphi(\mathbf{t}, b)}{\partial t_k} =
        -i h_{x, x(\mathbf{t})} H^{(1)}_k(x(\mathbf{t})) h_{x(\mathbf{t}), x} \varphi(\mathbf{t}, b),
            \quad \varphi(0, b) \in \mathcal{H}^L_B.
\end{equation}
Here $x(\mathbf{t})$ is a multi-time evolution connecting $L$ and $\pi^{-1}(b)$ in times
$\mathbf{t}$, i.e. $x(\mathbf{t}) \in L(\mathbf{t}) \cap \pi^{-1}(b)$,
and $\varphi(\mathbf{t}, b) \in V_{x(\mathbf{t})}$.

In the important particular example where the bundles \(E\) and \(V\) are trivial
and connection \(\alpha = 0\), the multitime evolution \eqref{evolution-4} will have the form
\begin{equation}
    \frac{\partial F(\mathbf{t}, x)}{\partial t_k} =
        -i H^{(1)}_k(x(\mathbf{t})) F(\mathbf{t}, x),
\end{equation}
and the compatibility condition for this multi-time flow is given by zero-curvature condition
\begin{equation} \label{eq: zero-curvature-for-M-operators-x(t)}
    \left[
        \frac{\partial}{\partial t_j} + i H^{(1)}_j(x(\mathbf{t})),
            \frac{\partial}{\partial t_k} + i H^{(1)}_k(x(\mathbf{t}))
    \right] = 0.
\end{equation}

\subsection{Hybrid integrable systems and deformation quantization}
Now let us describe how hybrid integrable systems appear naturally in the context of deformation 
quantization. We will start with an auxiliary lemma.

\begin{lemma} \label{lem: semiclassically-commutative}
Let $\hat{\mathcal{B}}$ be a commutative subalgebra in semiclassically hybrid
subalgebra $A_\hbar^{\mathrm{SH}} \subset A_\hbar$. Define $\mathcal{B}$ as the algebra generated by
classical parts of operators in $\hat{\mathcal{B}}$
\begin{equation}
    \mathcal{B} = 
        \langle F^{(0)} \mid 
            \phi_\hbar(\hat{F}) = F^{(0)} + O(\hbar), \,
            \hat{F} \in \hat{\mathcal{B}} \rangle \subset Z(A_0).
\end{equation}
Then $\mathcal{B}$ is a Poisson commutative subalgebra in $Z(A_0)$.
\end{lemma}

\begin{proof}
Since $\phi_\hbar(\hat{F} + \hat{G}) = F^{(0)} + G^{(0)} + O(\hbar)$ and
$\phi_\hbar(\hat{F} \hat{G}) = F^{(0)} G^{(0)} + O(\hbar)$ for 
$\hat{F}, \hat{G} \in \hat{\mathcal{B}} \subset A_\hbar^{\mathrm{SH}}$, 
$\mathcal{B} \subset Z(A_0)$ is a subalgebra. 

The commutativity of $\hat{F}, \hat{G} \in \hat{\mathcal{B}}$
\begin{equation}
    [\hat{F}, \hat{G}] = 0
\end{equation}
implies that
\begin{equation}
    \phi_\hbar([\hat{F}, \hat{G}]) = 0 = 
        [F^{(0)} + O(\hbar), G^{(0)} + O(\hbar)]_* = 
            -i \hbar \{F^{(0)}, G^{(0)}\} + O(\hbar^2).
\end{equation}
Then, $\{F^{(0)}, G^{(0)}\} = 0$ for arbitrary $F^{(0)}, G^{(0)} \in \mathcal{B}$,
which completes the proof.
\end{proof}

\begin{theorem} \label{th: deformation-to-hybrid-integrable}
Let \(A_0\) be a hybrid algebra, and assume that $Z(A_0)$ is the Poisson algebra
of rank \(2n\) with the trivial Poisson center. 
Let $\hat{\mathcal{B}}$ be a commutative subalgebra in
$A_\hbar^{\mathrm{SH}}$ such that $\mathcal{B}$ is the maximal Poisson commutative
subalgebra in $Z(A_0)$, defining a classical integrable system. 
Then $\hat{\mathcal{B}}$ defines a hybrid integrable system on $A_0$.
\end{theorem}

\begin{proof}
Lemma \ref{lem: semiclassically-commutative} shows that $\mathcal{B}$ is Poisson commutative 
subalgebra in $Z(A_0) = C(\mathcal{M})$. It defines a classical integrable system on 
$\mathcal{M}$, then we can choose independent classical Hamiltonians 
$H^{(0)}_1, \ldots, H^{(0)}_n$ defining this classical integrable system.
Choose any quantum Hamiltonians $\hat{H}_1, \ldots, \hat{H}_n \in \hat{\mathcal{B}}$, 
such that the classical Hamiltonians $H^{(0)}_1, \ldots, H^{(0)}_n$ are their classical counterparts
\begin{equation}
    \phi_\hbar(\hat{H}_j) = H^{(0)}_j + O(\hbar).
\end{equation}

The elements $\hat{H}_1, \ldots, \hat{H}_n \in \hat{\mathcal{B}}$ generate multi-time 
flow on $A_\hbar$
\begin{equation}\label{eq:Heis_dyn_mult}
    -i \hbar \frac{\partial \hat{a}(\mathbf{t})}{\partial t_k} =
        [\hat{H}_k, \hat{a}(\mathbf{t})], \quad \hat{a}(\mathbf{t}) \in A_\hbar.
\end{equation}
The commutativity of $\hat{H}_k$
\begin{equation}
    \label{quantum-commute}
    [\hat{H}_k, \hat{H}_l] = 0.
\end{equation}
is the compatibility condition of these time flows. 

Quantum Hamiltonians $\hat{H}_1, \ldots, \hat{H}_n$ are quasiclassically hybrid
\begin{equation}
    \phi_\hbar \colon \hat{H}_k \mapsto H_k^{(0)} + \hbar H_k^{(1)} + O(\hbar^2), \quad
        H_k^{(0)} \in Z(A_0), \ H_k^{(1)} \in A_0,
\end{equation}
thus, in the limit $\hbar \to 0$, they define a quantum Hamiltonians $H^{(1)}_k$ for 
the classical Hamiltonians $H^{(0)}_k$, and the equations \eqref{eq:Heis_dyn_mult} become
the multi-time hybrid Heisenberg dynamics on $A_0$
\begin{equation} \label{eq: multi-time_dynamics_from_deformation}
    \frac{\partial a(\mathbf{t})}{\partial t_k} =
        \{H_k^{(0)}, a(\mathbf{t})\} + i [H_k^{(1)}, a(\mathbf{t})], \quad a(\mathbf{t}) \in A_0.
\end{equation}
So we have a mapping \(H^{(0)}_j \mapsto \lambda(H^{(0)}_j) = H^{(1)}_j\) defined on the
independent set of classical Hamiltonians \(H^{(0)}_1, \ldots, H^{(0)}_n\). The map
\(H^{(0)}_j \mapsto \hat{\lambda}(H^{(0)}_j) = (H^{(0)}_j, \overline{\lambda(H^{(0)}_j)})\)
defines a homomorphism of commutative algebras \(\mathcal{B}\) and \(\mathcal{P}(A)\) (with the
multiplication rule \eqref{eq: P(A)-deformation-multiplication}), because
\begin{gather}
    \phi_\hbar(\hat{H}_j + \hat{H}_k) = 
        (H^{(0)}_j + H^{(0)}_k) + \hbar (H^{(1)}_j + H^{(1)}_k) + O(\hbar^2)
        \ \Rightarrow \ 
        \lambda(H^{(0)}_j + H^{(0)}_k) = H^{(1)}_j + H^{(1)}_k, \\
    \phi_\hbar(\hat{H}_j \hat{H}_k) = H^{(0)}_j H^{(0)}_k +
        \hbar (H^{(0)}_j H^{(1)}_k + H^{(0)}_k H^{(1)}_j + m_1(H^{(0)}_j, H^{(0)}_k)) + O(\hbar^2).
\end{gather}

Apply $\phi_\hbar$ to the condition \eqref{quantum-commute} and rewrite it via $*$-commutator
\eqref{eq: star-commutator}
\begin{equation}
    \phi_{\hbar} \big( \big[ \hat{H}_j, \hat{H}_k \big] \big) = 0 =
    \big[ H_j^{(0)} + \hbar H_j^{(1)} + O(\hbar^2), \
        H_k^{(0)} + \hbar H_k^{(1)} + O(\hbar^2) \big]_*.
\end{equation}
Expand the RHS in $\hbar$, this leads to a sequence of conditions:
\begin{align}
    \label{eq:comm1}
    \{H_k^{(0)}, H_l^{(0)}\} &= 0, \\
    \label{eq:flat-conn-deformation-quantization}
    \{H_k^{(0)}, H_l^{(1)}\} - \{H_l^{(0)}, H_k^{(1)}\} &+
        i [H_k^{(1)}, H_l^{(1)}] - \{H^{(0)}_k, H^{(0)}_l\}_2 = 0.
\end{align}
The condition \eqref{eq:comm1} together with independence of $H^{(0)}_k$ defines a classical
integrable system on $\mathcal{M}$. The identity \eqref{eq:flat-conn-deformation-quantization}
is the zero curvature condition \eqref{eq: zero-curvature-hybrid-definition}, which
guarantees that the map \(\hat{\lambda} \colon \mathcal{B} \to \mathcal{P}(A_0)\) 
is the homomorphism of Poisson algebras, and the corresponding hybrid system is integrable.

Note that \eqref{eq:flat-conn-deformation-quantization} also guarantees that the Heisenberg
dynamics \eqref{eq: multi-time_dynamics_from_deformation} can be evaluated and gives 
Schr\"odinger dynamics in representation.
\end{proof}

Note that if we choose another set of classical Hamiltonians
\begin{equation}
    \tilde{H}^{(0)}_k = F_k(H^{(0)}_1, \ldots, H^{(0)}_n),
\end{equation}
then the elements $F_k(\hat{H}_1, \ldots, \hat{H}_n)$ will be the elements in $\hat{\mathcal{B}}$,
corresponding to $\tilde{H}^{(0)}_k$. Since $\hat{\mathcal{B}}$ is the subalgebra in the 
algebra of semiclassically hybrid elements $A_\hbar^{\mathrm{SH}}$, $\phi_\hbar$ gives the
set of quantum Hamiltonians $\tilde{H}^{(1)}_k$ for classical Hamiltonians $\tilde{H}^{(0)}_k$
\begin{equation}
    F_k(\hat{H}_1, \ldots, \hat{H}_n) = F_k(H^{(0)}_1, \ldots, H^{(0)}_n) +
        \hbar F^{(1)}_k (H^{(0)}_1, \ldots, H^{(0)}_n, H^{(1)}_1, \ldots, H^{(1)}_n) + O(\hbar^2) =
        \tilde{H}^{(0)}_k + \hbar \tilde{H}^{(1)}_k + O(\hbar^2).
\end{equation}

For example, if we choose
\begin{equation}
    \tilde{H}^{(0)}_k = (H^{(0)}_k)^2,
\end{equation}
then
\begin{equation}
    \phi_\hbar((\hat{H}_k)^2) = (H^{(0)}_k)^2 + 
        \hbar \big(2 H^{(0)}_k H^{(1)}_k + m_1(H^{(0)}_k, H^{(0)}_k) \big) + O(\hbar^2).
\end{equation}
Therefore, the corresponding quantum Hamiltonian $\tilde{H}^{(1)}_k$ will be
\begin{equation}
    \tilde{H}^{(1)}_k = 2 H^{(0)}_k H^{(1)}_k + m_1(H^{(0)}_k, H^{(0)}_k).
\end{equation}

\section{The semiclassical asymptotic of a hybrid matrix Schr\"odinger equation} \label{Matrix}

\subsection{The nonstationary semiclassical asymptotic.}
The goal of this section is to describe semiclassical solutions to
the non-stationary matrix-valued Schr\"odinger equation when quantum Hamiltonian is
semiclassically proportional to the identity matrix.
The results of this section are contained in \cite{Mas, MF} where
they appear as part of a more general theory. See also \cite{BDT} where a related problem for infinite-dimensional fibers was addressed.

Consider a quantum mechanical system with the quantum algebra of observables being
$\hbar$-differential operators $D(\mathbb{R}^n, \mathrm{End}(L))$ with values in
$\mathrm{End}(L)$ where $L$ is a Hilbert space.
Elements of this algebra are differential operators of the form
$P(-i \hbar \tfrac{\partial}{\partial q}, q)$
\footnote{Here we use Weyl ordering.}
with coefficients being $\mathrm{End}(L)$-valued function on $\mathbb{R}^n$.
Here we assume that $L$ is $\mathbb{C}^N$ with the standard Hermitian structure.
\footnote{
In a more general case, one can consider a nontrivial vector bundle $V$.
}

Assume that as $\hbar \to 0$ the Hamiltonian of the system has the following structure
\begin{equation}\label{hyb-Sch}
    \hat{H} = H^{(0)}(p, q) I + \hbar H^{(1)}(p, q) + O(\hbar^2),
\end{equation}
where $H^{(0)} \in  C^\infty_{pol}(T^*\mathbb{R}^n)$ is the symbol of $\hat{H}$,
$I$ is the identity operator in $\mathbb{C}^N$ and
$H^{(1)}$ is a matrix-valued function on $T^*\mathbb{R}^n$,
i.e. $\hat{H}_\hbar$ is semiclassically hybrid.
It defines a hybrid integrable system with the bundle
of hybrid observables
$E = T^* \mathbb{R}^n \times \mathrm{End}(L)$
with a trivial flat connection.

Let us describe semiclassical solutions to the Schr\"odinger equation
\begin{equation}\label{nse-M}
    i \hbar \frac{\partial}{\partial t} \psi(t, q) = \hat{H} \psi(t, q),
   \end{equation}
with initial conditions
\begin{equation}\label{nse-in}
    \psi(0, q) = e^{\frac{i}{\hbar} f(q)} \varphi(q).
\end{equation}

Let $\phi_t \colon T^*Q \to T^*Q$ be the time evolution
generated by $H^{(0)} = \sum_{k = 0}^n A_k(q) p^k \in C(T^*Q)$. It acts as
$\phi_t \colon x \mapsto x(t)$ where $x(t)$ is the time evolution, i.e. the solution to
Hamilton's equations for $H_0$ with $x(0) = x$.

For a smooth function $f \colon \mathcal{M} \to \mathbb{R}$ define the Lagrangian submanifold
$L_f = \{(p, q) \mid p = df(q)\}$.  It remains Lagrangian with the evolution.
Assume that the Lagrangian submanifolds $\phi_t(L_f)$ and $T^*_qQ$ intersect transversally
over finitely many points.

Let $\sigma_\alpha = \{p^\alpha(\tau), q^\alpha(\tau)\}_{\tau = 0}^t$ be classical trajectories
connecting  Lagrangian submanifolds $L_f$ and $T^*_qQ$ in time $t$. They correspond to intersection
points $\phi_t(L_f) \cap T^*_q Q$. Denote by $q^\alpha_0(q, t) = q^\alpha(0) \in L_f$
initial points of these trajectories.

For a parametrized path
$\sigma \colon [0, t] \to T^*Q, \ \tau \mapsto (p(\tau), q(\tau)), \ 0 \le \tau \le t$
we have the Hamilton--Jacobi action
\begin{equation}\label{HJ}
    S[\sigma] =
        \int_0^t \big(p(\tau) \dot{q}(\tau) - H^{(0)}(p(\tau), q(\tau))\big) d\tau + f(q(0)).
\end{equation}

Fix $q^\alpha(0) = q_0$ in the trajectory $\sigma_\alpha = \{p^\alpha(\tau), q^\alpha(\tau)\}_{\tau = 0}^t$
and denote by $\Psi^\alpha(q_0, t)$ the solution to the vector-valued ODE
\begin{equation}
    \frac{d}{dt} \Psi^\alpha(q_0, t) =
        -i H^{(1)}\left(p^\alpha(t), q^\alpha(t) \right) \Psi^\alpha(q_0, t),
\end{equation}
with the initial condition
\begin{equation}
    \Psi^\alpha(q_0, 0) = \varphi(q_0).
\end{equation}

\begin{theorem}\label{thr:asympwavefun}
\footnote{This theorem can be found in \cite{Do, MF}. We outline the proof see Appendix \ref{sec:MatrixA}.}
As $\hbar \to 0$, the solution to (\ref{nse-M}) with the initial condition (\ref{nse-in})
has the following asymptotic
\begin{equation}\label{ss-as}
    \psi(q, t) = \sum_\alpha
          D^\alpha(q, t) \exp\left(\frac{i S^\alpha(q, t)}{\hbar} + i \frac{\pi}{4} \mu_\alpha\right)
           \Psi^\alpha(q_0^\alpha(q, t),t) (1 + O(\hbar)),
\end{equation}
where $S^\alpha(q, t)$ is the critical value of the modified Hamilton--Jacobi action on the
trajectory $\sigma_\alpha$, connecting $L_f$ and $T^*_q Q$ in time $t$, $q_0^\alpha(q, t)$ is
the initial point of this trajectory,
$D^\alpha(q, t) = \left| \frac{\partial q^\alpha_0(q, t)}{\partial q} \right|^{\frac{1}{2}}$,
$\Psi^\alpha(q, t)$ is defined above,
and $\mu_\alpha$ is the Morse index of the trajectory $\sigma_\alpha$, also known as the Maslov index.
\end{theorem}

\subsection{The semiclassical dynamics of hybrid Schr\"odinger integrable systems}
\label{sec:Semiclassical dynamics of matrix hybrid}

Now assume that we have $n$ commuting matrix-valued differential operators
on an $n$-dimensional manifold $Q$ of the form (\ref{hyb-Sch}):
\begin{equation}
    \widehat{H}_k = H^{(0)}_k(p, q) I + \hbar H^{(1)}_k(p, q) + O(\hbar^2), \quad k = 1, \ldots, n,
\end{equation}
i.e. we have a semiclassically hybrid integrable system.

The multi-time evolution $\psi\mapsto \psi(\mathbf{t})$ is a solution to the system of equations
\begin{equation}
    \label{mt_schroedinger}
    i \hbar \frac{\partial \psi(\mathbf{t})}{\partial t_k} =
        \widehat{H}_k \psi(\mathbf{t}), \quad \psi(0) = \psi,
\end{equation}
where $\mathbf{t} = (t_1, \dots, t_n)$.

Let us describe the semiclassical behavior of solutions to the multi-time nonstationary
equation (\ref{mt_schroedinger}) with initial conditions
\begin{equation}
    \label{mt_initial_conditions}
    \psi_0(q) = e^{i \frac{f(q)}{\hbar}} \varphi(q).
\end{equation}

As before let $L_f = \{(p = df(q), q)\} \subset T^* Q$ be the Lagrangian submanifold
which is the graph of the function $df \colon Q \to T^* Q$ and let
$\phi_\mathbf{t} \colon T^* Q \to T^* Q$ be the multi-time evolution generated by
Poisson commuting Hamiltonians $H^{(0)}_k(p, q)$. The image $\phi_\mathbf{t}(L_f)$ with
respect to the multi-time evolution remains Lagrangian submanifold and for generic
$q$ the intersection $\phi_\mathbf{t}(L_f) \cap T^*_q Q$ consists of finitely many points.
Preimages of these points in $L_f$ are initial points of multi-time trajectories $\sigma_\alpha$
connecting $L_f$ and $T^*_q Q$ in multi-time $\mathbf{t}$. Denote these points on $L_f$
by $q_0^{\alpha}(q, \mathbf{t})$. The trajectories $\sigma_\alpha$ are critical points
of the multi-time modified Hamilton--Jacobi action $S_\gamma[\sigma] + f(q(0))$
(see Appendix \ref{sec:appendix_mt_HJ} for details). Denote by $S^\alpha(q, \mathbf{t})$
corresponding critical values.
\begin{theorem}
    The solution to (\ref{mt_schroedinger}) with the initial conditions
    (\ref{mt_initial_conditions}) have the following asymptotic when $\hbar \to 0$:
    \begin{equation}
        \psi(q, \mathbf{t}) =
            \sum_\alpha \left|
                \frac{\partial q_0^{\alpha}(q, \mathbf{t})}{\partial q}
            \right|^{\frac{1}{2}}
            \exp\left(
                \frac{i}{\hbar} S^\alpha(q, \mathbf{t}) +
                    \frac{i \pi}{4} \tilde{\mu}_\alpha
            \right)
            \Psi^{\alpha}(q^\alpha_0(q, \mathbf{t}), \mathbf{t}) (1 + O(\hbar)),
    \end{equation}
    where $S^\alpha(q, \mathbf{t})$ and $q^\alpha_0(q, \mathbf{t})$ are as above,
    $\tilde{\mu}_\alpha$ is the multi-time version of the Maslov index,
    $\Psi^{\alpha}(q_0, \mathbf{t})$ is the solution to the multi-time initial
    value problem
\begin{equation}
\label{eq:dynamic reduced wavefunction}
    \frac{\partial \Psi^{\alpha}(q_0, \mathbf{t})}{\partial t_k} =
        -i H^{(1)}_k\left(p^\alpha(\mathbf{t}), q^\alpha(\mathbf{t})
        \right) \Psi^{\alpha}(q_0, \mathbf{t}), \qquad
            \Psi^\alpha(q_0, 0) = \varphi(q_0).
\end{equation}
Here $\{p^\alpha(\mathbf{\tau}), q^\alpha(\mathbf{\tau})\}$ is a multi-time trajectory
with the initial point $(p_0, q_0) \in L_f$.
\end{theorem}

The proof is entirely parallel to the proof of Theorem \ref{thr:asympwavefun}.

\section{Semiclassical asymptotic for integrable quantum spin chain}\label{spin-chains}

\subsection{Yangian type algebras and their classical counterparts}

\subsubsection{}
Assume that we have a collection of vector spaces $\{U_\alpha\}$, that for each pair of these vector spaces we have a family of invertible linear operators
$\{R^{U_\alpha U_\beta}(u)\}$ with $u \in \mathbb{C}$, and  that for each triple $\alpha, \beta, \gamma$ linear operators satisfy the Yang--Baxter relations:
\begin{equation}\label{qybe}
    R^{U_\alpha U_\beta}_{\alpha \beta}(u)
        R^{U_\alpha U_\gamma}_{\alpha \gamma}(u + v)
            R^{U_\beta U_\gamma}_{\beta \gamma}(v) =
    R^{U_\beta U_\gamma}_{\beta \gamma}(v)
        R^{U_\alpha U_\gamma}_{\alpha \gamma}(u + v)
            R^{U_\alpha U_\beta}_{\alpha \beta}(u).
\end{equation}
Here, as usual, operators act in $U_\alpha \otimes U_\beta \otimes U_\gamma$ and subindices
show in which factors of the tensor product the linear operator acts non-trivially.

Assume that quantum $R$-matrices $R^{U V}(u) \in \mathrm{End}(U \otimes V)$ are semiclassical,
i.e. each of them depends on a parameter $\hbar$ and as $\hbar \to 0$ it has the asymptotic
\begin{equation} \label{cR}
    R^{UV}(u, \hbar) = 1 + i \hbar r^{UV}(u) + O(\hbar^2),
\end{equation}
where $r^{UV}(u)$ is the corresponding classical $r$-matrix. Second order terms in $\hbar$
of (\ref{qybe}) gives the classical Yang--Baxter relations for $r^{UV}(u)$:
\begin{equation}\label{cybe}
    [r^{UV}_{12}(u), r^{UW}_{13}(u + v)] +
        [r^{UV}_{12}(u), r^{VW}_{23}(v)] +
            [r^{UW}_{13}(u + v), r^{VW}_{23}(v)] = 0.
\end{equation}
As in (\ref{qybe}), linear operators act in $U \otimes V \otimes W$.

\subsubsection{}
Define the bialgebra $Y_\hbar(R)$ as follows. The algebra $Y_\hbar(R)$ generated by the
coefficients of generating functions $T^U_{ij}(u)$ 
\footnote{
    At the moment, it is not important exactly how the generating functions are organized, 
    as power series in $u, u^{-1}$, as Laurent power series in $e^u$
    or in terms of elliptic functions, or in some other way.
} 
where $U$ is one of the vector spaces $\{U_\alpha\}$. 
For each pair of vector spaces $U$ and $V$ from our collection there is a relation
\begin{equation}\label{rtt}
    R^{UV}_{12}(u) T^U_1(u + v) T^V_2(v) = T^V_2(v) T^U_1(u + v) R^{UV}_{12}(u).
\end{equation}

Note that one can impose other relations, such as $\det_q(T(u)) = 1$ where $\det_q$ is the
quantum determinant. Under the appropriate assumptions 
$Y_\hbar(R)$ can be a Hopf algebra, but it is not important at the moment.
There are plenty of known examples of such algebras, 
such as Yangians, quantized universal enveloping algebras, elliptic algebras, etc., 
see, for example, \cite{QG1}.

The algebra $Y_0(R)$ is commutative. We assume flatness of the deformation family $Y_\hbar(R)$,
which means that $Y_\hbar(R)$ are all isomorphic to $Y_0(R)$ as topological vector spaces
as in section \ref{sec: deformation_quantization}. Denote by 
$\phi_\hbar \colon Y_\hbar(R) \to Y_0(R)$ such a linear isomorphism.

The elements
\begin{equation}
    t^V(u) = \mathrm{Tr}_V (T^V(u))
\end{equation}
generate a commutative subalgebra in $Y_\hbar(R)$. This is an immediate consequence of relations
(\ref{rtt}).

The bialgebra structure on $Y_\hbar(R)$ is determined by
the action of the comultiplication and counit on generating functions $T^U(u)$:
\begin{equation}
    \Delta T^U(u) = T^U(u) \otimes T^U(u), \quad \epsilon(T^U(u))=1.
\end{equation}

\subsubsection{}

The algebra $Y_0(R)$ is commutative, generated by the coefficients of generating functions
$T^{U_\alpha}_{ij}(u)$. To distinguish generating functions
for $Y_\hbar(R)$ and for $Y_0(R)$, we denote the latter by $L^{U_\alpha}_{ij}(u)$.
The notation is very standard in classical integrable systems, where $L^U(u)$ plays
the role of the classical Lax operator.

Assume that $Y_\hbar(R)$ admits a PBW type basis, i.e. that symmetrized monomials in 
coefficients of generating functions $T_{ij}(u)$ form a basis. Then, we choose t
he isomorphism $\phi_\hbar \colon Y_\hbar(R) \to Y_0(R)$ that brings symmetrized monomials
in $T$ to monomials in $L$. In particular,
\begin{equation} \label{phi-hbar-for-TT}
    \phi_\hbar(T_a(u)) = L_a(u), \qquad
    \phi_\hbar(T_a(u) T_b(v) + T_b(v) T_a(u)) = 2 L_a(u) L_b(v)
\end{equation}

Expanding the relation (\ref{rtt}) in $\hbar$ and taking into account \eqref{phi-hbar-for-TT}, 
we obtain the following expression for the $*$-commutator of $L_1^U(u)$ and $L_2^V(v)$ 
as $\hbar \to 0$
\begin{equation}
    [L_1^U(u), L_2^V(v)]_* = \phi_\hbar(T_1^U(u) T_2^V(v) - T_2^V(v) T_1^U(u)) =
        -i \hbar \{L_1^U(u), L_2^V(v)\} + \hbar^2 \{L_1^U(u), L_2^V(v)\}_2 + O(\hbar^3).
\end{equation}
Assuming that as $\hbar \to 0$
\begin{equation}
    R^{UV}(u) = 1 + i \hbar r^{UV}(u)  + \hbar^2 s^{UV}(u) + O(\hbar^3),
\end{equation}
we obtain
\begin{equation} \label{lr}
    \{L_1^U(u), L_2^V(v)\} = [r^{UV}(u - v), L_1^U(u) L_2^V(v)]
\end{equation}
and
\begin{equation}\label{eq: br2yangian}
    \{L_1^U(u), L_2^V(v)\}_2 =
        - \Big[\frac{1}{2} \big(r^{UV}(u - v) \big)^2 + s^{UV}(u - v), L_1^U(u) L_2^V(v) \Big].
\end{equation}

\begin{proposition} If the quantum $R$-matrix $R^{UV}$ satisfies the unitarity condition
\begin{equation}\label{eq: R_unitarity}
    R_{12}^{UV}(u - v, \hbar) R_{21}^{VU}(v - u, \hbar) = f(u - v, \hbar) I_{12},
\end{equation}
where $f(u - v, \hbar)$ is a function and $I_{12}$ is the identity operator,
and if the symmetry condition
\begin{equation}\label{eq: bar-symmetry}
 R_{12}^{UV}(u - v, \hbar) = R_{21}^{UV}(v - u, -\hbar),
\end{equation}
holds, then
\begin{equation}
    \frac{1}{2} \big(r^{UV}(u - v) \big)^2 + s^{UV}(u - v) =
        \frac{1}{4} \frac{\partial^2 f(u - v, \hbar)}{\partial \hbar^2} I_{12}.
\end{equation}
\end{proposition}
In this case, the commutator in \eqref{eq: br2yangian} equals zero. 
As a consequence, if the $R$-matrices in the definition of $Y_\hbar(R)$ are given by the
universal $R$-matrix of a Yangian, a quantum affine algebra, or an elliptic quantum group,
\eqref{eq: R_unitarity} and \eqref{eq: bar-symmetry} hold true, thus
\begin{equation} \label{eq: br2yangian0}
    \{L_1^U(u), L_2^V(v)\}_2 = 0.
\end{equation}

\subsubsection{} 

Let $\mathcal{H}_i$, $i = 1, \ldots, N$ be representations of $Y_\hbar$. Denote by
$R^{V \mathcal{H}_i}(u) \in \mathrm{End}(V \otimes \mathcal{H}_i)$ the image of
the generating function $T^V(u)$ in the representation space $\mathcal{H}_i$.
Let  $\mathcal{H} = \mathcal{H}_1 \otimes \ldots \otimes \mathcal{H}_N$ be the tensor product
of these representations. The operators
\begin{equation}\label{qt}
    t^V(u)=\mathrm{Tr}_V (R^{V, \mathcal{H}_1}_{a1}(u) \ldots R^{V, \mathcal{H}_N}_{aN}(u))
\end{equation}
form a commutative family and give many interesting and important examples of quantum integrable
spin chains, see, for example, \cite{BIK}.

Assume that the quantum spin chain \eqref{qt} each representation $\mathcal{H}_i$ is semiclassical. This means that as
$\hbar \to 0$, the family of algebras $\mathrm{End}(\mathcal{H}_i)$  converges, in the appropriate sense, to
the Poisson algebra of the corresponding classical observables. This Poisson algebra is usually a quotient algebra of $Y_0$.
Denote such quotient algebra as $Y_0(s_i)$ and the image of $L^U(u)$ in it by $L^{U, s_i}(u)$.
Then, the classical limit of the generating function (\ref{qt}) is
\begin{equation}\label{ct}
    t_c^V(u) = \mathrm{Tr}_V (L^{V, s_1}_{a1}(u) \ldots L^{V, s_N}_{aN}(u))
        \in Y_0(s_1) \otimes \ldots \otimes Y_0(s_N).
\end{equation}
As a consequence of (\ref{lr}) these generating functions Poisson commute:
\begin{equation}
    \{t_c^V(u), t_c^W(w)\} = 0.
\end{equation}
One should think of these generating functions as Poisson commuting functions on a Poisson manifold,
which is the phase space of the corresponding classical Hamiltonian systems.
This construction is the source of many important examples of integrable systems; see, for example, \cite{FT}.

\subsection{Hybrid spin chains}

Consider a spin chain that has both semiclassical representations $\mathcal{H}_i$ and a
"fixed representation" $U$ for which the $R$-matrices behave as in (\ref{cR}) when 
$\hbar\to 0$. So, the total space of states is $\mathcal{H}_1 \otimes \ldots \otimes \mathcal{H}_N \otimes U$.
The generating function $T^V(v)$ acts on this space as
\begin{equation}
    T^V_a(v, u) = 
        R^{V \mathcal{H}_1}_{a1}(v) \ldots R^{V \mathcal{H}_N}_{aN}(v) R^{VU}_{aq}(v - u),
\end{equation}
where $q$ refers to the last factor in the tensor product.
These operators are known in quantum integrable systems as quantum monodromy matrices \cite{BIK}.

The corresponding transfer matrix 
\footnote{
    In representation theory, it is known as the quantum character of representation $V$
    of $Y_\hbar(R)$ evaluated in $\mathcal{H}$.
} 
is
\begin{equation}
    t^V(v, u) =
        \mathrm{Tr}_V (R^{V \mathcal{H}_1}_{a1}(v) \ldots
            R^{V \mathcal{H}_N}_{aN}(v) R^{VU}_{aq}(v - u)).
\end{equation}
As $\hbar \to 0$ we have
\begin{equation} \label{eq: transfer-matrix-expansion}
    t^V(v, u) = t^V_c(v) + \hbar M^{VU}(v, u) + O(\hbar^2),
\end{equation}
where $t^V_c(u)$ is given in (\ref{ct}), and
\begin{equation}
    M^{VU}(v, u) =
        i\, \mathrm{Tr}_V (L^{V, s_1}_{a1}(v) \ldots L^{V, s_N}_{aN}(v) r^{VU}_{aq}(v - u)).
\end{equation}
If the algebra $Y_0(s_i)$ can be identified with functions on the symplectic space $\mathcal{S}(s_i)$, we have a hybrid system with the bundle
of hybrid observables
$E = \mathcal{S}(s_1) \times \ldots \times \mathcal{S}(s_N) \times \mathrm{End}(U)$.
The curvature of connection $\alpha$ is determined by \eqref{eq: br2yangian}.

Assume that the classical spin chain with Poisson commuting generating functions $t^V_c(v)$
is an integrable system with the phase space $\mathcal{S}(s_1) \times \ldots \times \mathcal{S}(s_N)$.
Let $x(t)$ be the Hamiltonian flow generated by $t^V_c(u)$. The classical $L$-operator evolves as
\begin{align}
    \frac{dL^U_b(v)(x(t))}{dt} &=
        \{t^V(u), L^U_b(v)\} (x(t)) =
            \mathrm{Tr}_a \{L^V_a(u), L^U_b(v)\} (x(t)) = \\ &=
                \mathrm{Tr}_a \big[r^{VU}_{ab}(u - v), L^V_a(u)(x(t)) L^U_b(v)(x(t))\big] = \\ &=
        \big[\mathrm{Tr}_a (r^{VU}_{ab}(u - v) L^V_a(u))(x(t)), L^U_b(v)(x(t))\big] =
            \big[M^{VU}_b(u, v)(x(t)), L^U_b(v)(x(t))\big],
\end{align}
thus, the first-order term in the expansion \eqref{eq: transfer-matrix-expansion}
is the classical $M$-operator, and the equation is the evolution of the Lax
operator $L^U(v)(x)$ with respect to the Hamiltonian flow generated by $t^V_c(u)$.

Applying the theorem \ref{th: deformation-to-hybrid-integrable} to the commutative
family $t^V(v, u)$ \eqref{eq: transfer-matrix-expansion}, we get the hybrid integrable system ---
hybrid spin chain with the generating function of classical Poisson commuting Hamiltonians
$t^V_c(u)$ and quantum Hamiltonians equal to classical $M$-operators $M^{VU}(u, v)$. 
The compatibility condition for this hybrid integrable system can be written in the form of
zero-curvature condition for classical $M$-operators \eqref{eq: zero-curvature-for-M-operators-x(t)}.

Let $x(\mathbf{t})$ be the multi-time Hamiltonian flows generated by $t^{V_j}_c(u_j)$
with $j = 1, \ldots, K$ to generate a complete multi-time flow.
We have
\begin{equation}
    \left[
        \frac{\partial}{\partial t_j} + i M^{V_j U}(u_j, u)(x(\mathbf{t})),
            \frac{\partial}{\partial t_k} + i M^{V_k U}(u_k, u)(x(\mathbf{t}))
    \right] = 0.
\end{equation}

Thus, in this case, the hybrid quantum system is simply the collection of
$M$-operators for the multi-time flow, see for example \cite{DKN, BIK}.

\section{Spin Calogero--Moser--Sutherland system and its hybrid features.}\label{CM}

\subsection{Quantum spin Calogero--Moser--Sutherland system}

\subsubsection{} Quantum spin Calogero--Moser system describes $n$ interacting quantum particles
on a circle with the internal degrees of freedom. Here we will focus on the system with the trigonometric
potential, also known as the Calogero--Moser--Sutherland (CMS) model \cite{Calogero, Moser, Sutherland}.

We will use coordinates $q_i \in \mathbb{R}/\mathbb{Z} \frac{L}{2\pi} \simeq S^1$, 
where $L > 0$ is the length of the physical system.
The Hamiltonian of this model is \cite{HH, HW, MP}
\begin{equation}\label{SCM}
    \hat{H} =
        -\frac{1}{2} \sum_{i = 1}^n \hbar^{2} \frac{\partial^2}{\partial q_i^2} +
        \frac{\pi^2}{2 L^2} \sum_{\substack{i, j = 1 \\ i \ne j}}^n
            \frac{  1 + \hbar P_{ij} }{\sin^2 \frac{\pi (q_i - q_j)}{L}}.
\end{equation}
Here the operator $P_{ij}$ is the spin permutation operator acting in $i$-th and $j$-th spaces.
The Hamiltonian acts on the space $ L_2(\mathbb{R}^n, (\mathbb{C}^{N})^{\otimes n})_\mathrm{sym}$
of functions invariant with respect to the simultaneous permutation of spins and coordinates
$\psi(\ldots, q_i, \ldots, q_j, \ldots) =  P_{ij} \psi(\ldots, q_j, \ldots, q_i, \ldots)$
\footnote{Here we consider bosonic version where functions $\psi$ in the space of states are
invariant with respect to simultaneous permutations of $q_i$ and $q_j$ and the action of
$P_{ij}$ on the spin variable. The analysis of the fermionic case, when $\psi$ is skew-symmetric
with respect to diagonal permutations of coordinates and spins, is completely parallel.
}.
Without loss of generality, rescaling $\hbar$ and $L$ we will fix $L = 2 \pi$.
Introduce new variables $z_j = \exp(i q_j)$.
In terms of $z_i$ the operator (\ref{SCM}) can be written as
\footnote{
This Hamiltonian also appears in the form with an extra coupling constant $\lambda$: the term
$1 + \hbar P_{ij}$ is replaced with $\lambda(\lambda + \hbar P_{ij})$. In this form
Cherednik--Dunkl operators and higher Hamiltonians also contain $\lambda$.
This coupling constant can be removed by rescaling of Planck constant $\hbar \to \lambda \hbar$
together with the rescaling of the Cherednik--Dunkl operators $d_j \to \lambda^{-1} d_2$
and Hamiltonians $\widehat{H}_k \to \lambda^{-k} \widehat{H}_k$.
}
\begin{equation}
    \widehat{H}_{2} =
        \frac{1}{2} \sum_{i = 1}^n
            \left( \hbar\, z_i \frac{\partial}{\partial z_i}\right)^2 -
        \frac{1}{2} \sum_{\substack{i, j = 1 \\ i \ne j}}^n
            \frac{z_i z_j}{(z_i - z_j)^2} (1 + \hbar P_{ij}).
\end{equation}

\subsubsection{} Let us recall how to construct higher commuting Hamiltonians using
Cherednik--Dunkl operators \cite{Cherednik, Dunkl}.

Cherednik--Dunkl operators are differential operators acting on $\mathbb{C}(z_1, \ldots z_n)$
\begin{equation} \label{dunkl_quantum}
    d_j = \hbar\, z_j \frac{\partial}{\partial z_j}   +
        \sum_{i > j} \frac{z_i}{z_i - z_j} K_{ij} -
        \sum_{i < j} \frac{z_j}{z_j - z_i} K_{ij},
\end{equation}
where $K_{ij}$ is coordinate permutation operator $K_{ij} z_{j} = z_{i} K_{ij}$.
They satisfy the following relations
\begin{equation}
    [d_i, d_j] = 0, \qquad
    K_{i, i + 1} d_i = d_{i + 1} K_{i, i + 1} + 1, \qquad
    [d_i, K_{j,j+1}] = 0, \quad i \ne j, j + 1,
\end{equation}
and thus give a representation of the degenerate affine Hecke algebra
{\color{red} \cite{Dr}}.

Commuting Hamiltonians of the quantum spin Calogero--Moser--Sutheland system can be derived as the
action of symmetric polynomials in Cherednik--Dunkl operators \cite{TH}
\begin{equation}
    \label{eq:QCM hamiltonis}
    H_k = \frac{1}{k} \sum_{i = 1}^n d_i^k
\end{equation}
on the space of $(\mathbb{C}^N)^{\otimes n}$-valued symmetric rational functions in $z_i$,
i.e. on
\begin{equation}
    \label{phase-space}
    \mathcal{H} = (\mathbb{C}(z_1, \ldots, z_n) \otimes (\mathbb{C}^N)^{\otimes n} )_\mathrm{sym}.
\end{equation}
We will write $\widehat{H}_k = H_k|_{\mathcal{H}}$.
Note that when we compute the action of $H_k$ on $\mathcal{H}$ we use ordering in which
coordinates $z_{i}$ are on the left, followed by momenta
$\hat{p}_{i} = \hbar z_i \tfrac{\partial}{\partial z_i}$, and permutations $K_{ij}$
are on the right. After that, we use $K_{ij} P_{ij}|_{\mathcal{H}} = 1$
to replace $K$ with $P$ to get the final form of Hamiltonians.

The first nontrivial Hamiltonians are
\begin{align}
    \widehat{H}_1 &= \sum_{i = 1}^n \hat{p}_{i}, \\
    \widehat{H}_2 &= \frac{1}{2} \sum_{i = 1}^n \hat{p}_i^2 -
        \frac{1}{2} \sum_{\substack{i, j = 1\\ i \ne j}}^n
            \frac{z_i z_j}{(z_i -z_j)^2} \left(1 + \hbar P_{ij}\right), \\
    \widehat{H}_3 &= \frac{1}{3} \sum_{i = 1}^n \hat{p}_{i}^{3} -
        \sum_{\substack{i, j = 1\\ i \ne j}}^n
            \frac{z_i z_j \, \left(1 + \hbar P_{ij}\right)}{(z_i - z_j)^2} \hat{p}_i
        - \frac{\hbar}{3} \sum_{\substack{i, j, k = 1 \\ i \ne j \ne k \ne i}}^n
            \frac{z_i z_j z_k \, P_{jk} P_{ij}}{(z_i - z_j)(z_j - z_k)(z_k - z_i)}.
\end{align}

For $\hbar \neq 0$, Dunkl operators are simultaneously diagonalizable on the space
$\mathbb{C}(z_1, \ldots, z_n)$ with simple joint spectrum \cite{TU, Uglov}.
The eigenvectors form an orthogonal basis in $\mathbb{C}(z_1, \ldots, z_n)$
and are called nonsymmetric Jack polynomials.

In the semiclassical limit $\hat{p}_{i}, z_i$ become the coordinate functions on $T^*(S^1)^n$.

\begin{lemma} \label{lem: unity lemma}
In the semiclassical limit Hamiltonians $\widehat{H}_k$ have the form
\begin{equation}\label{scCM}
    \widehat{H}_{k} = H_k^{CM}  + O(\hbar).
\end{equation}
where $H_k^{CM}$ is the corresponding classical Hamiltonian of the "usual" spinless CMS system
multiplied by the identity operator.
\end{lemma}

The proof of this assertion is given in appendix \ref{sec: unity lemma}
\footnote{
For the rational spin CM system, this property is equivalent to the statement of
Lemma 2.2 in \cite{EFMV}. For the trigonometric case, the proof can also be obtained as a
limit from the elliptic case (Prop. 5.1 and eq. (5.20) \cite{Chalykh}).
}.

\subsection{The dynamical Haldane--Shastry model}
The quantum spin CMS system is an example of matrix-valued quantum mechanics from the section
\ref{sec:Semiclassical dynamics of matrix hybrid}. By theorem 
\label{th: deformation-to-hybrid-integrable}, the commutative family $\widehat{H}_k$ 
defines a hybrid integrable system with a bundle of hybrid observables
$E = T^* (S^1)^n \times \mathrm{End}((\mathbb{C}^N)^{\otimes n})$
with the trivial connection $\alpha = 0$.
Passing to the semiclassical limit
$\hbar \to 0$ we have
\begin{equation} \label{sCM-M}
    \widehat{H}_{k} = H^{CM}_{k} + \hbar\, H^{CM, (1)}_k + O(\hbar^{2}).
\end{equation}
Here, as in (\ref{scCM}), $H^{CM}_{k}(p, z)$ are Hamiltonians of classical "spinless" CMS system,
which define the underlying classical dynamics.

For example, the first two of $H^{CM, (1)}_k$ can be computed explicitly
\begin{align}
    H^{CM, (1)}_2 & = -\frac{1}{2} \sum_{\substack{i, j = 1\\ i \ne j}}^n
        \frac{z_i z_j}{(z_i - z_j)^2} P_{ij}, \\
    H^{CM, (1)}_3 & =
        -\sum_{\substack{i, j = 1\\ i \ne j}}^n
            \frac{z_i z_j p_i}{(z_i - z_j)^2} P_{ij} -
        \frac{1}{3} \sum_{\substack{i, j, k = 1 \\ i \ne j \ne k \ne i}}^n
            \frac{z_i z_j z_k P_{jk} P_{ij}}{(z_i - z_j)(z_j - z_k)(z_k - z_i)}.
\end{align}

The classical multi-time evolution is generated by CMS Hamiltonians:
\begin{equation} \label{cHfl}
    \frac{\partial z_j}{\partial t_k} =
        i z_j \frac{\partial H^{CM}_k}{\partial p_j}, \qquad\quad
    \frac{\partial p_j}{\partial t_k} =
       - i z_j \frac{\partial H^{CM}_k}{\partial z_j}.
\end{equation}
Here $p$ and $z$ are natural coordinates on $T^*(S^1)^n$ with Poisson brackets
$\left\{ p_j, z_k \right\} = i \delta_{jk} z_k$.

Let $p(\mathbf{t}), q(\mathbf{t})$ be a multi-time flow on $T^*(S^1)^n$ generated by $H^{CM}_k$.
The first two equations are easy to compute explicitly:
\begin{align}
    \frac{\partial z_j}{\partial t_2} &= i p_j z_j, \qquad
    &&\frac{\partial p_j}{\partial t_2} =
        -i \sum_{\substack{k = 1 \\ k \ne j}}^n
            \frac{z_j z_k (z_j + z_k)}{(z_j - z_k)^3} , \\
    \frac{\partial z_j}{\partial t_3} &=
        i p_j^2 z_j -
            i \sum_{\substack{k = 1 \\ k \ne j}}^n
                \frac{z_j^2 z_k}{(z_j - z_k)^2},
    &&\frac{\partial p_j}{\partial t_3} =
        -i \sum_{\substack{k = 1 \\ k \ne j}}^n
            \frac{z_j z_k (z_j + z_k)}{(z_j - z_k)^3} (p_j + p_k).
\end{align}

The compatibility condition for Hamiltonians of a dynamical Haldane--Shastry model
could be written as the zero curvature equation \eqref{eq: zero-curvature-for-M-operators-x(t)}
\begin{equation}
    \label{eq:dynamical long-range spin chain}
    \left[
        \frac{\partial}{\partial t_k} + i H^{CM, {(1)}}_k(p(\mathbf{t}), q(\mathbf{t})),
            \frac{\partial}{\partial t_l} + i H^{CM, {(1)}}_l(p(\mathbf{t}), q(\mathbf{t}))
    \right] = 0.
\end{equation}

\subsection{The fixed point of the multi-time classical Calogero--Moser--Sutherland dynamics}

It turns out that the multi-time classical CMS dynamics has a fixed point \cite{Ruijsenaars}.

\begin{proposition}\label{prop:freezing point}
Consider the Calogero--Moser--Sutherland system restricted to the
zero-momentum subspace $H^{CM}_1 = \sum_{i = 1}^n p_i = 0$,
$dH^{CM}_1 = \sum_{i = 1}^n dp_i = 0$.
Then the point $x_*$
\begin{equation}
    p_i = 0, \qquad
    z_k = \exp\left(\frac{2 \pi i k}{n}\right)
\end{equation}
is the fixed point of the multi-time CM evolution, i.e.
\begin{equation}
    dH_k^{CM}(x_*) = 0, \qquad k = 2, \ldots, n.
\end{equation}
\end{proposition}
The proof is given in appendix \ref{sec:Proof of fixing}.

As a corollary, we have the commutativity of corresponding $M$-operators
\begin{equation}
    \left[M_k(x_*), M_l(x_*)\right] = 0.
\end{equation}
The operator $M_{2}(x_{*})$ is the Hamiltonian of the of Haldane--Shastry model
\cite{Haldane, Shastry, Inozemtsev}. The operators $M_k(x_*)$ were derived in \cite{BGHP, TH}
by a different method as the higher conservation laws for the Haldane--Shastry Hamiltonian.

The fixed point $x_*$ is known in the physics literature as the freezing point. It first appeared in the
paper \cite{Pol1} where it was shown that quantum spin Calogero--Moser--Sutherland model becomes long-range
spin chain in the strong interaction limit (in our terminology, it corresponds to $\hbar \to 0$).
Some recent results on the correspondence between long-range spin chains and quantum dynamical
systems in their freezing points could be found in \cite{Uglov95, SZ, MZ, LPS, LS}. 

In the forthcoming paper \cite{LMRS}
we will describe explicitly singular Liouville tori in Calogero--Moser--Sutherland models of type $A$, i.e.
invariant tori of dimension $1 \le k \le n - 1$.
An interesting next step is to describe explicitly the
corresponding hybrid dynamics for low-dimensional tori.

\appendix
\section{The semiclassical limit for non-stationary matrix Schr\"odinger equation} \label{sec:MatrixA}

Consider firstly one dimensional case, $n = 1$, i.e. one dimensional matrix
Schr\"odinger equation.
\begin{lemma} Any formally self-adjoint differential operator of degree $n$
of the form
\begin{equation}
    \hat{H}^{(0)} = \sum_{k = 0}^n \alpha_k(q) \hat{p}^k, \quad
    \hat{p} = -i \hbar \frac{\partial}{\partial q},
\end{equation}
with complex-valued coefficients $\alpha_k(q)$  can be written as
\begin{equation}\label{H-0}
    \hat{H}^{(0)} = \sum_{k = 0}^n A_k(q) \hat{p}^k -
        i \hbar \sum_{k = 1}^n \frac{k}{2}
            \frac{\partial A_k(q)}{\partial q} \hat{p}^{k - 1} +
        O(\hbar^2),
\end{equation}
with real-valued coefficients $A_k(q)$.
\end{lemma}

\begin{proof}
Write $\alpha_k(q)$ as
\begin{equation}
    \alpha_k(q) = A_k(q) + i B_k(q),
\end{equation}
with $A_k, B_k$ being real-valued.
Now let us find constraints which self-adjointness $\hat{H}^{(0)} = (\hat{H}^{(0)})^*$ imposes
on the imaginary and the real part of $\alpha_k(q)$. We have:
\begin{equation}
    (\alpha_k(q) \hat{p}^k)^*=  \hat{p}^k \alpha^*_k(q) =\alpha^*_k(q) \hat{p}^k -
        i \hbar k \frac{\partial \alpha^*_k(q)}{\partial q} \hat{p}^{k - 1} +
        O(\hbar^2).
 \end{equation}
Thus, for symmetric Hamiltonians, we should have
\begin{equation}
    \alpha_n(q) = \alpha^*_n(q),
\end{equation}
\begin{equation}
    \alpha_k(q) = \alpha^*_k(q) -
        i \hbar (k + 1) \frac{\partial \alpha^*_{k + 1}(q)}{\partial q} +
        O(\hbar^2), \quad k < n
\end{equation}
and therefore for $k<n$
\begin{equation}
    B_k(q) = -\hbar \frac{k + 1}{2}
        \frac{\partial A_{k + 1}(q)}{\partial q} +
        O(\hbar^2).
\end{equation}
This proves the lemma.
\end{proof}

Let us prove the following technical lemma.

\begin{lemma} \label{1} We have the identity
\begin{equation}
    \hat{p}^k e^{\frac{i}{\hbar} S} \chi =
        \left(\frac{\partial S}{\partial q}\right)^k
            e^{\frac{i}{\hbar} S} \chi -
        i \hbar e^{\frac{i}{\hbar} S} \left(
            k \left(\frac{\partial S}{\partial q}\right)^{k - 1}
                \frac{\partial \chi}{\partial q} +
            \frac{k (k - 1)}{2}
                \left(\frac{\partial S}{\partial q}\right)^{k - 2}
                    \frac{\partial^2 S}{\partial q^2} \chi
        \right) + O(\hbar^2).
\end{equation}
\end{lemma}
\begin{proof}
It is clear that
\begin{equation}
    \left(
        -i \hbar \frac{\partial}{\partial q}
    \right)^k
        e^{\frac{i}{\hbar} S} \chi =
    \left(
        \frac{\partial S}{\partial q}
    \right)^k
        e^{\frac{i}{\hbar} S} \chi +
    \hbar e^{\frac{i}{\hbar} S} \chi^{(k)} +
    O(\hbar^2),
\end{equation}
for some $\chi^{(k)}$.
Differentiating this identity, we obtain a recurrence
\begin{equation}
    \chi^{(k)} =
    \frac{\partial S}{\partial q} \chi^{(k - 1)} -
    i \left(
        \frac{\partial S}{\partial q}
    \right)^{k - 1}
        \frac{\partial \chi}{\partial q} -
    i (k - 1) \left(
        \frac{\partial S}{\partial q}
    \right)^{k - 2}
    \frac{\partial^2 S}{\partial q^2} \chi,
\end{equation}
which gives the desired formula for  $\chi^{(k)}$:
\begin{equation} \label{chi-k}
    \chi^{(k)} =
        -i k \left(\frac{\partial S}{\partial q}\right)^{k - 1}
            \frac{\partial \chi}{\partial q} -
        i \frac{k (k - 1)}{2}
            \left(\frac{\partial S}{\partial q}\right)^{k - 2}
                \frac{\partial^2 S}{\partial q^2} \chi.
\end{equation}

\end{proof}

\begin{proposition} \label{p1}
The action of the Hamiltonian $\hat{H}^{(0)}$ on the the family of functions
$\psi(q) = e^{\frac{i}{\hbar} S(q)} \chi(q)$ is
\begin{gather}
    \hat{H}^{(0)} e^{\frac{i}{\hbar} S} \chi =
        H^{(0)}\left(\frac{\partial S}{\partial q}, q \right)
            e^{\frac{i}{\hbar} S} \chi +
        \hbar \widetilde{H}^{(0)}\left(\frac{\partial S}{\partial q}, q \right)
            e^{\frac{i}{\hbar} S} \chi +  \nonumber
        \hbar e^{\frac{i}{\hbar} S}
            \sum_{k = 0}^n A_k(q) \chi^{(k)} + O(\hbar^2),
\end{gather}
where $\chi^{(k)}$ is given by (\ref{chi-k}), and $H^{(0)}$ and $\widetilde{H}^{(0)}$
are the first two terms in the semiclassical expansion of $\hat{H}^{(0)}$
\begin{equation} \label{p11}
    H^{(0)}\left(\frac{\partial S}{\partial q}, q \right) =
        \sum_{k = 0}^n A_k(q)
            \left(\frac{\partial S}{\partial q}\right)^k,
\end{equation}
\begin{equation}\label{p12}
    \widetilde{H}^{(0)}\left(\frac{\partial S}{\partial q}, q \right) =
        -i \sum_{k = 0}^n \frac{k}{2}
            \frac{\partial A_k(q)}{\partial q}
                \left(\frac{\partial S}{\partial q}\right)^{k - 1}.
\end{equation}
\end{proposition}

\begin{proof}

Lemma \ref{1} implies
\begin{equation}
    \sum_{k = 0}^n A_k(q) \hat{p}^k e^{\frac{i}{\hbar} S} \chi =
    \sum_{k = 0}^n A_k(q)
        \left(\frac{\partial S}{\partial q}\right)^k
            e^{\frac{i}{\hbar} S} \chi +
    \hbar \sum_{k = 0}^n A_k(q)
        e^{\frac{i}{\hbar} S} \chi^{(k)} + O(\hbar^2),
\end{equation}
where $\chi^{(k)}$ is given by (\ref{chi-k}), and
\begin{equation}
    -i \hbar \sum_{k = 1}^n \frac{k}{2}
        \frac{\partial A_k(q)}{\partial q} \hat{p}^{k - 1}
            e^{\frac{i}{\hbar} S} \chi =
    -i \hbar \sum_{k = 1}^n \frac{k}{2}
        \frac{\partial A_k(q)}{\partial q}
            \left(\frac{\partial S}{\partial q}\right)^{k - 1}
            e^{\frac{i}{\hbar} S} \chi + O(\hbar^2).
\end{equation}
From here, using (\ref{H-0}) we immediately obtain (\ref{p11}) and (\ref{p12}).

\end{proof}

From the proposition \ref{p1} we obtain the formula for the action of
$\hat{H} = \hat{H}^{(0)} + \hbar \hat{H}^{(1)} + O(\hbar^2)$ on $e^{\frac{i}{\hbar} S} \chi$:
\begin{multline}\notag
    \hat{H}e^{\frac{i}{\hbar} S} \chi =
    H^{(0)}\left(\frac{\partial S}{\partial q}, q \right)
        e^{\frac{i}{\hbar} S} \chi +
    \hbar \widetilde{H}^{(0)}\left(\frac{\partial S}{\partial q}, q \right)
        e^{\frac{i}{\hbar} S} \chi + \\ +
    \hbar e^{\frac{i}{\hbar} S}
        \sum_{k = 0}^n A_k(q) \chi^{(k)} +
    \hbar H^{(1)}\left(\frac{\partial S}{\partial q}, q\right)
        e^{\frac{i}{\hbar} S} \chi + O(\hbar^2).
\end{multline}

Now let us find the asymptotic of solutions to the nonstationary Schr\"odinger equation (\ref{nse-M}).
Evaluating both sides of (\ref{nse-M}) on functions
$\psi(q, t) = \exp(\frac{i}{\hbar} S) \left(\chi_0 + \hbar \chi_1 + O(\hbar^2)\right)$ as
$\hbar \to 0$ we obtain
\begin{equation}
    i \hbar \frac{\partial}{\partial t}
        e^{\frac{i}{\hbar} S} \left(\chi_0 + \hbar \chi_1 + O(\hbar^2) \right) =
    -\frac{\partial S}{\partial t} e^{\frac{i}{\hbar} S} \chi_0 +
    \hbar e^{\frac{i}{\hbar} S} \left(
        -\frac{\partial S}{\partial t} \chi_1 +
            i \frac{\partial \chi_0}{\partial t}
    \right) + O(\hbar^2),
\end{equation}
and
\begin{multline}\notag
    \hat{H} \psi(t, q) =
    H^{(0)}\left(\frac{\partial S}{\partial q}, q \right)
        e^{\frac{i}{\hbar} S} \chi_0 +
    \hbar H^{(0)}\left(\frac{\partial S}{\partial q}, q \right)
        e^{\frac{i}{\hbar} S} \chi_1 +
    \hbar \widetilde{H}^{(0)}\left(\frac{\partial S}{\partial q}, q \right)
        e^{\frac{i}{\hbar} S} \chi_0 + \\ +
    \hbar e^{\frac{i}{\hbar} S}
        \sum_{k = 0}^n A_k(q) \chi^{(k)}_0 +
    \hbar H^{(1)}\left(\frac{\partial S}{\partial q}, q\right)
        e^{\frac{i}{\hbar} S} \chi_0 + O(\hbar^2).
\end{multline}
Combining terms of degree zero and one, we obtain
\begin{equation} \label{deg0}
    -\frac{\partial S(q, t)}{\partial t} =
        H^{(0)}\left(\frac{\partial S(q, t)}{\partial q}, q \right),
\end{equation}
in degree zero and
\begin{gather} \label{deg1}
    -\frac{\partial S}{\partial t} \chi_1 +
    i \frac{\partial \chi_0}{\partial t} =
    H^{(0)}\left(\frac{\partial S}{\partial q}, q \right) \chi_1 +
    \widetilde{H}^{(0)}\left(\frac{\partial S}{\partial q}, q \right) \chi_0 +
        \sum_{k = 0}^n A_k(q) \chi^{(k)}_0 +
    H^{(1)}\left(\frac{\partial S}{\partial q}, q\right) \chi_0
\end{gather}
in degree one.

The equation (\ref{deg0}) is the Hamilton--Jacobi equation for classical Hamiltonian
$H_0(p, q)$. Taking this into account, we can rewrite the degree one equation as
\begin{equation}
    i \frac{\partial \chi_0}{\partial t} =
    \widetilde{H}^{(0)}\left(\frac{\partial S}{\partial q}, q \right) \chi_0 +
        \sum_{k = 0}^n A_k(q) \chi^{(k)}_0 +
    H^{(1)}\left(\frac{\partial S}{\partial q}, q\right) \chi_0.
\end{equation}

Now use formulae for $\widetilde{H}^{(0)}$ and $\chi^{(k)}_0$ that we derived earlier
and we have
\begin{multline} \nonumber
    i \frac{\partial \chi_0}{\partial t} =
    -i \sum_{k = 0}^n \left(\frac{k}{2}
        \frac{\partial A_k(q)}{\partial q}
            \left(\frac{\partial S}{\partial q}\right)^{k - 1} \chi_0 +
        k A_k(q)
            \left(\frac{\partial S}{\partial q}\right)^{k - 1}
            \frac{\partial \chi_0}{\partial q} +
            \right. \\  + \left.
        \frac{k (k - 1)}{2} A_k(q)
            \left(\frac{\partial S}{\partial q}\right)^{k - 2}
                \frac{\partial^2 S}{\partial q^2} \chi_0
        \right) +
    H^{(1)}\left(\frac{\partial S}{\partial q}, q\right) \chi_0.
\end{multline}
It is easy to rearrange it to
\begin{align}\label{chi-0}
    -i H^{(1)}\left(\frac{\partial S}{\partial q}, q\right) \chi_0  &=
    \left(\frac{\partial \chi_0}{\partial t} +
        \sum_{k = 0}^n k A_k(q)
            \left(\frac{\partial S}{\partial q}\right)^{k - 1}
            \frac{\partial \chi_0}{\partial q} \right) + \\
       &+  \frac{1}{2} \sum_{k = 0}^n
            \left( k \frac{\partial A_k(q)}{\partial q}
                \left(\frac{\partial S}{\partial q}\right)^{k - 1} +
            k (k - 1) A_k(q)
                \left(\frac{\partial S}{\partial q} \right)^{k - 2}
                    \frac{\partial^2 S}{\partial q^2}
        \right) \chi_0.
\end{align}

Let $\sigma_0 = \{p(\tau, q_0), q(\tau, q_0)\}_{\tau = 0}^t$ be the solution to Hamilton's
equations with the initial condition $p(0) = f'(q_0)$ and $q(0) = q_0$
\footnote{The condition $q(t, q_0)=q$, generically, gives finitely many trajectories connecting
two Lagrangian submanifolds $L_f$ and $T^*Q$ in time $t$.}.
For $q(t) = q(t, q_0)$ we have
\begin{equation}
        \dot{q}(t) = \frac{\partial H^{(0)}(p(t), q(t))}{\partial p}.
\end{equation}
We also have $p(t) = p(t, q_0) = \frac{\partial S(q, q_0, t)}{\partial q}$
where $S(q, q_0, t)$ is
the Hamilton--Jacobi action (\ref{HJ}) evaluated on $\sigma_0$. Thus
\begin{equation} \label{ham}
   \dot{q}(t) = \frac{\partial H^{(0)}(p(t), q(t))}{\partial p} =
        \sum_{k = 0}^n  k A_k(q) \big( p(t) \big)^{k - 1}.
\end{equation}
From here, we conclude
\begin{eqnarray} \nonumber
    \frac{\partial \chi_0(q(t), t)}{\partial t} +
     \frac{\partial \chi_0 (q(t), t)}{\partial q}  \sum_{k = 0}^n
        k\, A_k(q(t))  \big( p(t) \big)^{k - 1}  =
    \frac{\partial \chi_0 (q(t), t)}{\partial t} +
    \dot{q}(t) \frac{\partial \chi_0 (q(t), t)}{\partial q} =
    \frac{d \chi_0 (q(t),t)}{dt}.
\end{eqnarray}
Now we can write the equation (\ref{chi-0}) as
\begin{multline}
   - i H^{(1)}\left(p(t), q(t)\right) \chi_0(q(t), t) =  \frac{d \chi_0(q(t), t)}{dt} + \\ +
       \frac{1}{2} \sum_{k = 0}^n
        \left( k \frac{\partial A_k(q(t))}{\partial q}
               \big( p(t) \big)^{k - 1} +
               k (k - 1) A_k(q)  \big( p(t) \big)^{k - 2} \
               \frac{\partial^2 S(q(t))}{\partial q^2}
        \right) \chi_0(q(t), t).
\end{multline}
Differentiating (\ref{ham}) in $q_0$ we obtain
\begin{eqnarray} \nonumber
    \frac{\partial \dot{q}(t,q_0)}{\partial q_0} =
        \sum_{k = 0}^n \left(
            k \frac{\partial A_k(q(t))}{\partial q} \big( p(t) \big)^{k - 1}
            +
            k (k - 1) A_k(q(t))  \big( p(t) \big)^{k - 2} \
            \frac{\partial^2 S(q(t))}{\partial q^2}
        \right)
    \frac{\partial q(t,q_0)}{\partial q_0},
\end{eqnarray}
or
\begin{equation}    \frac{d}{dt} \log \left| \frac{\partial q(t, q_0)}{\partial q_0}\right| =
        \sum_{k = 0}^n \left(
            k \frac{\partial A_k(q(t))}{\partial q}\big( p(t) \big)^{k - 1}
            +
            k (k - 1) A_k(q(t)) \big( p(t) \big)^{k - 2} \
            \frac{\partial^2 S(q(t))}{\partial q^2}
        \right).
\end{equation}
Combining all these identities, we obtain
\begin{equation}
    \frac{d \chi_0(q(t, q_0), t)}{dt} +
        \frac{1}{2} \frac{d}{dt}
            \log \left| \frac{\partial q(t, q_0)}{\partial q_0}\right| \chi_0 (q(t, q_0), t)=
        -i H^{(1)}\big(p(t, q_0), q(t, q_0) \big) \chi_0 (q(t, q_0), t).
\end{equation}
Denote $D_{q_0}(t) = \left| \frac{\partial q(t, q_0)}{\partial q_0}\right|^{-\frac{1}{2}}$
and substitute $\chi_0(q(t), t) = D_{q_0}(t) \Psi(t, q_0)$
\footnote{Here we indicate the dependence on $q_0$ since this is the initial point determining
the classical trajectory $q(t, q_0)$.},
then
\begin{equation}
    \frac{d}{dt} \Psi(t,q_0) =
        -i H^{(1)}\big(p(t), q(t) \big) \Psi(t, q_0).
\end{equation}
Since $D_{q_0}(0) = 1$ and $q(0) = q_0$, we have $\Psi(0, q_0) = \varphi(q_0)$.

Now, assume that $\sigma_0$ connects $L_f$ and $T^*_q Q$ in time $t$, i.e. that it is one
of the trajectories $\sigma_\alpha$ with $q_\alpha(t, q_0)=q$. Let $q^\alpha_0(t, q) \in L_f$
be the starting point of $\sigma_\alpha$. Denote
$D^\alpha(q, t) = \left| \frac{\partial q_0^\alpha(t, q)}{\partial q}\right|^{\frac{1}{2}}$.
For the contribution to the semiclassical asymptotic (\ref{ss-as}) from $\sigma_\alpha$ we have
\begin{equation}
    e^{\frac{iS^\alpha(q,t)}{\hbar}} D^\alpha(q, t) \Psi^\alpha(t, q_0^\alpha(t, q)).
\end{equation}
This proves the theorem.

\section{Multi-time Hamilton--Jacobi action}
\label{sec:appendix_mt_HJ}

Here we recall some basic facts on the Hamilton--Jacobi action for integrable systems on an
exact symplectic manifold.

Let $(\mathcal{M}_{2n}, \omega), \ \omega = d \alpha$ be an exact symplectic manifold,
$\sigma \colon \mathbb{R}^n \to \mathcal{M}_{2n}$ be a multi-time parametrized path
in $\mathcal{M}_{2n}$, $\mathbf{t} \mapsto x(\mathbf{t})$ and
$\gamma \colon [0, 1] \to \mathbb{R}^n, \ \tau \mapsto \gamma(\tau) \in \mathbb{R}^n$ be
a parametrized path in $\mathbb{R}^n$.

The Hamilton--Jacobi action for the multi-time evolution of an integrable system with
Hamiltonians $H_1, \ldots, H_n$ is
\begin{equation}
    \label{mt_HJ}
    S_\gamma[\sigma] =
        \int_0^1 \left(
            \sum_{a = 1}^{2n} \alpha_a \big(x(\gamma(\tau))\big)
                \frac{d x^a(\gamma(\tau))}{d \tau} -
            \sum_{i = 1}^n H_i\big(x(\gamma(\tau))\big)
                \dot{\gamma}^i(\tau)
        \right) d\tau.
\end{equation}

Let $\mathrm{Im}(\sigma \circ \gamma) \subset \mathcal{M}_{2n}, \
\mathrm{Im} \gamma \subset \mathbb{R}^n$ be images of the corresponding parametrized
paths. The action (\ref{mt_HJ}) can be written as
\begin{equation}
    S_\gamma[\sigma] =
        \int\limits_{\mathrm{Im}(\sigma \circ \gamma)} \alpha -
        \int\limits_{\mathrm{Im} \gamma} H_\sigma,
\end{equation}
where $H_\sigma(\mathbf{t}) = \sum_{i = 1}^n H_i(x(\mathbf{t})) dt^i \in
\Omega^1(\mathbb{R}^n)$.

The variational problem for (\ref{mt_HJ}) is to find paths $\sigma$ such that
\begin{equation}
    \delta_\sigma S_\gamma[\sigma] = 0
\end{equation}
for all $\gamma$. Here $\delta_\sigma$ is a variation in $\sigma$ only, for
fixed $\gamma$. It can be easily computed
\begin{multline}\nonumber
    \delta_\sigma S_\gamma[\sigma] =
        \int_0^1 \sum_{a = 1}^{2n}  \left(
            \sum_{b = 1}^{2n} \omega_{ab}\big(x(\gamma(\tau))\big)
                \frac{d x^b(\gamma(\tau))}{d \tau} -
            \sum_{i = 1}^n
                \frac{\partial H_i\big(x(\gamma(\tau))\big)}{\partial x^a}  \dot{\gamma}^i(\tau)
            \right) \delta x^a(\gamma(\tau)) d\tau + \\ +
        \sum_a \alpha_a\big(x(\gamma(\tau))\big) \delta x^a(\gamma(\tau))
            \Big|_{\tau = 0}^{\tau = 1}.
\end{multline}
The Euler--Lagrange equations for this variational problem are
\begin{equation}
    \label{mt_EL}
    \sum_b \omega_{ab}(x(\mathbf{t}))
        \frac{\partial x^b(\mathbf{t})}{\partial t_k} =
            \frac{\partial H_k(x(\mathbf{t}))}{\partial x^a}.
\end{equation}

Solutions to these equations are critical points of $S_\gamma[\sigma]$ (for the
fixed $\gamma$) if the boundary terms
\begin{equation}
    \sum_a \alpha_a(x(\mathbf{t})) \delta x^a(\mathbf{t}) -
        \sum_a \alpha_a(x(0)) \delta x^a(0)
\end{equation}
also vanish.

In the case $\mathcal{M}_{2n} = T^* Q_n$, the boundary terms are
\begin{equation}
    \sum_i p_i(\mathbf{t}) \delta q^i(\mathbf{t}) -
        \sum_i p_i(0) \delta q^i(0).
\end{equation}
If $q(\mathbf{t}) = q$ is fixed, the first term vanishes.
If $p_i(0) = \tfrac{\partial f}{\partial q^i}(q(0))$, the second term is
$-\tfrac{\partial f}{\partial q^i}(q(0)) \delta q^i(0)$. This means that the modified
action
\begin{equation}
    \label{modified_action}
    S_{\gamma, f}[\sigma] = S_\gamma[\sigma] + f(q(0))
\end{equation}
is critical on solutions of (\ref{mt_EL}) with boundary conditions $q(\mathbf{t}) = q$
and $p_i(0) = \tfrac{\partial f}{\partial q^i}(q(0))$.

Let $S^\alpha(q, \mathbf{t})$ be the critical value of the modified action
(\ref{modified_action}) on the solution $\sigma_\alpha$. It is easy to show that
\begin{itemize}
    \item $S^\alpha(q, \mathbf{t})$ does not depend on $\gamma$.
    \item If $(p^\alpha(\mathbf{t}), q)$ is the endpoint of $\sigma_\alpha$,
        \begin{equation}
            p^\alpha_i(\mathbf{t}) =
                \frac{\partial S^\alpha(q, \mathbf{t})}{\partial q^i}.
        \end{equation}
\end{itemize}

Here are some more facts on the multi-time evolutions:
\begin{itemize}
    \item Consider the space $L_x$ of all multi-time trajectories through $x \in \mathcal{M}_{2n}$,
        $L_x = \{y \in \mathcal{M}_{2n} \mid y = x(\mathbf{t})
        \text{ for some } \mathbf{t} \in \mathbb{R}^n, x(0) = x\}$. It is easy to
        see that the pullback of $\omega$ to $L_x$ is
        \begin{equation}
            \omega|_{L_x} =
                -\frac{1}{2} \sum_{\substack{k, l \\ k < l}}
                    \{H_k, H_l\} dt_k \wedge dt_l.
        \end{equation}
        It is vanishing since $\{H_k, H_l\} = 0$.
        Therefore, $L_x$ is a Lagrangian submanifold.

    \item For a generic Lagrangian submanifold $L$ the intersection
        $L \cap L_x$ is a finite collection of points. These points
        are endpoints of the multi-time trajectories connecting $x$ and
        $L$. If $x_\alpha \in L \cap L_x$, $x_\alpha = x(\mathbf{t}_\alpha)$
        for some $\mathbf{t}_\alpha \in \mathbb{R}^n$, where $x(\mathbf{t})$
        is multi-time trajectory originated at $x = x(0)$.

    \item For a generic Lagrangian submanifold
        $L' \subset \mathcal{M}_{2n}$ we will have finitely multi-time trajectories
        connecting $L'$ with $L$ in a given multi-time $\mathbf{t}$.
        The intersection points $\phi_\mathbf{t}(L') \cap L$ are the endpoints
        of these trajectories.
\end{itemize}

\section{The proof of the lemma \ref{lem: unity lemma}}
\label{sec: unity lemma}

The leading order in the semiclassical expansion $\hbar \to 0$ of the symmetric combination
of Cherednik--Dunkl operators \eqref{dunkl_quantum} is given by the same symmetric combination
of semiclassical Cherednik--Dunkl operators
\begin{equation}\label{eq: clas CD}
    D_i = p_i +  \sum_{i > j} \frac{z_i}{z_i - z_j} K_{ij} -
        \sum_{i < j} \frac{z_j}{z_j - z_i} K_{ij}
\end{equation}
where momenta $p_i$ and coordinates $z_k$ are commuting variables $[p_i, z_k] = 0$, and
transpositions $K_{ij}$ acts both on momenta and coordinates: $p_i K_{ij} = K_{ij} p_j$,
$z_i K_{ij} = K_{ij} z_j$.

The semiclassical Cherednik--Dunkl operators \eqref{eq: clas CD} satisfy the following
relations
\begin{equation}
    [D_i, D_j] = 0, \qquad
    K_{i, i + 1} D_i = D_{i + 1} K_{i, i + 1} + 1, \qquad
    [D_i, K_{j, j + 1}] = 0, \quad i \ne j, j + 1,
\end{equation}
\begin{equation}\label{eq: class coms2}
    [D_i, z_j] = - z_{\mathrm{max}(i, j)} K_{ij}, \qquad
    [D_i, z_i] = \sum_{j \ne i} z_{\mathrm{max}(i, j)} K_{ij}.
\end{equation}
One can use these identities to present any function of the semiclassical
Cherednik--Dunkl operators $f(D_1, \ldots, D_n)$ as an element of the form
\begin{eqnarray} \label{eq: group_algebra_Sn}
    f(D_1, \ldots, D_n) = \sum_{w \in S_n} f_\omega(p, z) K_\omega,
\end{eqnarray}
where $f_\omega(p, z)$ are functions of all $p_i, z_k$, and $K_\omega$ is the permutation
obtained as a product of transpositions $K_{ij}$.

Thus, to prove that quantum Hamiltonians of spin Calogero--Moser system $\widehat{H}_k$
\eqref{scCM} in the leading order in $\hbar$ are independent on permutations, it suffices
to prove that any symmetric polynomial of semiclassical Cherednik--Dunkl operators has
$f_\omega = 0$ for $\omega \ne id$ in the expansion \eqref{eq: group_algebra_Sn}. It is
equivalent to the fact that $f(D_1, \ldots, D_n)$ commutes with all coordinates $z_k$.
It is enough to prove that only for the first $n$ elementary symmetric polynomial of
$D_1, \ldots, D_n$, because any symmetric polynomial of $n$ variables can be expressed as
a polynomial of the first $n$ elementary symmetric polynomials.

\begin{lemma}
The generating function of elementary symmetric functions of the
semiclassical Cherednik--Dunkl operators
\begin{equation}\label{eq: gen fun}
    t(\lambda) = \prod_{j = 1}^n (\lambda + D_j)
\end{equation}
commute with coordinates $z_1, \ldots, z_n$
\begin{equation}
    [t(\lambda), z_k] = 0, \quad 1 \le k \le n.
\end{equation}
\end{lemma}

\begin{proof}
First, let us prove the auxiliary identities.
\begin{proposition}\label{prop: Kcancel}
The identities
\begin{equation} \label{eq: Kcancel}
    K_{l - 1, l} \left( (\lambda + D_{l - 1}) \sum_{j < l} z_l K_{jl} -
        z_l K_{l - 1, l} (\lambda + D_l) \right) K_{l - 1, l} =
            \sum_{j < l - 1} z_{l - 1} K_{j, l - 1}  \cdot (\lambda + D_l)
\end{equation}
hold for $1 < l \le n$.
\end{proposition}
\begin{proof}
Let us show first that the coefficients of the $\lambda$-term in \eqref{eq: Kcancel}
on the left-hand side and the right-hand side are equal
\begin{equation}
    K_{l - 1, l} \left(\sum_{j < l} z_l K_{jl} - z_l K_{l - 1, l}\right) K_{l - 1, l} =
    K_{l - 1, l} \left(\sum_{j < l - 1} z_l K_{jl} \right) K_{l - 1, l} =
    \sum_{j < l - 1} z_{l - 1} K_{j, l - 1}.
\end{equation}
Consider the $\lambda^0$ coefficient on the left-hand side of \eqref{eq: Kcancel}
\begin{equation}
    K_{l - 1, l} \left( D_{l - 1} \sum_{j < l} z_l K_{jl} -
        z_l K_{l - 1, l} D_l \right) K_{l - 1, l} =
    \left( (D_l K_{l - 1, l} + 1) \sum_{j < l} z_l K_{jl} -
        z_{l - 1} D_l \right) K_{l - 1, l}.
\end{equation}
The first summand in the brackets
\begin{multline} \notag
    (D_l K_{l - 1, l} + 1) \sum_{j < l} z_l K_{jl} =
        \sum_{j < l} z_l K_{jl} +
            D_l \cdot \sum_{j < l - 1} z_{l - 1} K_{j, l - 1} K_{l - 1, l} +
                D_l z_{l - 1} = \\ =
        \sum_{j < l - 1} z_l K_{jl} + z_l K_{l - 1, l} +
            D_l \cdot \sum_{j < l - 1} z_{l - 1} K_{j, l - 1} K_{l - 1, l} +
                z_{l - 1} D_l - z_l K_{l - 1, l} = \\ =
        \sum_{j < l - 1} z_l K_{jl} +
            D_l \cdot \sum_{j < l - 1} z_{l - 1} K_{j, l - 1} K_{l - 1, l} +
                z_{l - 1} D_l.
\end{multline}
Then,
\begin{equation}
    K_{l - 1, l} \left( D_{l - 1} \sum_{j < l} z_l K_{jl} -
        z_l K_{l - 1, l} D_l \right) K_{l - 1, l} =
        \sum_{j < l - 1} z_l K_{jl} K_{l - 1, l} +
            D_l \cdot \sum_{j < l - 1} z_{l - 1} K_{j, l - 1}.
\end{equation}
And, using the fact that $[D_l, z_{l - 1} K_{j, l - 1}] = [D_l, z_{l - 1}] K_{j, l - 1} =
-z_l K_{l - 1, l} K_{j, l - 1} = -z_l K_{jl} K_{l - 1, l}$ for $j < l - 1$, we can write
\begin{equation}
    K_{l - 1, l} \left( D_{l - 1} \sum_{j < l} z_l K_{jl} -
        z_l K_{l - 1, l} D_l \right) K_{l - 1, l} =
    \sum_{j < l - 1} z_{l - 1} K_{j, l - 1} \cdot D_l,
\end{equation}
so, the $\lambda^0$ coefficients on the left-hand side and the right-hand side of
\eqref{eq: Kcancel} coincide.
\end{proof}

Let us consider
\begin{equation}
    A_n = [t(\lambda), z_n].
\end{equation}
Using the relations \eqref{eq: class coms2}, one can show that
\begin{equation}
    A_n = \prod_{k = 1}^{n - 1} (\lambda + D_k) \cdot \sum_{j < n} q_n K_{jn} -
        \sum_{j < n} \prod_{k = 1}^{j - 1} (\lambda + D_k) \cdot q_n K_{jn} \cdot
            \prod_{k = j + 1}^n (\lambda + D_n)
\end{equation}
or, excluding $D_{n - 1}$ and $K_{n - 1, n}$ containing terms
\begin{multline} \notag
    A_n = \prod_{k = 1}^{n - 2} (\lambda + D_k) \cdot (\lambda + D_{n - 1}) \cdot
        \sum_{j < n} q_n K_{jn} -
        \prod_{k = 1}^{n - 2} (\lambda + D_k) \cdot q_n K_{n - 1, n} \cdot (\lambda + D_n) - \\ -
        \sum_{j < n - 1} \prod_{k = 1}^{j - 1} (\lambda + D_k) \cdot q_n K_{jn} \cdot
            \prod_{k = j + 1}^n (\lambda + D_k).
\end{multline}
Let $A_{n - 1} = K_{n - 1, n} A_n K_{n - 1, n}$. Taking into account that
$[K_{n - 1,n}, (\lambda + D_{n - 1})(\lambda + D_{n})] = 0$ and $[K_{n - 1, n}, D_k] = 0$ for
$k \ne n - 1, n$, we get
\begin{multline} \notag
    K_{n - 1, n} \left(
        \prod_{k = 1}^{n - 2} (\lambda + D_k) \cdot (\lambda + D_{n - 1}) \cdot
            \sum_{j < n} q_n K_{jn} -
            \prod_{k = 1}^{n - 2} (\lambda + D_k) \cdot q_n K_{n - 1, n} \cdot (\lambda + D_n)
        \right) K_{n - 1, n} = \\ =
    \prod_{k = 1}^{n - 2} (\lambda + D_k) \cdot K_{n - 1, n} \left(
        (\lambda + D_{n - 1}) \sum_{j < n} q_n K_{jn} - q_n K_{n - 1, n} (\lambda + D_n)
    \right) K_{n - 1, n},
\end{multline}
\begin{multline} \notag
    K_{n - 1, n} \left(
        \sum_{j < n - 1} \prod_{k = 1}^{j - 1} (\lambda + D_k) \cdot q_n K_{jn} \cdot
        \prod_{k = j + 1}^n (\lambda + D_k) \right) K_{n - 1, n} = \\ =
    \sum_{j < n - 1} \prod_{k = 1}^{j - 1} (\lambda + D_k) \cdot q_{n - 1} K_{j, n - 1} \cdot
        \prod_{k = j + 1}^n (\lambda + D_k).
\end{multline}
Applying proposition \ref{prop: Kcancel} with $l = n$, we have
\begin{multline} \notag
    A_{n - 1} = \left( \prod_{k = 1}^{n - 2} (\lambda + D_k) \cdot
        \sum_{j < n - 1} q_{n - 1} K_{j, n - 1} -
        \sum_{j < n - 1} \prod_{k = 1}^{j - 1} (\lambda + D_k) q_{n - 1} K_{j, n - 1}
            \prod_{k = j + 1}^{n - 1}(\lambda + D_k) \right) (\lambda + D_n).
\end{multline}
Repeating this step with $A_{j - 1} = K_{j - 1, j} A_j K_{j - 1, j}$ and taking into account
proposition \ref{prop: Kcancel} for $l = j$ in the end we come to
\begin{equation}
    A_2 = (q_2 K_{12} - q_2 K_{12}) \prod_{k = 3}^n (\lambda + D_k) = 0
\end{equation}
which implies all the $A_j = 0$.

So, we have proved that $t(\lambda) z_n = z_n t(\lambda)$. The generating function $t(\lambda)$
commutes with any permutation $K_{jn}$, then, conjugating by $K_{jn}$, we obtain
and using conjugations by permutations $K_{jn}$
\begin{equation}
    K_{jn} t(\lambda) z_n K_{jn} = K_{jn} z_n t(\lambda) K_{jn}
        \quad \Rightarrow \quad t(\lambda) z_j = z_j t(\lambda),
\end{equation}
which completes our proof.
\end{proof}

\section{The proof of the proposition \ref{prop:freezing point}}
\label{sec:Proof of fixing}

Proposition \ref{prop:freezing point} can be proven using the Lax matrix formula for Hamiltonians.
The Lax operator for Calogero--Moser system for $N$ particles is an $N \times N$ matrix
\begin{equation}
    L = P + M,
\end{equation}
where $P = \mathrm{diag}(p_1, \ldots, p_N)$ is the diagonal matrix and $M$ is a matrix
with $M_{ii} = 0$ and
\begin{equation}
    M_{ij} =  \frac{z_{i}}{z_{i}-z_{j}}, \qquad  \text{ for }  i\ne j.
\end{equation}

Let $S(\lambda)$ be the generating function of the Hamiltonians
\begin{equation}
    S(\lambda, p, z) = \det \left(\lambda + L\right) =
        \sum_{i = 1}^N \lambda^i H_{N - i}(p, z).
\end{equation}
The idea of the proof is to show that the differential of this generating function
\begin{equation}
    dS(\lambda, p, z) =
        \sum_{i = 1}^N
            \lambda^i \ dH_{N - i}(p, z) =
        \sum_{i = 1}^N
            \Big(F_{i}(\lambda, p, z) dp_i +
                G_i(\lambda, p, z) dq_{i} \Big)
\end{equation}
vanishes at $(0, \zeta)$,
\begin{equation}
 \zeta =  \{\zeta_1, \ldots, \zeta_n\}, \qquad \text{with }  \zeta_k = \exp \left( \tfrac{2 \pi i k}{n} \right).
\end{equation}

Let us show that we have
\begin{align}
    \label{eq:FGfixing-one}
        F_i(\lambda, 0, \zeta) &= F_j(\lambda, 0, \zeta),
            & i, j = 1, \ldots, N, \\
    \label{eq:FGfixing-two}
        G_i(\lambda, 0, \zeta) &= 0,
            & i = 1, \ldots, N.
\end{align}
The function $F_i(0, z)$ is the determinant of the matrix $M^{(i)}$ (of size
$(N - 1) \times (N - 1)$), obtained from $M$ by removing the $i$-th row and column.
\begin{equation}
    F_i (\lambda, 0, z) = \det \left( M^{(i)}(z) + \lambda \right).
\end{equation}
Note that the function $F_{i}$ is independent of $z_i$, because only the $i$-th column
and the $i$-th row of $M$ contain $z_i$. The function $F_i$ is also symmetric in the rest
variables $z_j$ (because the permutation of $z_k$ and $z_l$ is just a simultaneous
transposition of the $k$-th and $l$-th columns and $k$-th and $l$-th rows of $M$).

Therefore, all $F_i$ can be written in terms of one symmetric function depending
in $(N - 1)$ variables
\begin{equation}
    F_i(\lambda, 0, z) = f(\lambda, z_1, \ldots, \widehat{z_i}, \ldots, z_n),
\end{equation}
where $\widehat{z_i}$ means that this variable is omitted.
The matrix elements of $M$ are invariant with respect to the dilation
$z_i \mapsto \rho z_i$. Therefore
\begin{equation}
    f(\lambda, \rho x_1, \ldots, \rho x_{n - 1}) =
        f(\lambda, x_1, \ldots, x_{n - 1}).
\end{equation}
Now, observe that
\begin{equation}
    (\zeta_1, \ldots, \widehat{\zeta_j}, \ldots, \zeta_n) =
        \zeta_j (\zeta_{n - j + 1}, \ldots, \zeta_{n - 1}, \widehat{1},
            \zeta_1, \ldots, \zeta_{n - j}) =
        \zeta_j \sigma(\widehat{1}, \zeta_1, \ldots, \zeta_{n - 1}),
\end{equation}
where $\sigma$ is a cyclic permutation. Thus,
\begin{equation}
    f(\lambda, \ldots, \widehat{\zeta_j}, \ldots) =
        f(\lambda, \ldots, \widehat{\zeta_i}, \ldots), \quad
            i, j = 1, \ldots, n
\end{equation}
and we proved (\ref{eq:FGfixing-one}).

For the function $G_{i}$, we have
\begin{equation}
    G_i(\lambda, 0, z) =
        \frac{\partial}{\partial z_{i}} \det \left(M + \lambda \right).
\end{equation}
As above, it is sufficient to prove that only one of the functions $G_{i}$ is equal to zero.
Let us prove that $G_{N}(0, \zeta)$ vanishes
\begin{equation}
    G_N(\lambda, 0, z) = \mathrm{det}(M') + \mathrm{det}(M''),
\end{equation}
where $M'$ ($M''$) is the matrix $M + \lambda$ where the last row (column) is replaced by its
derivative in $z_{N}$.

Taking into account the identities
\begin{equation}
    \zeta_j = \zeta_{n - j}^{-1}, \qquad \left( \zeta_j \right)^n = 1,
\end{equation}
we have
\begin{equation}
    \det \left( M' \right)|_{z = \zeta} =
        - \det \left( M'' \right)|_{z = \zeta},
\end{equation}
i.e. $G_{N}(0, \zeta) = 0$. Together with the symmetry arguments, this implies
\eqref{eq:FGfixing-two}. This completes the proof of the proposition.

\section*{Declarations}
The authors do not have conflicts of interests and have no data to share.

\end{document}